\documentclass[12pt,english]{article}
\usepackage{mathptmx}

\usepackage[T1]{fontenc}
\usepackage[latin9]{inputenc}
\usepackage{geometry}
\usepackage[active]{srcltx}
\usepackage{bm}
\usepackage{amsmath}
\usepackage{amsthm}
\usepackage{amssymb}
\usepackage{graphicx}
\usepackage{setspace}
\usepackage{esint}
\usepackage[authoryear]{natbib}
\makeatletter

\providecommand{\tabularnewline}{\\}

\theoremstyle{definition}
\newtheorem{defn}{\protect\definitionname}
\theoremstyle{plain}
\newtheorem{lem}{\protect\lemmaname}
\theoremstyle{plain}
\newtheorem{thm}{\protect\theoremname}
\theoremstyle{plain}
\newtheorem{prop}{\protect\propositionname}

\usepackage{amsfonts}\usepackage{babel}
\providecommand{\U}[1]{\protect\rule{.1in}{.1in}}
 \usepackage{fancyheadings}
\usepackage{bm}

\usepackage{amsthm,pstricks,pst-node,pst-blur,pstricks-add}
\RequirePackage{amsmath,amsfonts,amssymb,latexsym}

\usepackage{hyperref}
\hypersetup{
	colorlinks,
	linkcolor={red!50!black},
	citecolor={blue!50!black},
	urlcolor={blue!80!black}
}
\usepackage{breakcites}

\newtheorem{assumption}{Assumption}
\newtheorem{example}{Example}

\newtheorem{lemma}{Lemma}
\newtheorem{claim}{Claim}
\usepackage{breakcites}

\makeatother

\usepackage{babel}
\providecommand{\definitionname}{Definition}
\providecommand{\lemmaname}{Lemma}
\providecommand{\propositionname}{Proposition}
\providecommand{\theoremname}{Theorem}

\begin{document}
\title{Asymptotic Behavior of Bayesian Learners with Misspecified Models
\thanks{We thank Drew Fudenberg, Ryota Iijima, Yuhta Ishii, Dale Stahl, Philipp
Strack, and several seminar participants for helpful comments. Esponda:
Department of Economics, UC Santa Barbara, 2127 North Hall, Santa
Barbara, CA 93106, iesponda@ucsb.edu; Pouzo: Department of Economics,
UC Berkeley, 530-1 Evans Hall \#3880, Berkeley, CA 94720, dpouzo@econ.berkeley.edu.
Yamamoto: 2-1 Naka, Kunitachi, Tokyo, 186-8603 Japan, yyamamoto@ier.hit-u.ac.jp.} \bigskip{}
}
\author{%
\begin{tabular}{ccc}
Ignacio Esponda~~ & ~Demian Pouzo & ~~Yuichi Yamamoto\tabularnewline
(UC Santa Barbara)~  & ~(UC Berkeley) & (Hitotsubashi Univ.)\tabularnewline
\end{tabular}}
\maketitle
\begin{abstract}
We consider an agent who represents uncertainty about the environment
via a possibly misspecified model. Each period, the agent takes an
action, observes a consequence, and uses Bayes' rule to update her
belief about the environment. This framework has become increasingly
popular in economics to study behavior driven by incorrect or biased
beliefs. Current literature has characterized asymptotic behavior
under fairly specific assumptions. By first showing that the key element
to predict the agent's behavior is the frequency of her past actions,
we are able to characterize asymptotic behavior in general settings
in terms of the solutions of a generalization of a differential equation
that describes the evolution of the frequency of actions. We then
present a series of implications that can be readily applied to economic
applications, thus providing off-the-shelf tools that can be used
to characterize behavior under misspecified learning.
\end{abstract}
\bigskip{}

\thispagestyle{empty}

\newpage{}

\tableofcontents\thispagestyle{empty}\addtocontents{toc}{\protect\setcounter{tocdepth}{1}}\newpage{}

\setcounter{page}{1}

\section{Introduction}

Over the last few decades, evidence of systematic mistakes and biases
in beliefs has been collected in a large range of economic environments.
Moreover, the evidence indicates that many of these mistakes persist
with experience.\footnote{For discussions of the evidence, see, for example, \citet{camerer1997process}
and Section 3.D in \citet{rabin1998psychology}.} One approach to incorporating these findings in our theories is to
simply postulate that economic agents have fixed, wrong beliefs about
aspects of their environment, and never learn about these aspects.
A different approach that has gained popularity over the last few
years is to postulate that agents do learn about their environment,
but they do so in the context of a misspecified model that misses
some important aspects of reality. The idea is that the world is complex
and it is natural for economic agents to represent uncertainty about
the world with parsimonious models that are likely to be misspecified.
The researcher who follows this approach is forced to specify the
agent's misspecification, and the direction of biases is often not
ex-ante obvious without further analysis.

Examples of misspecified learning in economics date back to the 1970s
and include the following: A firm estimates a demand model but wrongly
excludes competitors' prices (\citet{Arrow-Green}, \citet{kirman75learning});
a teacher assesses how praise and criticism affect student performance,
but does not understand regression to the mean (\citet{KahnemanTversky1973},
\citet{esponda2016berk}); a person faces an increasing marginal income
tax rate but behaves as if facing a constant marginal tax (\citet{sobel1984non},
\citet{liebman2004schmeduling}, \citet{esponda2016berk}); when learning
the value of assets, policies, or investment projects, traders, voters,
and investors fail to account for sample selection (\citet{esponda2008behavioral},
Esponda and Pouzo (2017, 2019a)\nocite{esponda2017conditional}\nocite{esponda2019retrospective},
\citet{jehiel2018investment}); a seller estimates a constant-elasticity
demand function, but elasticity is not constant (\citet{nyarko1991learning},
\citet*{fudenberg2017active}); a person inverts causal relationships
and incorrectly believes that diet affects a chemical in the blood
which in turn affects health (\citet{spiegler2016bayesian}); overconfidence
biases an agent's learning of a fundamental (\citet*{heidhues2018unrealistic}).

In all of these examples, the agent processes information through
the lens of a simple model that misses some aspect of reality. The
main question in the literature is what happens to the agent\textquoteright s
behavior as time goes by and she uses feedback to update her belief
about the model's primitives. The direction of the bias is often not
obvious because the agent\textquoteright s behavior affects the feedback
she observes, this feedback is in turn processed via the agent\textquoteright s
misspecified model, and this processing leads to updated beliefs and
subsequent changes in behavior, which in turn lead to changes in beliefs,
and so on.

Despite these examples, we have not yet fully understood how model
misspecification affects the long-run learning outcome. Indeed, most
of the existing papers consider somewhat specialized setups, and we
do not know whether the learning process converges beyond these particular
cases. This paper develops a unified theory on Bayesian learning with
model misspecification, which hopefully shapes our understanding of
why different models in the literature lead to different conclusions
and allows us to characterize behavior in a much wider range of settings.

We consider the following environment, which includes many situations
of interest, including the examples described above. Time is discrete
and there is a single, infinitely-lived agent who discounts the future
and must take an action in each period. The agent's action potentially
affects the distribution of an observable variable, which we call
a consequence. Her per-period payoff depends on the agent's action
and the realized consequence. The true distribution over consequences
as a function of an action $x\in X$ is given by $Q(\cdot\mid x)\in\Delta(Y)$,
where $Y$ is the set of consequences. The agent, however, does not
know $Q$. She has a parametric model of it, given by $(Q_{\theta}(\cdot\mid x))_{x\in X}$,
where parameter values, such as $\theta$, belong to a parameter space
$\Theta$. The agent is Bayesian, so she has a prior over $\Theta$
and updates her prior in each period after observing the realized
consequence. The agent's model is misspecified if the support of her
prior does not include the true distribution $Q$, and it is correctly
specified otherwise.\footnote{The correctly-specified version of this environment was originally
studied by \citet{easley1988controlling} and \citet*{aghion1991optimal}.}

Our key point of departure from previous literature is that we begin
by focusing on the evolution of the \emph{frequency} of actions rather
than on actions alone or on the agent's belief. The frequency of actions
at time $t+1$ can be written recursively as a function of the frequency
at time $t$ plus some innovation term that depends on the agent's
action at time $t+1$. The action at time $t+1$, however, depends
on the agent's belief at time $t$, and one challenge is to be able
to write this belief as a function of frequencies of actions so as
to make this recursion depend exclusively on frequencies, not beliefs.

Extending results by \citet{berk1966limiting} and \citet{esponda2016berk},
we show that eventually the posterior at time $t$ roughly concentrates
on the set of parameter values that minimize Kullback-Leibler divergence
given the frequency of actions up to time $t$. This result allows
us to write the evolution of frequencies of actions recursively as
a function of the past frequency alone, excluding the belief. We then
apply techniques from stochastic approximation developed by \citet*{benaim2005stochastic}
to show that the continuous-time approximation of the frequency of
actions can be essentially characterized as a solution to a generalization
of a differential equation.\footnote{The type of differential equation is called a differential inclusion
in the literature. It differs from a differential equation in that
there may be multiple derivatives at certain points and therefore
multiple trajectories that solve the equation. Multiplicity arises
in our environment because there are certain beliefs at which the
agent is indifferent between different actions, and we need to keep
track of what would happen to beliefs and subsequent actions if the
agent were to follow any one of these actions.} Finally, we present a series of implications that can be readily
applied to economic applications, thus providing off-the-shelf tools
that can be valuable to people working in this area. In fact, for
the special case of one-dimensional models\textendash a case that
includes most of the applications in the literature\textendash our
results imply that one can essentially characterize convergence and
stability by looking at a simple two-dimensional figure.

Our results pertain to the agent's long-run behavior, and there are
at least three reasons why the focus on long-run behavior is important.
First, there are many instances where it is not surprising that people
initially make incorrect decisions and the more interesting question
is what types of biases persist with experience. Second, systematic
patterns tend to arise as time goes by, while initial behavior tends
to be more dependent on random draws. Finally, there is a long tradition
in statistic and economics focusing on asymptotic or equilibrium behavior,
and so we can use existing tools as well as compare our results to
existing results in these literatures.

Our work benefits from previous work on misspecified learning. Our
environment is the single-agent version of the environment studied
by \citet{esponda2016berk}. They introduce the notion of a Berk-Nash
equilibrium and show that, under some conditions, if behavior converges
then it must converge to a Berk-Nash equilibrium.\footnote{There are many examples of boundedly-rational equilibrium concepts
that abstract away from the question of dynamics and convergence,
including \citep{jehiel2005analogy,jehiel1995limited}, \citet{osborne1998games},
\citet{eyster2005cursed}, \citet{esponda2008behavioral}, \citet{jehiel2008revisiting},
and \citep{spiegler2016bayesian,spiegler2017data}.} But they do not study convergence in general.\footnote{\citet{esponda2016berk} tackle the issue of convergence in Theorem
3, where they use an idea from \citet{fudenberg1993learning} to show
that, if agents are allowed to make possibly large but vanishing mistakes,
then behavior can converge to any equilibrium. Here, as in the rest
of the literature on misspecified learning, we consider the case where
agents don't make these types of mistakes.}

Papers that study convergence in misspecified settings are few and
ingeniously establish results, though for somewhat specialized setups.
\citet{nyarko1991learning} presents an example where the agent's
action does not converge. \citet*{fudenberg2017active} consider a
more general model where the agent has a finite number of actions
but still updates between two possible models (i.e., $\Theta$ has
two elements). They provide a full characterization of asymptotic
actions and beliefs, including cases where the action converges and
cases where it does not. Their model is in continuous time and they
exploit the fact that the belief over $\Theta$ follows a one-dimensional
stochastic differential equation. \citet*{heidhues2018unrealistic}
study a model of an agent whose overconfidence biases his learning
of a fundamental that is relevant for determining the optimal action.
They are able to establish convergence by exploiting the monotone
structure of their environment. \citet*{heidhues2018convergence}
consider an environment where action spaces are continuous, the state
has a unidirectional effect on output, and the prior and noise are
normal. These assumptions imply that the posterior admits a one-dimensional
summary statistic, to which they apply tools from stochastic approximation
theory to establish convergence.\footnote{For another example using normality assumptions and stochastic approximation,
see the online appendix of \citet{esponda2016berk}.} As mentioned above, we are able to make significantly more progress
by focusing on the frequency of actions, as opposed to the action
itself or the belief.\footnote{Incidentally, we show that in some of the examples of \citet{nyarko1991learning}
and \citet*{fudenberg2017active} where the action diverges, the action
\emph{frequency} converges. This is the first result of its kind and
it provides a new interpretation of a mixed-action steady state. We
also present examples where not even the action frequency converges.}

Tools from stochastic approximation have been previously applied in
economics, including the literature on learning in games (e.g., \citet{fudenberg1993learning},
\citet{benaim1999mixed}, and \citet{hofbauer2002global}) and learning
in macroeconomics (e.g., \citet{sargent1993bounded}). Our approach
is inspired by \citet{fudenberg1993learning}'s model of stochastic
fictitious play. In that environment, the frequency of past actions
exactly represents the agents' beliefs about other agents' strategies.
In our environment, we characterize beliefs to be a function of the
frequency of actions.

Misspecified learning has also been studied in other environments.
\citet{rabin2010gambler} study a case where shocks are i.i.d. but
agents believe them to be autoregressive. \citet{esponda2017equilibrium}
extend Berk-Nash equilibrium to Markov decision problems, where a
state variable, other than a belief, affects continuation values.
\citet{he2018mislearning} considers agents suffering from the gambler's
fallacy who mislearn from endogenously censored data. \citet{molavi2018macroeconomics}
studies a general-equilibrium framework that nests a class of macroeconomic
models where agents learn with misspecified models. \citet{bohren2018social}
and \citet*{frick2019misinterpreting} characterize asymptotic behavior
in social learning environments with model misspecification.\footnote{See also \citet{eyster2010naive}, \citet{bohren2016informational},
and \citet{gagnon2017naive}.} Finally, \citet*{frick2019stability} focus on convergence and robustness
of the stability of equilibrium in both single-agent and social learning
environments. We do not study robustness but study asymptotic properties
of the learning process in a more general way. In particular, we develop
general tools to study whether behavior converges or not, and what
happens if it does not converge. We believe our results can be extended
to these other environments.

Finally, we take the misspecification as given and establish results
for all types of misspecifications. For work that could help understand
which types of misspecifications are more likely to arise, see, e.g.,
\citet{aragones2005fact}, \citet{al-najjar2009decision}, \citet{al-najjar2013coarse},
\citet{schwartzstein2014selective}, and \citet{olea2019competing}.

We present the model in Section \ref{sec:environment}, characterize
asymptotic beliefs in Section \ref{sec:asymptotic-beliefs} and asymptotic
behavior in Section \ref{sec:Characterization-of-asymptotic}, and
then present implications relevant to economic applications in Section
\ref{sec:Convergence-to-equilibrium}. We relate our findings to the
notion of a Berk-Nash equilibrium in Section \ref{sec:berk-nash}.

\section{\label{sec:environment}The environment}

\emph{Objective environment. }There is a single agent facing the following
infinitely repeated problem. Each period $t=1,2,...$, the agent must
choose an action from a finite set $X$. She then receives a consequence
according to the consequence function $Q:X\rightarrow\Delta Y$, where
$Y$ is the set of consequences and $\Delta Y$ is the set of all
(Borel) probability measures over it. Finally, the payoff function
$\pi:X\times Y\rightarrow\mathbb{R}$ determines the agent's current
payoff. In particular, if $x_{t}\in X$ is the agent's choice at time
$t$, then $y_{t}\in Y$ is drawn according to the probability measure
$Q(\cdot\mid x_{t})\in\Delta Y$, and the agent's payoff at time $t$
is $\pi(x_{t},y_{t})$.\medskip{}

\begin{assumption}\label{ass:1} (i) $Y$ is a compact subset of
Euclidean space; (ii) There exists a Borel probability measure $\nu\in\Delta Y$
such that, for all $x\in X$, $Q(\cdot|x)\ll\nu$, i.e., $Q(\cdot|x)$
is absolutely continuous with respect to $\nu$ (an implication is
the existence of densities $q(\cdot\mid x)\in L^{1}(Y,\mathbb{R},\nu)$
such that $\int_{A}q(y\mid x)\nu(dy)=Q(A|x)$ for any $A\subseteq Y$
Borel); (iii) For all $x\in X$, $\pi(x,\cdot)\in L^{1}(Y,\mathbb{R},Q(\cdot\mid x))$.\footnote{As usual, $L^{p}(Y,\mathbb{R},\nu)$ denotes the space of all functions
$f:Y\rightarrow\mathbb{R}$ such that $\int\left|f(y)\right|^{p}\nu(dy)<\infty$.}\end{assumption}\medskip{}

Assumption \ref{ass:1} collects some standard technical conditions.
It includes both the case where the consequence is a continuous variable
( $\nu$ is the Lebesgue measure and $q(\cdot\mid x)$ is the density
function) and the case where it is discrete ( $\nu$ is the counting
measure and $q(\cdot\mid x)$ is the probability mass function).

In the special case in which the agent knows the primitives and wishes
to maximize discounted expected utility, she chooses an action in
each period from the set of actions that maximizes 
\[
\int_{Y}\pi(x,y)Q(dy\mid x)=\int_{Y}\pi(x,y)q(y|x)\nu(dy).
\]
 We will study the case where the agent does not know the consequence
function $Q$.\medskip{}

\emph{Subjective family of models.} The agent is endowed with a parametric
family of consequence functions, $\mathcal{Q}_{\Theta}=\{Q_{\theta}:\theta\in\Theta\}$,
where each $Q_{\theta}:X\rightarrow\Delta Y$ is indexed by a \emph{model}
$\theta\in\Theta$. We refer to $\mathcal{Q}_{\Theta}$ as the family
of models and say that it is \emph{correctly specified} if $Q\in\mathcal{Q}_{\Theta}$
and \emph{misspecified} otherwise.

\medskip{}

\begin{assumption}\label{ass:2}\textbf{ }(i) For all $\theta\in\Theta$
and $x\in X$, $Q_{\theta}(\cdot|x)\ll\nu$, where $\nu$ is defined
in A1 (an implication is the existence of densities $q_{\theta}(\cdot\mid x)\in L^{1}(Y,\mathbb{R},\nu)$
such that $\int_{A}q_{\theta}(y\mid x)\nu(dy)=Q_{\theta}(A|x)$ for
any $A\subseteq Y$ Borel); (ii) $\Theta$ is a compact subset of
an Euclidean space and, for all $x\in X,$ $\theta\mapsto q_{\theta}(\cdot\mid x)$
is continuous $Q(\cdot\mid x)$-a.s.; ; (iii) For all $x\in X$, there
exists $g_{x}\in L^{2}(Y,\mathbb{R},Q(\cdot\mid x))$ such that, for
all $\theta\in\Theta$, $\left|\ln(q(\cdot\mid x)/q_{\theta}(\cdot\mid x))\right|\leq g_{x}(\cdot)$
a.s.-$Q(\cdot\mid x)$.\end{assumption}\medskip{}

Assumption \ref{ass:2}(i) guarantees the existence of a density function,
and \ref{ass:2}(ii) is a standard parametric assumption on the subjective
model. Assumption \ref{ass:2}(iii) will be used to establish a uniform
law of large numbers. This condition also implies that, for all $\theta$
and $x$, the support of $Q_{\theta}(\cdot\mid x)$ contains the support
of $Q(\cdot\mid x)$; in particular, every observation can be generated
by the agent's model.

\medskip{}

\emph{Bayesian learning}. The agent is Bayesian and starts with a
prior $\mu_{0}$ over the space of models $\Theta$. She observes
past actions and consequences and uses this information to update
her belief about $\Theta$ in every period. The timing is as follows:
At each time $t$, the agent holds some belief $\mu_{t}$. Given $\mu_{t}$,
she chooses an action $x_{t}$. Then the consequence $y_{t}$ is drawn
according to $Q(\cdot\mid x_{t})$. The agent observes $y_{t}$, receives
an immediate payoff of $\pi(x_{t},y_{t})$, and updates her belief
to $\mu_{t+1}=B(x_{t},y_{t},\mu_{t})$, where $B$ is the Bayesian
operator.\footnote{The Bayesian operator $B:X\times Y\times\Delta\Theta\rightarrow\Delta\Theta$
satisfies, for all $A\subseteq\Theta$ Borel, for any $x\in X$, and
a.s.-$Q(\cdot\mid x)$, $B(x,y,\mu)(A)=\int_{A}q_{\theta}(y\mid x)\mu(d\theta)/\int_{\Theta}q_{\theta}(y\mid x)\mu(d\theta)$. } The next assumption guarantees that the prior has full support.\medskip{}

\begin{assumption}\label{ass:3}$\mu_{0}(A)>0$ for any $A$ open
and non-empty.\end{assumption}

\medskip{}

\emph{Policy and probability distribution over histories}. A \emph{policy}
$f$ is a function $f:\Delta\Theta\rightarrow X$ specifying the action
$f(\mu)\in X$ that the agent takes at any moment in time in which
her belief is $\mu$.\footnote{We do not allow the agent to mix to simplify the exposition and to
highlight the fact that a mixed distribution over actions may describe
limiting behavior despite the fact that the agent never actually mixes.
In the more general case where $f$ maps into $\Delta X$, our main
result (Theorem \ref{Theo:APT}) holds exactly as stated but some
of the statements in Section \ref{sec:Convergence-to-equilibrium}
need to be modified accordingly.} A \emph{history} is a sequence $h=(x_{0},y_{0},...,x_{t},y_{t},...)\in\mathbb{H}\equiv(X\times Y)^{\infty}$.
Together with the primitives of the problem, a policy $f$ induces
a probability distribution over the set of histories, which we will
denote by $P^{f}$.

\medskip{}

\emph{Policy correspondence}. It will be convenient to characterize
behavior for a family of policies, and not just for a single policy
function. For this purpose, we define a \emph{policy correspondence
}to be a mapping $F:\Delta\Theta\rightrightarrows X$, where $F(\mu)\subseteq X$
denotes the set of actions that the agent might choose any time her
belief is $\mu\in\Delta\Theta$. We sometimes abuse notation and,
for a set of probability measures $A\subseteq\Delta\Theta$, we let
$F(A)$ represent the set of actions $x$ such that $x\in F(\mu)$
for some $\mu\in A$. Let $Sel(F)$ denote the set of all policies
$f$ that constitute a selection from the correspondence $F$, i.e.,
with the property that $f(\mu)\in F(\mu)$ for all $\mu$.\medskip{}

\begin{assumption}\label{ass:4} The policy correspondence $F$ is
upper hemi-continuous (uhc).\end{assumption}

\medskip{}

An important special case is one where the agent maximizes discounted
expected utility with discount factor $\beta\in[0,1)$. This problem
can be cast recursively as
\begin{equation}
W(\mu)=\max_{x\in X}\int_{Y}\left\{ \pi(x,y)+\beta W(\mu')\right\} \bar{Q}_{\mu}(dy|x)\label{eq:BellmanWopt-1-1}
\end{equation}
where $W:\Delta\Theta\rightarrow\mathbb{R}$ is the (unique) solution
to the Bellman equation (\ref{eq:BellmanWopt-1-1}), $\mu'=B(x,y,\mu)$
is the Bayesian posterior, and $\bar{Q}_{\mu}\equiv\int_{\Theta}Q_{\theta}\mu(d\theta)$.
In this case, it is well known that the correspondence mapping beliefs
to optimal actions is uhc.\medskip{}

\emph{Action frequency}. Our main objective is to study regularities
in asymptotic behavior. Previous work has focused on characterizing
the limit of the sequence of actions, whenever it exists. But there
are cases where actions do not converge (e.g., \citet{nyarko1991learning}),
and in those cases previous work has not much else to say about asymptotic
behavior. We make progress by studying the action \emph{frequency}.
We do so for two reasons. First, from a practical perspective, even
if actions do not converge, it is possible for the frequency of actions
to converge. Thus, studying frequencies can help uncover additional
regularities in behavior, with important implications regarding, for
example, limiting average payoffs. Second, as we will show, asymptotic
beliefs depend crucially on the action frequency. Because actions
in turn depend on beliefs, future actions depend crucially on the
frequency of past actions.

For every $t$, we define the \emph{action frequency at time $t$}
to be a function $\sigma_{t}:\mathbb{H}\rightarrow\Delta X$ defined
such that, for all $h\in\mathbb{H}$ and $x\in X$, 
\[
\sigma_{t}(h)(x)=\frac{1}{t}\sum_{\tau=1}^{t}\mathbf{1}_{(x)}(x_{\tau}(h))
\]
is the fraction of times that action $x$ occurs in history $h$ by
time period $t$.

\section{\label{sec:asymptotic-beliefs}Asymptotic characterization of beliefs}

In this section, we take as given the sequence of action frequencies,
$(\sigma_{t})_{t}$, and we characterize the agent's asymptotic beliefs.
In subsequent sections, we will use the characterization of beliefs
to characterize the sequence $(\sigma_{t})_{t}$, which is ultimately
an endogenous object. The key object in our characterization is the
notion of Kullback-Leibler divergence.\footnote{Formally, what we call KLD is the Kullback-Leibler divergence between
the distributions $q\cdot\sigma$ and $q_{\theta}\cdot\sigma$ defined
over the space $X\times Y$.}\medskip{}

\begin{defn}
The \textbf{Kullback-Leibler divergence} (KLD) is a function $K\colon\Theta\times\Delta X\rightarrow\mathbb{R}$
such that, for any $\theta\in\Theta$ and $\sigma\in\Delta X$,
\begin{align*}
K(\theta,\sigma) & =\sum_{x\in X}E_{Q(\cdot\mid x)}\left[\ln\frac{q(Y\mid x)}{q_{\theta}(Y\mid x)}\right]\sigma(x)\\
 & =\sum_{x\in X}\int_{Y}\ln\frac{q(y\mid x)}{q_{\theta}(y\mid x)}q(y\mid x)\nu(dy)\sigma(x).
\end{align*}

The \textbf{set of closest models given }$\sigma$ is the set $\Theta(\sigma)\equiv\arg\min_{\theta\in\Theta}K(\theta,\sigma)$
and the \textbf{minimized KLD given $\sigma$} is $K^{*}(\sigma)\equiv\min_{\theta\in\Theta}K(\theta,\sigma)$.
\end{defn}
\medskip{}

\begin{lem}
\label{Lemma:Theta(sigma)}(i) $(\theta,\sigma)\mapsto K(\theta,\sigma)-K^{*}(\sigma)$
is continuous; (ii) $\Theta(\cdot)$ is uhc, nonempty-, and compact-valued.
\end{lem}
\begin{proof}
See Appendix \ref{pf:Lemma:Theta(sigma)}.
\end{proof}
\medskip{}

If the actions were drawn from an i.i.d. distribution $\sigma\in\Delta X$,
we could directly apply Berk's (1966)\nocite{berk1966limiting} result
to conclude that the posterior eventually concentrates on the set
of closest models given $\sigma$ (i.e., for all open sets $U\supseteq\Theta(\sigma)$,
$\lim_{t\rightarrow\infty}\mu_{t}(U)=1$ ~~$P^{f}$-a.s.).\footnote{See also \citet{bunke1998asymptotic}. Relatedly, \citet{white1982maximum}
shows that the Kullback-Leibler divergence characterizes the limiting
behavior of the maximum quasi-likelihood estimator.} EP2016 showed that Berk's conclusion extends even if actions are
not i.i.d., provided that the distribution over actions at time $t$
converges to a distribution $\sigma$. This type of result is useful
to characterize behavior under the assumption that it stabilizes,
but it is insufficient to determine whether or not behavior stabilizes.

In the current section, we provide a characterization of beliefs that
does not rely on the assumption that behavior stabilizes. Roughly
speaking, we will show that the distance between the agent's belief
at time $t$, $\mu_{t}$, and the set of probability measures with
support in $\Theta(\sigma_{t})$ goes to zero as time goes to infinity,
irrespective of whether or not $(\sigma_{t})_{t}$ converges. We will
establish this result in several steps, which we now discuss informally
and then address formally in the proofs. First, we note that for any
Borel set $A\subseteq\Theta$, the posterior belief over $A$ can
be written as
\begin{align}
\mu_{t+1}(A) & =\frac{\int_{A}\prod_{\tau=1}^{t}q_{\theta}(y_{\tau}\mid x_{\tau})\mu_{0}(d\theta)}{\int_{\Theta}\prod_{\tau=1}^{t}q_{\theta}(y_{\tau}\mid x_{\tau})\mu_{0}(d\theta)}\nonumber \\
 & =\frac{\int_{A}e^{-tL_{t}(\theta)}\mu_{0}(d\theta)}{\int_{\Theta}e^{-tL_{t}(\theta)}\mu_{0}(d\theta)},\label{eq:posterior}
\end{align}
where $L_{t}(\theta)\equiv t^{-1}\sum_{\tau=1}^{t}\ln\frac{q(y_{\tau}\mid x_{\tau})}{q_{\theta}(y_{\tau}\mid x_{\tau})}$
is the sample average of the log-likelihood ratios, and where we have
omitted the history for simplicity. Naturally, we might expect the
sample average to converge to its expectation for each $\theta$.
The next result strengthens this intuition and establishes that the
difference between $L_{t}(\cdot)$ and $K(\cdot,\sigma_{t})$ converges
\emph{uniformly} to zero as $t\rightarrow\infty$.

\medskip{}

\begin{lem}
\label{Lemma:uniformconvergence}Under Assumptions \ref{ass:1}-\ref{ass:2},
for any policy $f$, $\lim_{t\rightarrow\infty}\sup_{\theta\in\Theta}\left|L_{t}(\theta)-K(\theta,\sigma_{t})\right|=0$
$P^{f}$\emph{-a.s.}
\end{lem}
\begin{proof}
See Appendix \ref{pf:Lemma:uniformconvergence}.
\end{proof}
\medskip{}

The next step is to replace $L_{t}(\cdot)$ in (\ref{eq:posterior})
with $K(\cdot,\sigma_{t})$. By Lemma \ref{Lemma:uniformconvergence},
for sufficiently large $t$, we obtain 
\begin{equation}
\mu_{t+1}(A)\approx\frac{\int_{A}e^{-tK(\theta,\sigma_{t})}\mu_{0}(d\theta)}{\int_{\Theta}e^{-tK(\theta,\sigma_{t})}\mu_{0}(d\theta)}.\label{eq:posterior_approx}
\end{equation}
As $t\rightarrow\infty$, the posterior concentrates on models where
$K(\theta,\sigma_{t})$ is close to its minimized value, $K^{*}(\sigma_{t})$.
This statement is seen most easily for the case where $\Theta$ has
only two elements, $\theta_{1}$ and $\theta_{2}$. In this case,
(\ref{eq:posterior_approx}) becomes 
\begin{equation}
\mu_{t+1}(\theta_{1})\approx1/(1+\frac{\mu_{0}(\theta_{2})e^{-tK(\theta_{2},\sigma_{t})}}{\mu_{0}(\theta_{1})e^{-tK(\theta_{1},\sigma_{t})}}).\label{eq:posterior_approx2}
\end{equation}
Suppose, for example, that $(\sigma_{t})_{t}$ converges to $\sigma$
and that KLD is minimized at $\theta_{1}$ given $\sigma$. Then there
exists $\varepsilon>0$ such that, for all sufficiently large $t$,
$K(\theta_{2},\sigma_{t})-K(\theta_{1},\sigma_{t})>\varepsilon$.
It follows from (\ref{eq:posterior_approx2}) that $\mu_{t+1}(\theta_{1})$
converges to 1, so the posterior concentrates on the model that minimizes
KLD given $\sigma$. When $(\sigma_{t})_{t}$ does not converge, however,
we have to account for the possibility that $K(\theta_{2},\sigma_{t})-K(\theta_{1},\sigma_{t})>0$
for all $t$ but $K(\theta_{2},\sigma_{t})-K(\theta_{1},\sigma_{t})\rightarrow0$
as $t\rightarrow0$. In this case, we cannot say that the posterior
eventually puts probability 1 on $\theta_{1}$, even though $\theta_{1}$
always minimizes KLD. This is why the next result says that the posterior
concentrates on models where $K(\theta,\sigma_{t})$ is close to its
minimized value, $K^{*}(\sigma_{t})$, as opposed to saying that the
posterior asymptotically concentrates on the minimizers of KLD given
$\sigma_{t}$.\footnote{Formally, what we are saying is that it is not generally true that
$\lim_{t\rightarrow\infty}\int_{\Theta}\inf_{\theta'\in\Theta(\sigma_{t})}\left\Vert \theta-\theta'\right\Vert \mu_{t+1}(d\theta)=0$.
This type of statement is true in Berk's iid setup and, as the previous
discussion suggests, it is also true in our environment under the
additional assumption that $(\sigma_{t})_{t}$ converges.} We now state the result formally and provide a proof.\medskip{}

\begin{thm}
\label{Theo:Berk}Under Assumptions \ref{ass:1}-\ref{ass:3}, for
any policy $f$,
\begin{equation}
\lim_{t\rightarrow\infty}\int_{\Theta}(K(\theta,\sigma_{t})-K^{*}(\sigma_{t}))\mu_{t+1}(d\theta)=0\,\,\,\,\,\,\,\,\,\text{ \ensuremath{P^{f}}-a.s.}\label{eq:eqTheoBerk}
\end{equation}
\end{thm}
\begin{proof}
Fix a history $h$ such that the condition of uniform convergence
in Lemma \ref{Lemma:uniformconvergence} holds, and note that the
set of histories with this property has probability one (henceforth,
we omit the history from the notation). In particular, for all $\eta>0$,
there exists $t_{\eta}$ such that, for all $t\geq t_{\eta}$,
\begin{equation}
\left|L_{t}(\theta)-K(\theta,\sigma_{t})\right|<\eta\label{eq:L_cerca_K}
\end{equation}
 for all $\theta\in\Theta$.

Let $\bar{K}(\theta,\sigma)\equiv K(\theta,\sigma)-K^{*}(\sigma)$.
Fix any $\varepsilon>0$. Using (\ref{eq:posterior}) and the facts
that $0\leq K^{*}(\sigma)$ (the proof is standard) and $K^{*}(\sigma)<\infty$
(follows from Assumption \ref{ass:2}(iii)) for all $\sigma$, we
obtain
\begin{align*}
\int\bar{K}(\theta,\sigma_{t})\mu_{t+1}(d\theta) & =\frac{\int_{\Theta}\bar{K}(\theta,\sigma_{t})e^{-tL_{t}(\theta)}\mu_{0}(d\theta)}{\int_{\Theta}e^{-tL_{t}(\theta)}\mu_{0}(d\theta)}\\
 & =\frac{\int_{\Theta}\bar{K}(\theta,\sigma_{t})e^{-t(L_{t}(\theta)-K^{*}(\sigma_{t}))}\mu_{0}(d\theta)}{\int_{\Theta}e^{-t(L_{t}(\theta)-K^{*}(\sigma_{t}))}\mu_{0}(d\theta)}\\
 & \leq\varepsilon+\frac{\int_{\{\theta:\bar{K}(\theta,\sigma_{t})\geq\varepsilon\}}\bar{K}(\theta,\sigma_{t})e^{-t(L_{t}(\theta)-K^{*}(\sigma_{t}))}\mu_{0}(d\theta)}{\int_{\{\theta:\bar{K}(\theta,\sigma_{t})\leq\varepsilon/2\}}e^{-t(L_{t}(\theta)-K^{*}(\sigma_{t}))}\mu_{0}(d\theta)}\\
 & =:\varepsilon+\frac{A_{t}^{\varepsilon}}{B_{t}^{\varepsilon}}.
\end{align*}
The proof concludes by showing that $\lim_{t\rightarrow\infty}A_{t}^{\varepsilon}/B_{t}^{\varepsilon}=0$.

By (\ref{eq:L_cerca_K}), there exists $t_{\eta}$ such that, for
all $t\geq t_{\eta}$, 
\begin{align*}
\frac{A_{t}^{\varepsilon}}{B_{t}^{\varepsilon}} & \leq\frac{\int_{\{\theta:\bar{K}(\theta,\sigma_{t})\geq\varepsilon\}}\bar{K}(\theta,\sigma_{t})e^{-t(\bar{K}(\theta,\sigma_{t})-\eta)}\mu_{0}(d\theta)}{\int_{\{\theta:\bar{K}(\theta,\sigma_{t})\leq\varepsilon/2\}}e^{-t(\bar{K}(\theta,\sigma_{t})+\eta)}\mu_{0}(d\theta)}\\
 & =e^{2t\eta}\frac{\int_{\{\theta:\bar{K}(\theta,\sigma_{t})\geq\varepsilon\}}\bar{K}(\theta,\sigma_{t})e^{-t\bar{K}(\theta,\sigma_{t})}\mu_{0}(d\theta)}{\int_{\{\theta:\bar{K}(\theta,\sigma_{t})\leq\varepsilon/2\}}e^{-t\bar{K}(\theta,\sigma_{t})}\mu_{0}(d\theta)}.
\end{align*}
Observe that the function $x\mapsto x\exp\{-tx\}$ is decreasing for
all $x>1/t$. Thus, for any $t\geq\max\{t_{\eta},1/\varepsilon\}$
it follows that $\bar{K}(\theta,\sigma_{t})e^{-t\bar{K}(\theta,\sigma_{t})}\leq\varepsilon e^{-t\varepsilon}$
over $\left\{ \theta\colon\bar{K}(\theta,\sigma_{t})\geq\varepsilon\right\} $.
Thus for all $t\geq\max\{t_{\eta},1/\varepsilon\}$, 
\begin{equation}
\frac{A_{t}^{\varepsilon}}{B_{t}^{\varepsilon}}\leq e^{t2\eta}\frac{e^{-t\varepsilon/2}}{\mu_{0}\left(\left\{ \theta\colon\bar{K}(\theta,\sigma_{t})\leq\varepsilon/2\right\} \right)}.\label{eq:ratioA/B}
\end{equation}

In Appendix \ref{pf:eq:kappa_eps}, we show that continuity of $\bar{K}$
and compactness of $\Delta X$ imply that
\begin{equation}
\kappa_{\varepsilon}\equiv\inf_{\sigma\in\Delta X}\mu_{0}\left(\left\{ \theta\colon\bar{K}(\theta,\sigma)\leq\varepsilon/2\right\} \right)>0\label{eq:kappa_eps}
\end{equation}
for all $\varepsilon>0$. Thus, setting $\eta=\varepsilon/8>0$, (\ref{eq:ratioA/B})
implies that, for all $t\geq\max\{t_{\eta},1/\varepsilon\}$, 
\[
\frac{A_{t}^{\varepsilon}}{B_{t}^{\varepsilon}}\leq\frac{e^{-t\varepsilon/4}}{\kappa_{\varepsilon}},
\]
which goes to zero as $t\rightarrow\infty$.
\end{proof}
\medskip{}

In Section \ref{sec:Characterization-of-asymptotic}, we use Theorem
\ref{Theo:Berk} to approximate the agent's belief, $\mu_{t}$, with
the set of probability measures with support in $\{\theta\in\Theta:K(\theta,\sigma_{t})-K^{*}(\sigma_{t})\leq\delta_{t}\}$,
where $\delta_{t}\rightarrow0$. Therefore, we will be able to study
the asymptotic behavior of $(\sigma_{t})_{t}$ via a stochastic difference
equation that only depends on $\sigma_{t}$ and a vanishing approximation
error, and not on $\mu_{t}$.

\section{\label{sec:Characterization-of-asymptotic}Asymptotic characterization
of action frequencies}

In this section, we propose a method to study the asymptotic behavior
of the frequencies of actions. Among other benefits, one can use the
method to determine if behavior converges or not. The key departure
from previous approaches in the literature is to focus on the evolution
of frequencies of actions. Using the characterization of beliefs in
Theorem \ref{Theo:Berk}, we write this evolution as a stochastic
difference equation expressed exclusively in terms of the frequencies
of actions. We then use tools from stochastic approximation developed
by \citet*{benaim2005stochastic} (henceforth, BHS2015) to characterize
the solutions of this difference equation in terms of the solution
to a generalization of a differential equation.

We first provide a heuristic description of our approach. The sequence
of frequencies of actions, $(\sigma_{t})_{t}$, can be written recursively
as follows:
\begin{equation}
\sigma_{t+1}=\sigma_{t}+\frac{1}{t+1}\left(\mathbf{1}(x_{t+1})-\sigma_{t}\right),\label{eq:system0-1}
\end{equation}
where $\mathbf{1}(x_{t+1})=(\mathbf{1}_{x}(x_{t+1}))_{x\in X}$ and
$\mathbf{1}_{x}(x_{t+1})$ is the indicator function that takes the
value $1$ if $x_{t+1}=x$ and $0$ otherwise. 

By adding and subtracting the conditional expectation of $\mathbf{1}(x_{t+1})$
(i.e., the probability that each action is played at time $t+1$ given
the belief at time $t+1$), we obtain 
\begin{equation}
\sigma_{t+1}=\sigma_{t}+\frac{1}{t+1}\left(E\left[\mathbf{1}(x_{t+1})\mid\mu_{t+1}\right]-\sigma_{t}\right)+\frac{1}{t+1}\underset{=0}{\left(\underbrace{\mathbf{1}(x_{t+1})-E\left[\mathbf{1}(x_{t+1})\mid\mu_{t+1}\right]}\right)}.\label{eq:system0-1-1}
\end{equation}

The last term in equation (\ref{eq:system0-1-1}) is exactly equal
to zero because the agent chooses pure actions.\footnote{More generally, if the agent were allowed to mix, this last term is
a Martingale difference sequence and essentially adds a noise term
to the equation that can be controlled asymptotically in a manner
that is standard in the theory of stochastic approximation. Theorem
\ref{Theo:APT} continues to hold as stated, where now $\Delta F(\mu)$
is a set of \emph{compound} lotteries, i.e., it is the set of all
distributions over actions $\hat{\sigma}$ that are induced by some
compound lottery $z$ chosen from $\Delta F(\mu)$, that is, $\hat{\sigma}(x)=\int_{\sigma\in F(\mu)}z(\sigma)\sigma(x)d\sigma$
for each $x$.} The reason it is hard to characterize $(\sigma_{t})_{t}$ using (\ref{eq:system0-1-1})
is that its evolution depends on the agent's belief. If we could somehow
write the belief $\mu_{t+1}$ as a function of $\sigma_{t}$, then
we would have a recursion where $\sigma_{t+1}$ depends only on $\sigma_{t}$.
This is where Theorem \ref{Theo:Berk} from Section \ref{sec:asymptotic-beliefs}
is useful. This theorem will allow us to approximate $\mu_{t+1}$
with a set of probability measures that depends on $\sigma_{t}$.

The ultimate objective is not really to approximate $\mu_{t+1}$ but
rather the conditional expectation $E\left[\mathbf{1}(x_{t+1})\mid\mu_{t+1}\right]$
in equation (\ref{eq:system0-1-1}). The conditional expectation,
however, is typically discontinuous in the belief (this is particularly
so for a belief under which the agent is indifferent between two actions).
Thus, replacing $\mu_{t+1}$ with a good approximation does not necessarily
yield a good approximation for the conditional expectation. We tackle
this discontinuity issue by replacing the function $\mu\mapsto E\left[\boldsymbol{1}(x_{t+1})\mid\mu\right]$
with a correspondence that contains this function and is well behaved.

To see how this approach works, note that $E[\boldsymbol{1}(x_{t+1})\mid\mu]\in\Delta F(\mu)$
for all $\mu$. Therefore, we can view equation (\ref{eq:system0-1-1})
as a particular case of the following stochastic difference \emph{inclusion}:
\begin{equation}
\sigma_{t+1}=\sigma_{t}+\frac{1}{t+1}\left(r_{t+1}-\sigma_{t}\right),\label{eq:system0-1-1-1}
\end{equation}
where $r_{t+1}\in\Delta F(\mu_{t+1})$. It is called a difference
inclusion because $r_{t+1}$ can take multiple values. Importantly,
we use Theorem \ref{Theo:Berk} to approximate $\mu_{t+1}$ with the
set of probability measures $\mu$ satisfying $\int_{\Theta}(K(\theta,\sigma_{t})-K^{*}(\sigma_{t}))\mu(d\theta)\leq\delta_{t}$,
where $\delta_{t}\rightarrow0$ is a vanishing approximation error.
In particular, if the error were exactly zero, the set would be equal
to $\Delta\Theta(\sigma_{t})$. More generally, the difference equation
(\ref{eq:system0-1-1-1}) can be written entirely in terms of $(\sigma_{t})_{t}$
and approximation errors.

A key insight from the theory of stochastic approximation is that,
in order to characterize a discrete-time process such as $(\sigma_{t})_{t}$,
it is convenient to work with its continuous-time interpolation. Because
of the multiplicity inherent in equation (\ref{eq:system0-1-1-1}),
we apply the specific methods developed by BHS2015, who extend \citet{benaim1996dynamical}'s
ordinary-differential equation method to the case of differential
inclusions.\footnote{See \citet{borkar2009stochastic} for a textbook treatment of the
ordinary-differential equation method in stochastic approximation.}

Set $\tau_{0}=0$ and $\tau_{t}=\sum_{i=1}^{t}1/i$ for $t\geq1$.
The continuous-time interpolation of $(\sigma_{t})_{t}$ is the function
$\mathbf{w}:\mathbb{R}_{+}\rightarrow\Delta X$ defined as
\begin{equation}
\mathbf{w}(\tau_{t}+s)=\sigma_{t}+s\frac{\sigma_{t+1}-\sigma_{t}}{\tau_{t+1}-\tau_{t}},\,\,\,\,\,\,\,\,\,\,s\in[0,\frac{1}{t+1}).\label{eq:interpolation}
\end{equation}
Figure \ref{fig:interpol} illustrates this simple interpolation for
a specific value of $x\in X$. A convenient property of the interpolation
is that it preserves the accumulation points of the discrete process.

\begin{figure}
\qquad{}\qquad{}\qquad{}\qquad{}\qquad{}\includegraphics{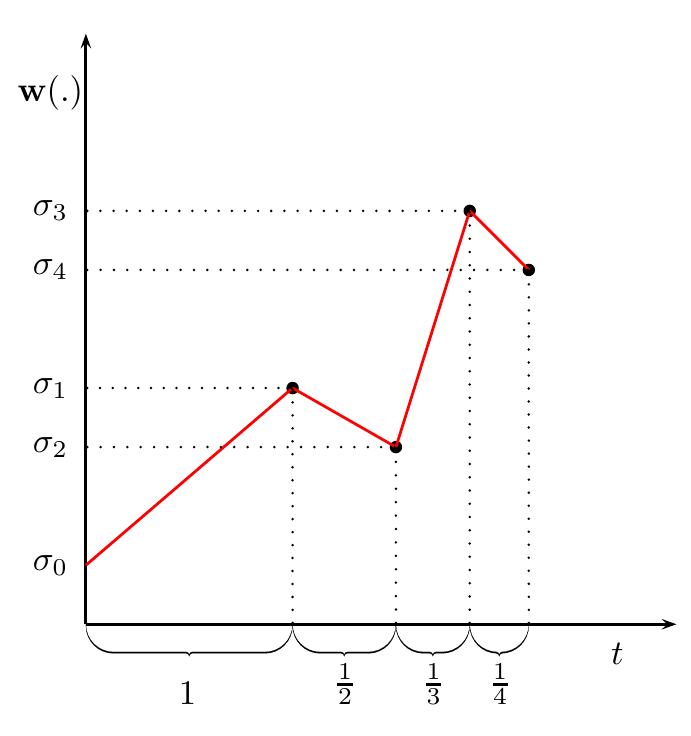}
\caption{\label{fig:interpol} Example of a continuous-time interpolation.}
\end{figure}
Equations (\ref{eq:system0-1-1-1}) and (\ref{eq:interpolation})
can be combined to show that the derivate of $\mathbf{w}$ with respect
to (a re-indexing of) time, which we denote by $\dot{\mathbf{w}}$,
is approximately given by $r_{t+1}-\sigma_{t}$. As argued earlier,
$r_{t+1}$ belongs to a set that depends on $\sigma_{t}$ and an approximation
error, and this set is equal to $\Delta F(\Delta\Theta(\sigma_{t}))$.
Thus, the derivate approximately takes values in $\Delta F(\Delta\Theta(\sigma_{t}))-\sigma_{t}$.
The next step is to replace $\sigma_{t}$ in this last expression
by its interpolation $\mathbf{w}(t)$. This replacement adds yet another
vanishing approximation error, and we therefore obtain, ignoring the
approximation error, that $\dot{\mathbf{w}}(t)\in\Delta F(\Delta\Theta(\mathbf{w}(t)))-\mathbf{w}(t)$.
Thus, we can show that the continuous-time interpolation of $(\sigma_{t})_{t}$
is well approximated by solutions to the following differential inclusion:

\begin{equation}
\boldsymbol{\dot{\sigma}}(t)\in\Delta F(\Delta\Theta(\boldsymbol{\sigma}(t)))-\boldsymbol{\sigma}(t).\label{eq:DI-1-1}
\end{equation}

To state the result formally, we first define what we mean by a solution
to the differential inclusion. A \textbf{solution to the differential
inclusion}\emph{ (\ref{eq:DI-1-1}) }\textbf{with initial point} $\sigma\in\Delta X$
is a mapping $\boldsymbol{\boldsymbol{\sigma}}:\mathbb{R}\rightarrow\Delta X$
that is absolutely continuous over compact intervals with the properties
that $\boldsymbol{\boldsymbol{\sigma}}(0)=\sigma$ and that (\ref{eq:DI-1-1})
is satisfied for almost every $t$. Let $S_{\sigma}^{T}$ denote the
set of solutions to (\ref{eq:DI-1-1}) over $[0,T]$ with initial
point $\sigma$. The assumption that $F$ is uhc implies that, for
every initial point, there exists a (possibly nonunique) solution
to (\ref{eq:DI-1-1}); see, e.g., \citet{aubin2012differential}.

We now state the main characterization result.

\medskip{}

\begin{thm}
\label{Theo:APT} Suppose that Assumptions \ref{ass:1}-\ref{ass:3}
hold and let $F$ be an uhc policy correspondence. For any policy
$f\in Sel(F)$, the following holds $P^{f}$-a.s.: For all $T>0$,
\begin{equation}
\lim_{t\rightarrow\infty}\inf_{\boldsymbol{\sigma}\in S_{\boldsymbol{\mathbf{w}}(t)}^{T}}\sup_{0\leq s\leq T}\left\Vert \mathbf{w}(t+s)-\boldsymbol{\sigma}(s)\right\Vert =0.\label{eq:APTeq}
\end{equation}
\end{thm}
\begin{proof}
See Appendix \ref{pf:Theo:APT}.
\end{proof}
\medskip{}

Theorem \ref{Theo:APT} says that, for any $T>0$, the curve $\mathbf{w}(t+\cdot):[0,T]\rightarrow\Delta X$
defined by the continuous-time interpolation of $(\sigma_{t})_{t}$
approximates some solution to the differential inclusion (\ref{eq:DI-1-1})
with initial condition $\mathbf{w}(t)$ over the interval $[0,T]$
with arbitrary accuracy for sufficiently large $t$. As we will show,
this result is convenient because it allows us to characterize asymptotic
properties of $(\sigma_{t})_{t}$ by solving the differential inclusion
in (\ref{eq:DI-1-1}).

BHS2005 refer to a function $\mathbf{w}$ satisfying (\ref{eq:APTeq})
as an asymptotic pseudotrajectory of the differential inclusion. They
show that the limit set of a (bounded) asymptotic pseudotrajectory
is internally chain transitive.\footnote{For a definition of an internally chain transitive set, see BHS2005,
Section 3.3, Definition VI.} Thus, one corollary of Theorem \ref{Theo:APT} is that the frequency
of actions converges almost surely to an internally chain transitive
set of the differential inclusion. Because the notion of internally
chain transitive is fairly complex, in the next section we provide
a series of results that help characterize behavior in most common
economic applications.

\section{\label{sec:Convergence-to-equilibrium}Convergence results}

We now present a series of implications of Theorem \ref{Theo:APT}
that can be readily applied to economic applications. Throughout this
section we assume that the agent chooses a policy $f$ that is a selection
from $F$ and that Assumptions 1-4 are satisfied. All probabilistic
statements are with respect to the corresponding probability measure
$P^{f}$.

\subsection{Equilibrium}

We begin by defining the notion of equilibrium as a stationary point
of the differential inclusion.
\begin{defn}
\label{def:equilibrium}$\sigma\in\Delta X$ is an \textbf{equilibrium}
given a policy correspondence $F$ if $\sigma\in\Delta F(\Delta\Theta(\sigma))$.
\end{defn}
\bigskip{}
If $\sigma$ is an equilibrium, then there is a solution of the differential
inclusion that starts at $\sigma$ and forever remains at $\sigma$.
The next result shows that, if the action frequency converges, then
it must converge to an equilibrium. In Section \ref{sec:berk-nash},
we relate this result to previous results in the literature and show
that the notion of equilibrium that arises naturally from our approach
is more general than notions previously considered.\medskip{}

\begin{prop}
\label{prop:converge->eqm}The following property holds almost surely:
If $\sigma_{t}$ converges to some point $\sigma^{*}$, then $\sigma^{*}$
must be an equilibrium.
\end{prop}
\begin{proof}
See Appendix \ref{Pf:prop:converge->eqm}.
\end{proof}
\medskip{}

\subsection{Attracting sets and repelling equilibrium}

Proposition \ref{prop:converge->eqm} shows that, if the action frequency
$\sigma_{t}$ converges, its limit must be an equilibrium of the differential
inclusion. However, this is not a complete characterization of the
long-run behavior of the action frequency, for two reasons. First,
the proposition applies only to the case in which the action frequency
converges. It does not tell us what happens when the action frequency
does not converge, and also it is not clear when the action frequency
converges. Second, even when the action frequency converges, if there
are multiple equilibria, the proposition does not tell us which one
will arise as a long-run outcome. In this section, we will introduce
two concepts, \textit{attracting sets} and \textit{repelling equilibria},
which are useful to make a better prediction about the asymptotic
behavior of the action frequency.

Let $d(\sigma,A)$ denote the distance from a point $\sigma$ to a
set $A$, that is, let $d(\sigma,A)=\inf_{\tilde{\sigma}\in A}\left\Vert \sigma-\tilde{\sigma}\right\Vert $.
The following definition is standard in the stochastic approximation
literature (e.g. BHS2015).\medskip{}

\begin{defn}
\label{def:attracting} A set $A\subseteq\triangle X$ is \textbf{\textit{\emph{attracting}}}
if there is a set $\mathcal{U}$ such that $A\subset\textrm{int}\mathcal{U}$
and such that for any $\varepsilon>0$, there is $T$ such that $d(\bm{\sigma}(t),A)<\varepsilon$
for any initial value $\bm{\sigma}(0)\in\mathcal{U}$, for any solution
$\bm{\sigma}\in S_{\bm{\sigma}(0)}^{\infty}$ to the differential
inclusion, and for any $t>T$.
\end{defn}
\medskip{}

In this definition, we require uniform convergence, in that as long
as the initial value is chosen from $\mathcal{U}$, $\bm{\sigma}(t)$
is in the $\varepsilon$-neighborhood of $A$ for \textit{all} periods
$t>T$. Intuitively, this implies that once $\bm{\sigma}(t)$ enters
the $\varepsilon$-neighborhood of $A$, it will never leave this
neighborhood. The largest set $\mathcal{U}$ which satisfies the property
in this definition is the \textbf{basin of attraction} of $A$, and
we will denote it by $\mathcal{U}_{A}$. A set $A$ is a \textbf{\textit{\emph{globally
attracting}}} if it is attracting and its basin of attraction is the
whole space $\triangle X$. An equilibrium $\sigma^{\ast}$ is \textbf{\textit{\emph{attracting}}}
if the set $A=\{\sigma^{\ast}\}$ is attracting.

The following proposition shows that an attracting set appears as
a long-run outcome in some sense. Let $E$ denote the set of all equilibria.

\medskip{}

\begin{prop}
\label{Prop:Attracting} The following results hold:
\end{prop}
\begin{itemize}
\item[(i)] If $A$ is globally attracting, then the action frequency $\sigma_{t}$
approaches this set $A$ almost surely: $\lim_{t\to\infty}d(\sigma_{t},A)=0$.
In particular, if $A$ is a globally attracting equilibrium, $\sigma^{t}$
converges to that equilibrium almost surely.
\item[(ii)] Suppose that there are finitely many attracting sets $(A_{1},\cdots,A_{N})$
such that $\triangle X$ is the union of the basins $(\mathcal{U}_{A_{1}},\cdots,\mathcal{U}_{A_{N}})$
of these attractors and of the equilibrium set $E$. Then almost surely,
$\sigma_{t}$ approaches the equilibrium set $E$ or one of these
attractors: $\lim_{t\to\infty}d(\sigma_{t},E)=0$ or $\lim_{t\to\infty}d(\sigma_{t},A_{n})=0$
for some $n$.
\end{itemize}
\begin{proof}
See Appendix \ref{Pf:Prop:Attracting}.
\end{proof}
\medskip{}

Recall that in the definition of attracting sets, we require uniform
convergence. This property is crucial in order to obtain Proposition
\ref{Prop:Attracting}. To see this, let $\bm{w}(t)$ denote the current
action frequency. Theorem \ref{Theo:APT} implies that the motion
of the action frequency in the future is approximated by a solution
$\bm{\sigma}\in S_{\bm{w}(t)}^{\infty}$ to the differential inclusion
for some (long but) \textit{finite time} $T$; but it does not guarantee
that the action frequency $\bm{w}$ is approximated by $\bm{\sigma}$
forever. So even if all solutions $\bm{\sigma}\in S_{\bm{w}(t)}^{\infty}$
starting from the current value $\bm{w}(t)$ converge to some equilibrium
$\sigma^{\ast}$, the action frequency $\bm{w}$ may not converge
there.\footnote{This is a so-called ``shadowing'' problem in the literature on stochastic
approximation. See Section 8 of \citet{Benaim99} for more details.} Formally, Theorem \ref{Theo:APT} implies that, for any $T$ and
$\varepsilon>0$, if $t$ is large enough, then $\lVert\bm{w}(t+T)-\bm{\sigma}(T)\rVert<\varepsilon$
for some $\bm{\sigma}\in S_{\bm{w}(t)}^{\infty}$, so the action frequency
$\bm{w}(t+T)$ in time $t+T$ is close to the equilibrium $\sigma^{\ast}$.
However, after time $t+T$, the action frequency $\bm{w}$ can be
quite different from $\bm{\sigma}$, and it may move away from the
equilibrium $\sigma^{\ast}$.

This suggests that in order to guarantee convergence to $\sigma^{\ast}$,
we need a stronger assumption, and uniform convergence is precisely
the property we want. To see how it works, note that Theorem \ref{Theo:APT}
can be applied iteratively, so that the action frequency $\bm{w}$
from time $t+T$ to $t+2T$ is approximated by a solution $\bm{\sigma}^{\prime}\in S_{\bm{w}(t+T)}^{\infty}$
starting from $\bm{w}(t+T)$. As mentioned earlier, this value $\bm{w}(t+T)$
is close to the equilibrium $\sigma^{\ast}$. So if the equilibrium
$\sigma^{\ast}$ is attracting, then $\bm{w}(t+T)$ is in the basin
of $\sigma^{\ast}$, and the solution $\bm{\sigma}^{\prime}$ starting
from this point stays around the equilibrium $\sigma^{\ast}$. This
in turn implies that the action frequency $\bm{w}(t+2T)$ in time
$t+2T$ is also close to $\sigma^{\ast}$. A similar argument shows
that $\bm{w}(t+nT)$ in time $t+nT$ is close to $\sigma^{\ast}$
for every $n=1,2,\cdots$. The proof of Proposition \ref{Prop:Attracting}
generalizes this idea and shows convergence to attracting sets.\footnote{Formally, the whole path of $\bm{w}$ is approximated by a \textit{chain}
of trajectories $(\bm{\sigma}_{1},\bm{\sigma}_{2},\cdots)$ where
$\lVert\bm{\sigma}_{n}(T)-\bm{\sigma}_{n+1}(0)\rvert<\varepsilon$,
and uniform convergence ensures that this chain of trajectories converges
to $\sigma^{\ast}$.}

Now we apply the result to an example where the action does not converge,
but the action frequency does.

\medskip{}

\begin{example} \label{example_negativereinf} The consequence space
is $Y=\{0,1\}$. There are two actions, $x_{1}$ and $x_{2}$, and
the probability that the consequence equals one depends on the action:
$Q(1\mid x_{1})=3/4$ and $Q(1\mid x_{2})=1/4$. The agent, however,
has a misspecified model and believes the consequence does not depend
on the action: $Q_{\theta}(1\mid x_{k})=\theta\in\Theta=[0,1]$ for
$k\in\{1,2\}$, i.e., according to model $\theta$, the probability
that the consequence is $y=1$ is $\theta$ irrespective of the action.
The KLD function in this case is 
\[
K(\sigma,\theta)=\sigma(x_{1})\left(\frac{3}{4}\ln\frac{3/4}{\theta}+\frac{1}{4}\ln\frac{1/4}{1-\theta}\right)+\sigma(x_{2})\left(\frac{1}{4}\ln\frac{1/4}{\theta}+\frac{3}{4}\ln\frac{3/4}{1-\theta}\right),
\]
and there is a unique minimizer $\theta(\sigma)=\frac{3}{4}\sigma(x_{1})+\frac{1}{4}\sigma(x_{2})$.
Naturally, the model that best fits the true model is a convex combination
of the probability that $y=1$ under actions $x_{1}$ and $x_{2}$.

The payoff function satisfies $\pi(x_{1},y)=1-y$ and $\pi(x_{2},y)=y$;
in particular, the agent prefers action $x_{1}$ if $y=0$ and $x_{2}$
if $y=1$. Letting $E_{\mu}[y]\equiv\int\theta\mu(d\theta)$ denote
the agent's perceived expected value of $y$, the optimal correspondence
satisfies $F(\mu)=\{x_{1}\}$ if $E_{\mu}[y]<.5$, $=\{x_{2}\}$ if
$E_{\mu}[y]<.5$, and $=\{x_{1},x_{2}\}$ if $E_{\mu}[y]=.5$.

Actions in this example are negatively reinforcing in the sense that
doing more of one action makes the agent want to do less of that action.
This feature can be seen in Figure \ref{fig:Example_negat_reinfor},
where we have plotted $\sigma\mapsto\Delta F(\Delta\Theta(\sigma))$.
For example, if the agent takes pure action $x_{1}$, i.e., $\sigma(x_{1})=1$,
then the closest model is $\theta(\sigma)=\frac{3}{4}$, so the agent's
belief is degenerate at $\theta=3/4$, i.e., $\mu=\delta_{3/4}$,
and she prefers to take action $x_{2}$, not $x_{1}$. Similarly,
if the agent takes only action $x_{2}$, she prefers to take action
$x_{1}$.

\begin{figure}
\qquad{}\qquad{}\qquad{}\qquad{}\includegraphics{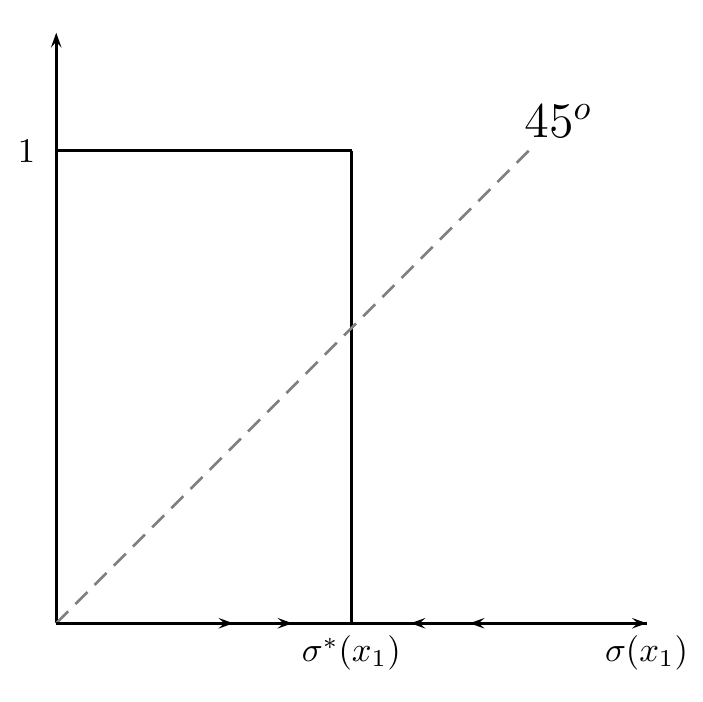}

\caption{\label{fig:Example_negat_reinfor} Globally attracting equilibrium}
\end{figure}
The feature of negatively reinforcing actions is present in several
examples in the literature, and previous work has shown that the action
does not converge in those examples (e.g., \citet{nyarko1991learning},
\citet{esponda2016berk}, \citet*{fudenberg2017active}).\footnote{Negative reinforcement is also present in some of the examples in
\citet{spiegler2016bayesian} as well as in the voting environments
of Esponda and Pouzo (2017, 2019a)\nocite{esponda2017conditional}\nocite{esponda2019retrospective}
and \citet{esponda2018endogenous} and in the investment environment
of \citet{jehiel2018investment}.} We can use our differential inclusion to go beyond this result and
show that the action \emph{frequency} does converge. In the example,
$\sigma^{*}(x_{1})=\sigma^{*}(x_{2})=1/2$ is the unique equilibrium
point: Given $\sigma^{*}$, the closest model is $\theta(\sigma^{*})=1/2$,
and, given the belief $\delta_{1/2}$, the agent is indifferent between
each of the actions in the support of $\sigma^{*}$. Moreover, as
Figure \ref{fig:Example_negat_reinfor} shows, for any initial condition,
the solutions to the differential inclusion converge to $\sigma^{*}$,
and so $\{\sigma^{*}\}$ is a globally attracting set. Proposition
\ref{Prop:Attracting}(i) implies that the action frequency almost
surely converges to $\sigma^{*}$. \end{example}

\medskip{}

In the next example, we show that the attracting set need not be an
equilibrium.\medskip{}

\begin{example} \label{exampletriangle} The consequence space is
$Y=\mathbb{R}^{3}$. There are three actions, $x_{1}$, $x_{2}$,
and $x_{3}$. Given an action $x_{k}$, the consequence $y$ follows
the normal distribution $N(e_{k},I)$ where $e_{k}\in\mathbb{R}^{3}$
is the unit vector whose $k$th component is one, and $I$ is the
identity matrix. However, the agent does not recognize that the action
influences the consequence. Formally, the model space is the probability
simplex $\Theta=\{\theta=(\theta_{1},\theta_{2},\theta_{3})|\sum_{k=1}^{3}\theta_{k}=1,\;\;\theta_{k}\geq0\;\;\forall k\}$,
and for each model $\theta=(\theta_{1},\theta_{2},\theta_{3})$, the
agent believes that $y$ follows the normal distribution $N(\theta,I)$.
Assume that for each degenerate belief $\delta_{\theta}$, the policy
$F(\delta_{\theta})$ is given as in Figure \ref{figuretrianglebr},
where the triangle represents the model space $\Theta$. For example,
if the current belief puts probability one on the model $\theta=e_{1}$,
then the policy $F$ selects the action $x_{2}$.

 \begin{figure}[htbp]  
\begin{minipage}{0.5\hsize}   
\begin{center}       
\includegraphics{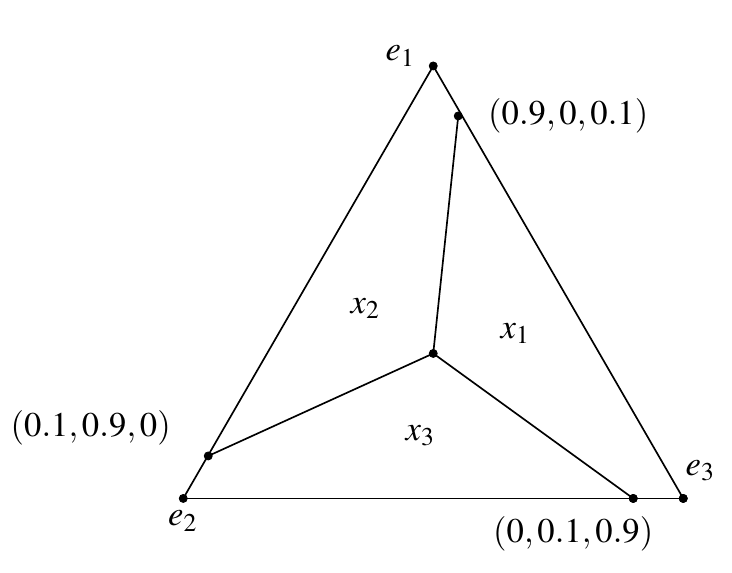}
\caption{Policy $F(\delta_{\theta})$ for each model $\theta$}        
\label{figuretrianglebr}   
\end{center}        
\end{minipage}  
\begin{minipage}{0.5\hsize}   
\begin{center}       
\includegraphics{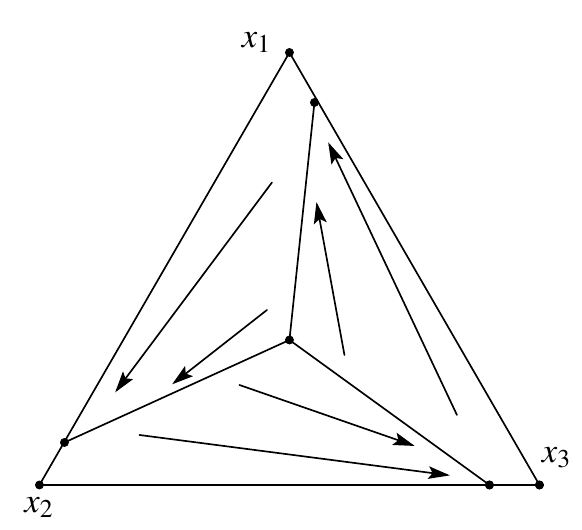} 
\caption{Differential Inclusion}     
\label{figuretriangledrift}   
\end{center}  
\end{minipage}  
\end{figure}

In this example, given a mixed action $\sigma\in\triangle X$, the
consequence follows the normal distribution $N(\sigma,I)$, so the
Kullback-Leibler divergence $K(\theta,\sigma)$ has a unique minimizer
$\theta=\sigma$. Accordingly, a solution to the differential inclusion
is described as in Figure \ref{figuretriangledrift}, where the triangle
represents the whole action space $\triangle X$ and each arrow points
to the corresponding vertex in the large triangle.

This example has a unique equilibrium, $\sigma^{\ast}=(\frac{1}{3},\frac{1}{3},\frac{1}{3})$.
This equilibrium is not attracting. Indeed, starting from any nearby
point $\sigma\neq\sigma^{\ast}$, a solution to the differential inclusion
moves away from the equilibrium, as described in Figure \ref{figureunstable}.

On the other hand, the cycle described by the arrows in Figure \ref{figuretrianglecycle}
is attracting. The basin of attraction is the whole space $\triangle X$
except the equilibrium point $\sigma^{\ast}=(\frac{1}{3},\frac{1}{3},\frac{1}{3})$.
That is, given any initial value $\sigma\neq\sigma^{\ast}$, any solution
to the differential inclusion will eventually follow this cycle. (The
proof is straightforward and hence omitted. )

 \begin{figure}[htbp] 	
\begin{minipage}{0.5\hsize} 		
\begin{center} 			
\includegraphics{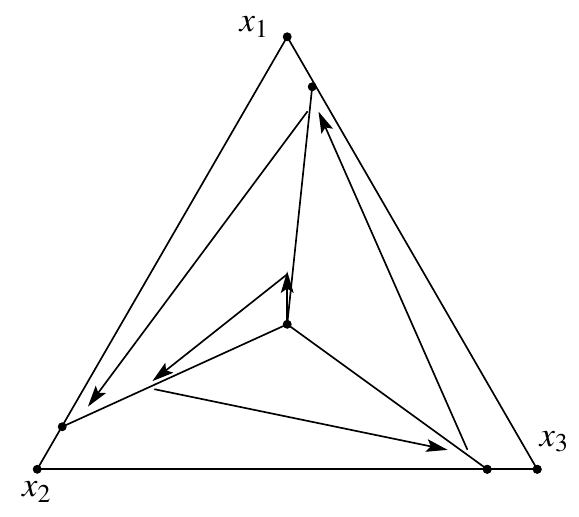}			
\caption{Instability of the equilibrium }  			
\label{figureunstable} 		
\end{center}       	
\end{minipage} 	
\begin{minipage}{0.5\hsize} 		
\begin{center} 			
\includegraphics{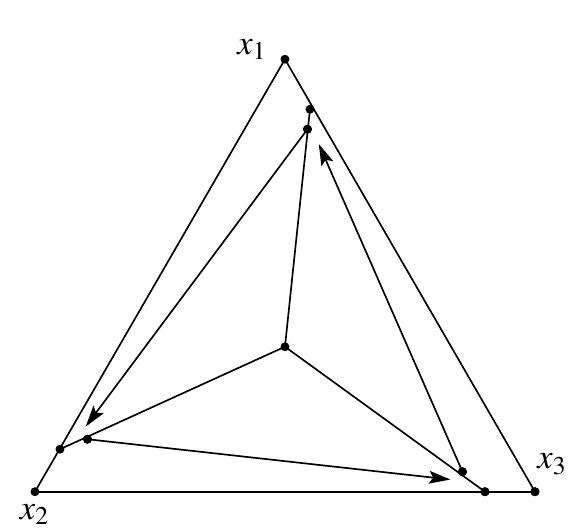}	
\caption{Limit cycle.}   
\label{figuretrianglecycle} 		
\end{center} 	
\end{minipage} 
\end{figure}

\end{example}

\medskip{}

Proposition \ref{Prop:Attracting}(ii) implies that in Example \ref{exampletriangle},
the action frequency $\sigma_{t}$ must converge to the (non-attracting)
equilibrium $\sigma=(\frac{1}{3},\frac{1}{3},\frac{1}{3})$ or follow
the limit cycle described in Figure \ref{figuretrianglecycle}. But
which one is more likely to occur? It turns out that the equilibrium
$\sigma^{\ast}$ in the example above is unstable, in that the action
frequency never converges there. So the action frequency follows the
limit cycle almost surely.

To see why the equilibrium $\sigma^{\ast}$ is unstable, suppose that
the current action frequency is exactly this equilibrium, i.e., $\sigma_{t}=\sigma^{\ast}$.
Suppose also that the agent chooses some action today, say $x_{1}$.
This changes the action frequency in the next period, and we have
$\sigma_{t+1}=\frac{1}{t+1}\delta_{x_{1}}+\frac{t}{t+1}\sigma^{\ast}$.
Note that this new action frequency is slightly different from the
equilibrium $\sigma^{\ast}$. Then starting from this action frequency,
a solution to the differential inclusion moves away from the equilibrium
(See Figure \ref{figureunstable}), which implies instability of $\sigma^{\ast}$.
More formally, this equilibrium $\sigma^{\ast}$ is \textit{repelling}
in the following sense:

\medskip{}

\begin{defn}
\label{def:repelling_eqm} An equilibrium $\sigma^{\ast}$ is \textit{repelling}
if there is a natural number $T$ and an open neighborhood $\mathcal{U}$
of $\sigma^{\ast}$ such that for any $\sigma\in\mathcal{U}$, for
any $x\in F(\triangle\Theta(\sigma^{\ast}))$, and for any $\overline{\beta}\in(0,1)$,
there is $\beta\in(\overline{\beta},1)$ such that for any $\bm{\sigma}\in S_{\beta\sigma+(1-\beta)\delta_{x}}^{\infty}$,
we have $\bm{\sigma}(t)\notin\mathcal{U}$ for some $t\in[0,T]$.
\end{defn}
\medskip{}

In Example \ref{exampletriangle}, starting from \textit{any} nearby
point $\sigma\neq\sigma^{\ast}$ of the equilibrium $\sigma^{\ast}$,
the solution to the differential inclusion moves away from the equilibrium
$\sigma^{\ast}$. In such a case, the condition stated in the definition
is satisfied, so this equilibrium $\sigma^{\ast}$ is repelling.

But our definition of repelling equilibrium is a bit more general.
Roughly, an equilibrium $\sigma^{\ast}$ is repelling if starting
from \textit{almost all} nearby points $\sigma\neq\sigma^{\ast}$
of the equilibrium $\sigma^{\ast}$, all the solutions to the differential
inclusion eventually leave its neighborhood. This is illustrated in
the following example:

\medskip{}

\begin{example} We add one more action $x_{3}^{\prime}$ to Example
\ref{exampletriangle}. This new action $x_{3}^{\prime}$ is redundant,
and is identical to the action $x_{3}$. Formally, the signal distribution
given the action $x_{3}^{\prime}$ is $N(e_{3},I)$, and the policy
$F(\mu)$ contains $x_{3}^{\prime}$ for all $\mu$ such that $F(\mu)$
contains $x_{3}$ in Example \ref{exampletriangle}. The agent still
believes that the action does not influence the signal distribution.

This example has a continuum of equilibria; any mixed action $\sigma$
with $\sigma(x_{1})=\sigma(x_{2})=\sigma(x_{3})+\sigma(x_{3}^{\prime})=\frac{1}{3}$
is an equilibrium. Pick one equilibrium $\sigma^{\ast}$, and pick
an open neighborhood $\mathcal{U}$. This neighborhood $\mathcal{U}$
contains equilibrium points and non-equilibrium points. The set of
equilibrium points is continuous, but has measure zero; so almost
all the points in $\mathcal{U}$ are non-equilibrium points. Starting
from these non-equilibrium points, all the solutions to the differential
inclusion leave the neighborhood $\mathcal{U}$, just as described
in Figure \ref{figureunstable}. However, starting from the equilibrium
points, a solution to the differential inclusion can stay there forever.
So $\mathcal{U}$ contains some points from which the solution to
the differential inclusion does not leave $\mathcal{U}$. Still, this
equilibrium $\sigma^{\ast}$ is repelling. Indeed, given any point
$\sigma\in\mathcal{U}$ and given any action $x$, if we choose $\beta$
sufficiently close to one, the perturbed point $\beta\sigma+(1-\beta)\delta_{x}$
is not an equilibrium; so starting from this perturbed point, the
solution to the differential inclusion eventually leaves $\mathcal{U}$.\end{example}\medskip{}

The following proposition asserts that repelling equilibria do not
arise as long-run outcomes.\medskip{}

\begin{prop}
\label{Prop:Repelling}If $\sigma^{\ast}$ is a repelling equilibrium,
then the action frequency $\sigma_{t}$ converges to $\sigma^{\ast}$
with probability zero.
\end{prop}
\begin{proof}
See Appendix \ref{Pf:Prop:Repelling}.
\end{proof}

\subsection{\label{subsec:Converge_attracting}Convergence to attracting sets
for some prior}

Proposition \ref{Prop:Attracting} provides a useful set of conditions
under which the action frequency converges to an attracting set, such
as the set of equilibria. Moreover, Proposition \ref{Prop:Repelling}
shows that the frequency cannot converge to repelling equilibria.
These propositions, however, do not imply that the action frequency
converges to any one specific attracting set or equilibrium (unless
it is globally attracting). We will show that if an attracting set
$A$ satisfies some additional property, then the action frequency
converges to it (i.e., $\lim_{t\to\infty}d(\sigma_{t},A)=0$) with
positive probability for some initial prior.\footnote{Theorem 7.3 of \citet{Benaim99} gives a sufficient condition for
convergence to an attracting set for some prior. This result, however,
relies on a technical assumption ((24) in his paper) that does not
hold in our environment. Hence we cannot use his theorem and need
to develop a new tool.} Throughout this section, let $B_{\varepsilon}(A)$ denote the $\varepsilon$-neighborhood
of $A$, i.e., it is the set of all $\sigma$ such that $d(\sigma,A)<\varepsilon$.

We first introduce the idea of a ``perturbed differential inclusion.''
Given an initial value $\bm{\sigma}(0)$, let $\bm{S}_{\bm{\sigma}(0)}^{\infty,\varepsilon}$
denote the set of all solutions to the following differential inclusion:
\begin{align}
\dot{\bm{\sigma}}(t)\in\bigcup_{\tilde{\sigma}\in B_{\varepsilon}(\bm{\sigma}(t))}\triangle F(\triangle\Theta(\tilde{\sigma}))-\bm{\sigma}(t).\label{perturbeddi}
\end{align}
Recall that in the original differential inclusion, the agent chooses
an action from $F(\triangle\Theta(\bm{\sigma}(t)))$. In (\ref{perturbeddi}),
this choice set is expanded, so that the agent chooses an action from
$F(\triangle\Theta(\tilde{\sigma}))$, where $\tilde{\sigma}$ is
a perturbation of the current action frequency $\bm{\sigma}(t)$.

\medskip{}

\begin{defn}
\label{def:robust_attracting} A set $A$ is \textbf{\textit{\emph{robustly
attracting}}} if it is attracting and there is $\zeta>0$ and $\varepsilon>0$
such that for any initial value $\bm{\sigma}(0)\in B_{\zeta}(A)$,
any solution $\bm{\sigma}\in\bm{S}_{\bm{\sigma}(0)}^{\infty,\varepsilon}$
to the perturbed differential inclusion never leaves the basin $\mathcal{U}_{A}$;
i.e., $\bm{\sigma}(t)\in\mathcal{U}_{A}$ for all $t\geq0$.
\end{defn}
\medskip{}

In some special cases, attracting sets and robustly attracting sets
are equivalent. For example, as will be explained in Proposition \ref{propequivalence},
attracting sets are robustly attracting when $\Theta$ is the one-dimensional
interval $[0,1]$. The same result holds when there are only two actions
i.e., $|X|=2$. (The proof is straightforward and hence omitted.)
However, in general, attracting sets need not be attracting. Such
an example can be found in Appendix \ref{subsec:ex_robust_attract}.

A sufficient condition for a set $A$ to be robustly attracting is
that the (non-perturbed) differential inclusion has a contraction
property in a neighborhood of $A$. Formally, let $V(\sigma)=d(\sigma,A)$,
and suppose that there is an open neighborhood $U$ of $A$ such that
$(\tilde{\sigma}-\sigma)\cdot\nabla V(\sigma)<0$ for all $\sigma\in U\setminus A$
and $\tilde{\sigma}\in\triangle F(\triangle\Theta(\sigma))$.\footnote{ More generally, $A$ is robustly attracting if there is an open neighborhood
$U$ of $A$ and a function $V:U\to\bm{R}_{+}$ such that (i) $V(\sigma)=0$
if and only if $\sigma\in A$, (ii) $(\tilde{\sigma}-\sigma)\cdot\nabla V(\sigma)<0$
for all $\sigma\in U\setminus A$ and $\tilde{\sigma}\in\triangle F(\triangle\Theta(\sigma))$,
and (iii) $\nabla V$ is Lipchitz-continuous. Note that condition
(ii) here is a bit more demanding than Lyapunov stability, which requires
$V(\bm{\sigma}(t))<V(\bm{\sigma}(0))$ for all $\bm{\sigma}(0)$ and
$\bm{\sigma}\in S_{\bm{\sigma}(0)}^{\infty}$.} Then this $A$ is robustly attracting.\footnote{The proof is as follows. It is obvious that $A$ is attracting, so
we will show that the condition stated in the definition of robustly
attracting sets is satisfied. Pick $\zeta>0$ such that $B_{2\zeta}(A)$
is in the set $U$ defined above. Let $C=\overline{B_{2\zeta}(A)}\setminus B_{\zeta}(A)$.
Since $C$ is compact, $(\hat{\sigma}-\sigma)\cdot\nabla V(\sigma)<0$
is bounded away from zero uniformly. Then Lipchitz-continuity of $\nabla V$
ensures that there is $\varepsilon>0$ such that $(\hat{\sigma}-\tilde{\sigma})\cdot\nabla V(\tilde{\sigma})<0$
for all $\sigma\in C$, $\tilde{\sigma}\in B_{\varepsilon}(\sigma)$,
and $\hat{\sigma}\in F(\triangle\Theta(\sigma))$. This implies that
any solution to the $\varepsilon$-perturbed differential inclusion
(\ref{perturbeddi}) also has a contraction property in the interior
of the set $C$; i.e., if the current action frequency $\tilde{\sigma}$
is an interior point of $C$ and $d(\tilde{\sigma},C)\geq\varepsilon$,
then at the next instant, the action frequency becomes closer to the
set $A$. This immediately implies that $A$ is robustly attracting. } Note that this contraction property is satisfied by any strict equilibrium;
a pure action $\delta_{x}$ is a \textbf{\textit{\emph{strict equilibrium}}}
if there is an open neighborhood $U$ of $\delta_{x}$ such that $F(\triangle\Theta(\tilde{\sigma}))=\{x\}$
for all $\tilde{\sigma}\in U$. So any strict equilibrium is robustly
attracting.

We will show that the action frequency converges to a robustly attracting
set with positive probability, at least for some initial prior.\medskip{}

\begin{prop}
\label{prop:converge_robust_attracting} For each robustly attracting
set $A$, there is an initial prior $\mu_{0}^{\ast}$ with full support
such that $\lim_{t\to\infty}d(\sigma_{t},A)=0$ with positive probability.
\end{prop}
\begin{proof}
See Appendix \ref{pf:prop:converge_robust_attracting}.
\end{proof}

\subsection{Belief convergence when $\Theta=[0,1]$}

In this subsection, we will focus on a special case in which the model
space is the one-dimensional interval $\Theta=[0,1]$. This special
case includes most of the current applications in the literature and
it allows us to provide a more powerful characterization of the action
frequency and the belief. Specifically, we will first explain that
our differential inclusion reduces to a one-dimensional problem; this
reduction considerably simplifies our analysis, because in general
the action frequency $\sigma$ is multi-dimensional and solving the
differential inclusion can be a difficult task. Then we will show
that the belief converges to an equilibrium belief almost surely,
and we will provide a simple characterization of attracting/repelling
equilibria.

Throughout this subsection, we will impose the following identifiability
assumption:

\medskip{}

\begin{assumption} \label{assumptionidentifiability} The following
two conditions hold:
\begin{itemize}
\item[(i)] For each $\sigma$, there is a unique minimizer of $K(\sigma,\theta)$
which we denote by $\theta(\sigma)\in[0,1]$, that is, $\Theta(\sigma)=\{\theta(\sigma)\}$.
\item[(ii)] For each $\sigma$ with $\theta(\sigma)\in(0,1)$, we have $\left.\frac{\partial^{2}K(\sigma,\theta)}{\partial\theta^{2}}\right|_{\theta=\theta(\sigma)}>0$.
\end{itemize}
\end{assumption}

\medskip{}

Part (i) is what we call the identifiability condition, which asserts
that for each mixed action $\sigma$, there is a unique model which
best fits the true world. EP2016 provide a more detailed discussion
about this identifiability condition. Note that the best model $\theta(\sigma)$
is continuous in $\sigma$, because $\Theta(\sigma)$ is upper hemi-continuous
in $\sigma$,

Part (ii) requires that whenever $\theta(\sigma)$ is an interior
solution (so that the first-order condition is satisfied at the minimum),
it satisfies the second-order condition. This technical assumption
is crucial for the strict monotonicity result (Proposition \ref{propmonotone}(iii)),
which is needed to prove instability of repelling models (Proposition
\ref{propequivalence2}). But all other results remain true even if
this assumption (ii) is dropped.

The following proposition shows that the closest model $\theta(\sigma)$
is monotone with respect to the action $\sigma$. In the proof, we
first show that the KL divergence $K(\sigma,\theta)$ has the increasing
differences property. Then the result follows from the monotone selection
theorem of \citet{topkis1998supermodularity} and \citet{edlin1998strict}.

\medskip{}

\begin{prop}
\label{propmonotone} Suppose that Assumption \ref{assumptionidentifiability}
holds. Pick any $\sigma$ and $\tilde{\sigma}$, and for each $\beta\in[0,1]$,
let $\sigma_{\beta}=\beta\sigma+(1-\beta)\tilde{\sigma}$. Then the
following results are true: 
\end{prop}
\begin{itemize}
\item[(i)] If $\theta(\sigma)=\theta(\tilde{\sigma})$, then $\theta(\sigma_{\beta})=\theta(\sigma)$
for all $\beta\in[0,1]$.
\item[(ii)] If $\theta(\tilde{\sigma})<\theta(\sigma)$, then $\theta(\sigma_{\beta})$
is weakly increasing with respect to $\beta$.
\item[(iii)] If $\theta(\tilde{\sigma})<\theta(\sigma)$, then $\theta(\sigma_{\beta_{1}})<\theta(\sigma_{\beta_{2}})$
for any $\beta_{1}$ and $\beta_{2}$ such that $\beta_{1}<\beta_{2}$
and $\theta(\sigma_{\beta_{1}})\in(0,1)$.
\end{itemize}
\begin{proof}
See Appendix \ref{pf:propmonotone}.
\end{proof}
\medskip{}

The monotonicity result above ensures that the motion of the closest
model $\theta(\sigma)$ is characterized by a simple, one-dimensional
problem. Note that when $\Theta=[0,1]$, the best model $\theta(\sigma)$
can move in only three directions; it can go up, down, or stay the
same. In particular, since the motion of the action frequency is approximated
by 
\[
\dot{\bm{\sigma}}=\sigma-\bm{\sigma}(t)
\]
for some $\sigma\in\triangle F(\delta_{\theta(\bm{\sigma}(t))})$,
the monotonicity result in Proposition 5 implies that, at each time
$t$,
\begin{itemize}
\item $\theta(\bm{\sigma}(t))$ moves up if $\theta(\sigma)>\theta(\bm{\sigma}(t))$
for all $\sigma\in\triangle F(\delta_{\theta(\bm{\sigma}(t))})$.
\item $\theta(\bm{\sigma}(t))$ moves down if $\theta(\sigma)<\theta(\bm{\sigma}(t))$
for all $\sigma\in\triangle F(\delta_{\theta(\bm{\sigma}(t))})$.
\end{itemize}
To better understand the motion of $\theta(\bm{\sigma}(t))$, consider
the following example:

\medskip{}

\begin{example} \label{exampleonedimensional} The consequence space
is $Y=\mathbb{R}$, and the agent has two actions, $x_{0}$ and $x_{1}$.
Given an action $x_{k}$, the consequence $y$ follows the normal
distribution $N(k,1)$. The agent does not recognize that the action
influences the consequence, and she believes that given a model $\theta\in[0,1]$,
$y$ follows the normal distribution $N(\theta,1)$ regardless of
the chosen action. Consider an upper hemi-continuous policy $F$ which
satisfies 
\[
F(\delta_{\theta})=\left\{ \begin{array}{ll}
\{x_{0}\} & \textrm{if }\theta\in[0,\frac{1}{3})\cup(\frac{2}{3},1]\\
\{x_{1}\} & \textrm{if }\theta\in(\frac{1}{3},\frac{2}{3})\\
\{x_{0},x_{1}\} & \textrm{if }\theta\in\{\frac{1}{3},\frac{2}{3}\}
\end{array}\right..
\]

Given a mixed action $\sigma$, the consequence follows the normal
distribution $N(\sigma(x_{1}),1)$, so the closest model is $\theta(\sigma)=\sigma(x_{1})$.
Hence the motion of $\theta(\bm{\sigma}(t))$ can be described by
the arrows in Figure \ref{figureonedimensional}: $\theta(\bm{\sigma}(t))$
will move up in the middle region (i.e., $\theta(\bm{\sigma}(t))\in(\frac{1}{3},\frac{2}{3})$),
because the agent chooses the action $x_{1}$ and the corresponding
model is $\theta(\delta_{x_{1}})=1$. For the other region, $\theta(\bm{\sigma}(t))$
will move down because the agent chooses the action $x_{0}$ and the
corresponding model is $\theta(\delta_{x_{0}})=0$. 
\begin{figure}[htbp]
\centering{} 
\includegraphics{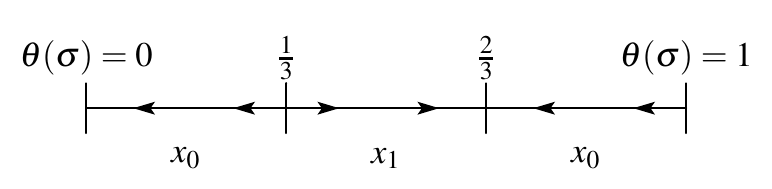}
\caption{Motion of $\theta(\bm{\sigma}(t))$}
\label{figureonedimensional}
\end{figure}
\end{example}

\medskip{}

Using the fact that the closest model $\theta(\sigma)$ follows the
simple rule above, we will now show that it converges almost surely.
The intuition is as follows. If the model space is $\Theta=[0,1]$
and $\theta(\sigma)$ follows a recursive rule as in Figure \ref{figureonedimensional},
it cannot be cyclic. This implies that $\theta(\sigma)$ cannot oscillate
forever and it must converge. A model $\theta^{\ast}$ is an \textbf{\textit{\emph{equilibrium
model}}} if there is an equilibrium $\sigma^{\ast}$ such that $\theta(\sigma^{\ast})=\theta^{\ast}$.
In Example \ref{exampleonedimensional}, there are three equilibrium
models, $0$, $1/3$, and $2/3$. Let $\Theta^{\ast}\subseteq\Theta$
denote the set of all equilibrium models. Then we have the following
result:

\medskip{}

\begin{prop}
\label{proponedimensional} Suppose that Assumption \ref{assumptionidentifiability}
holds, and that $\Theta^{\ast}$ is finite. Then almost surely, $\lim_{t\to\infty}\theta(\sigma_{t})$
exists and $\lim_{t\to\infty}\theta(\sigma_{t})\in\Theta^{\ast}$.
\end{prop}
\begin{proof}
See Appendix \ref{pf:proponedimensional}.
\end{proof}
\medskip{}

This proposition, together with Theorem \ref{Theo:Berk}, implies
that the posterior belief $\mu_{t}$ converges almost surely, and
the limit belief is a degenerate belief on some equilibrium model.
When there are multiple equilibrium models, Proposition \ref{proponedimensional}
does not tell us which one will arise as a long-run outcome. To address
this concern, we define \textit{attracting models} as follows.

\medskip{}

\begin{defn}
\label{def:attractingmodel}A model $\theta^{\ast}\in[0,1]$ is \textbf{\textit{\emph{attracting}}}
if there is $\varepsilon>0$ such that
\end{defn}
\begin{itemize}
\item $\theta(\delta_{x})\geq\theta^{\ast}$ for any $\theta\in(\theta^{\ast}-\varepsilon,\theta^{\ast})$
and for any $x\in F(\delta_{\theta})$.
\item $\theta(\delta_{x})\leq\theta^{\ast}$ for any $\theta\in(\theta^{\ast},\theta^{\ast}+\varepsilon)$
and for any $x\in F(\delta_{\theta})$. 
\end{itemize}
\medskip{}

Intuitively, a model $\theta^{\ast}$ is attracting if it is locally
absorbing, in that $\theta(\bm{\sigma}(t))$ moves toward $\theta^{\ast}$
in its neighborhood. Indeed, the first bullet point in the definition
asserts that if $\theta(\bm{\sigma}(t))$ is slightly lower than $\theta^{\ast}$
in the current period $t$, then it will go up, and hence be closer
to $\theta^{\ast}$ at the next instant. Similarly, the second bullet
point in the definition ensures that if $\theta(\bm{\sigma}(t))$
is slightly higher than $\theta^{\ast}$ in the current period $t$,
then it will go down. In Example \ref{exampleonedimensional}, the
equilibrium models $0$ and $2/3$ are attracting, while $1/3$ is
not.

Given an attracting model $\theta^{\ast}$, let $A=\{\sigma\in\triangle F(\delta_{\theta^{\ast}})|\theta(\sigma)=\theta^{\ast}\}$
be the set of equilibria $\sigma$ in which the agent has a degenerate
belief on $\theta^{\ast}$.\footnote{Upper hemi-continuity of $F$ ensures that this set $A$ is non-empty,
which in turn implies that any attracting model is an equilibrium
model. For the special case in which $F(\delta_{\theta^{\ast}})$
contains only one component, this set $A$ is a singleton. Similarly,
even when $F(\delta_{\theta^{\ast}})$ contains only two components,
the set $A$ is a singleton for generic parameters. On the other hand,
when $F(\delta_{\theta^{\ast}})$ contains three or more actions,
the set $A$ is typically continuous.} The following proposition shows that this set $A$ is robustly attracting,
which means that these equilibria should arise as a long-run outcome
at least for some initial prior. Also the proposition shows that the
converse is true, i.e., if a set $A=\{\sigma\in\triangle F(\delta_{\theta^{\ast}})|\theta(\sigma)=\theta^{\ast}\}$
is robustly attracting, then $\theta^{\ast}$ is an attracting model.\medskip{}

\begin{prop}
\label{propequivalence} Under Assumption \ref{assumptionidentifiability},
for each $\theta^{\ast}$, the following properties are equivalent:
\end{prop}
\begin{itemize}
\item[(a)] $\theta^{\ast}$ is attracting.
\item[(b)] The set $A=\{\sigma\in\triangle F(\delta_{\theta^{\ast}})|\theta(\sigma)=\theta^{\ast}\}$
is attracting.
\item[(c)] The set $A$ is robustly attracting.
\end{itemize}
\begin{proof}
See Appendix \ref{pf:propequivalence}.
\end{proof}
\medskip{}

In the same spirit, we define \textit{repelling models} as follows:

\medskip{}

\begin{defn}
\label{def:repellingmodel}A model $\theta^{\ast}\in(0,1)$ is \textbf{\textit{\emph{repelling}}}
if $\theta^{\ast}\neq\theta(\delta_{x})$ for each pure action $x\in F(\delta_{\theta^{\ast}})$
and there is $\varepsilon>0$ such that
\end{defn}
\begin{itemize}
\item $\theta(\delta_{x})\leq\theta^{\ast}-\varepsilon$ for any $\theta\in(\theta^{\ast}-\varepsilon,\theta^{\ast})$
and for any $x\in F(\delta_{\theta})$.
\item $\theta(\delta_{x})\geq\theta^{\ast}+\varepsilon$ for any $\theta\in(\theta^{\ast},\theta^{\ast}+\varepsilon)$
and for any $x\in F(\delta_{\theta})$.
\end{itemize}
\medskip{}

In words, a model $\theta^{\ast}$ is repelling if $\theta(\bm{\sigma}(t))$
moves away from $\theta^{\ast}$ in its neighborhood. Indeed, the
first bullet point implies that if $\theta(\bm{\sigma}(t))$ is slightly
below $\theta^{\ast}$, it will move down further at the next instant.
The second bullet point implies that if $\theta(\bm{\sigma}(t))$
is slightly above $\theta^{\ast}$, it will go up at the next instant.
In Example \ref{exampleonedimensional}, the equilibrium model $\theta=\frac{1}{3}$
is repelling. In the definition above, we consider only interior models
$\theta\in(0,1)$. This is so because whenever an extreme point $\theta=0,1$
is supported by some equilibrium (i.e., there is an equilibrium $\sigma$
such that $\theta(\sigma)=\theta$), there is a pure-strategy equilibrium
$\delta_{x}$ supporting it.

The following proposition shows that if $\theta^{\ast}$ is repelling,
then any equilibrium in which the agent has a degenerate belief on
this model $\theta^{\ast}$ is repelling; hence these equilibria do
not arise as long-run outcomes.

\medskip{}

\begin{prop}
\label{propequivalence2} Under Assumption \ref{assumptionidentifiability},
$\theta^{\ast}\in(0,1)$ is repelling if and only if it is not supported
by a pure equilibrium, there is at least one mixed equilibrium $\sigma^{\ast}$
with $\theta(\sigma^{\ast})=\theta^{\ast}$, and all mixed equilibria
$\sigma^{\ast}$ with $\theta(\sigma^{\ast})=\theta^{\ast}$ are repelling.
\end{prop}
\begin{proof}
See Appendix \ref{pf:propequivalence}.
\end{proof}

\subsection{\label{subsec:Applications}Application: Positively reinforcing beliefs}

We conclude by applying our results to a large class of economically
relevant environments where beliefs are positively reinforcing in
the sense that higher beliefs lead to higher actions which in turn
lead to higher beliefs. Two examples in this class are \citet{esponda2008behavioral}'s
economies with adverse selection, which includes applications to bilateral
trade, insurance markets, auctions, and performance pay, and \citet*{heidhues2018unrealistic}'s
environment of an agent whose overconfidence biases her learning about
a fundamental, which includes applications to delegation, control
in organizations, and public policy choices. In contrast to this work,
we are able to derive results without additional restrictive assumptions
of a technical nature.\footnote{\citet{esponda2008behavioral} does not generally tackle the question
of convergence; this question is tackled by \citet*{heidhues2018unrealistic},
who establish convergence under the assumption that there is a unique
equilibrium and that the distribution of noise is log concave.}

Without loss of generality, we assume the actions are ordered according
to $x_{1}<...<x_{\left|X\right|}$. We then make the following meaningful
economic assumptions:
\begin{itemize}
\item $\Theta=[0,1]$
\item The identifiability conditions (i)-(ii) in Assumption \ref{assumptionidentifiability}
hold; in particular, let $\theta(\sigma)$ denote the closest model
given $\sigma$.
\item Higher actions lead to higher beliefs: $x\mapsto\theta(\delta_{x})$
is an increasing function.
\item Higher beliefs lead to higher actions: Formally, the mapping $\theta\mapsto F(\delta_{\theta})$
is uhc and satisfies $\max F(\delta_{\theta})\leq\min F(\delta_{\theta'})$
for all $\theta'>\theta$.
\end{itemize}
The first assumption requires a one-dimensional space of models, but
it is less restrictive than one might imagine. In applications, $\theta$
usually represents the mean of a continuous random variable. But,
more generally, the assumption allows for non-parametric models whenever
the random variables takes a finite number of values. For example,
in \citet{esponda2008behavioral}'s case of a buyer who does not know
the distribution over a finite number of product values, $v_{1},...,v_{K}$,
a model can be represented by a vector $\phi=(\phi_{1},...,\phi_{K})$,
where $\phi_{j}$ denotes the probability that the value is $v_{j}$.
If the buyer cares about the expected value of the object, we can
work with a one-dimensional space by defining the transformation $\theta=\sum_{j=1}^{K}\phi_{j}v_{j}$.

The second assumption is an identification condition that says that,
no matter the data, there is a unique model that best fits the data.
This is a natural assumption in the class of environments we study.
Whenever it fails in applications, it is natural to refine beliefs
in a manner that the condition holds. In the buyer example, if the
buyer decides to offer a price of zero, then no trade happens. Then
the buyer does not get any feedback about the value of the product
and her beliefs are unrestricted. In this case, it is common to assume
that the buyer's belief at a price of zero is the limit, as price
goes to zero, of her belief at a positive price, where trade does
happen with positive probability.

The third assumption says that the agent's degenerate belief increases
with the action. For example, the higher the price offered by the
buyer, the higher the quality of products that she trades and, therefore,
her belief about the value of the product.

The last assumption says that the optimal action increases with the
agent's (degenerate) belief. For example, the higher the agent's belief
about the value of the object, then the higher the optimal price.
The continuity requirement simply says that the agent must be indifferent
at beliefs where her behavior changes. The left panel of Figure \ref{fig:Example_positivelyreinf}
depicts an example of this property.

By definition, $\theta^{*}$ is an equilibrium model if and only if
$\theta^{*}\in\theta(\Delta F(\delta_{\theta^{*}}))$. As we show
in the proof of the next proposition, the mapping $\theta\mapsto\theta(\Delta F(\delta_{\theta}))$
has the staircase property. An example is depicted in the right panel
of Figure \ref{fig:Example_positivelyreinf}. In the example, there
are three equilibrium models, two of which are attracting ($\theta_{1}^{*}$
and $\theta_{3}^{*}$) and one of which is repelling ($\theta_{2}^{*}$).
The attracting models are associated with a corresponding pure-action
equilibrium ($x_{2}$ and $x_{4}$, respectively) while the repelling
model is associated with a mixed-action equilibrium (a combination
of $x_{3}$ and $x_{4}$). Our previous results imply that the agent's
action converges to a pure-action equilibrium. This is a general feature
of this class of environments, as long as pure-action equilibria are
strict, a property that is generically true.

\begin{figure}
\begin{centering}
\hspace{-2cm}\includegraphics[scale=0.4]{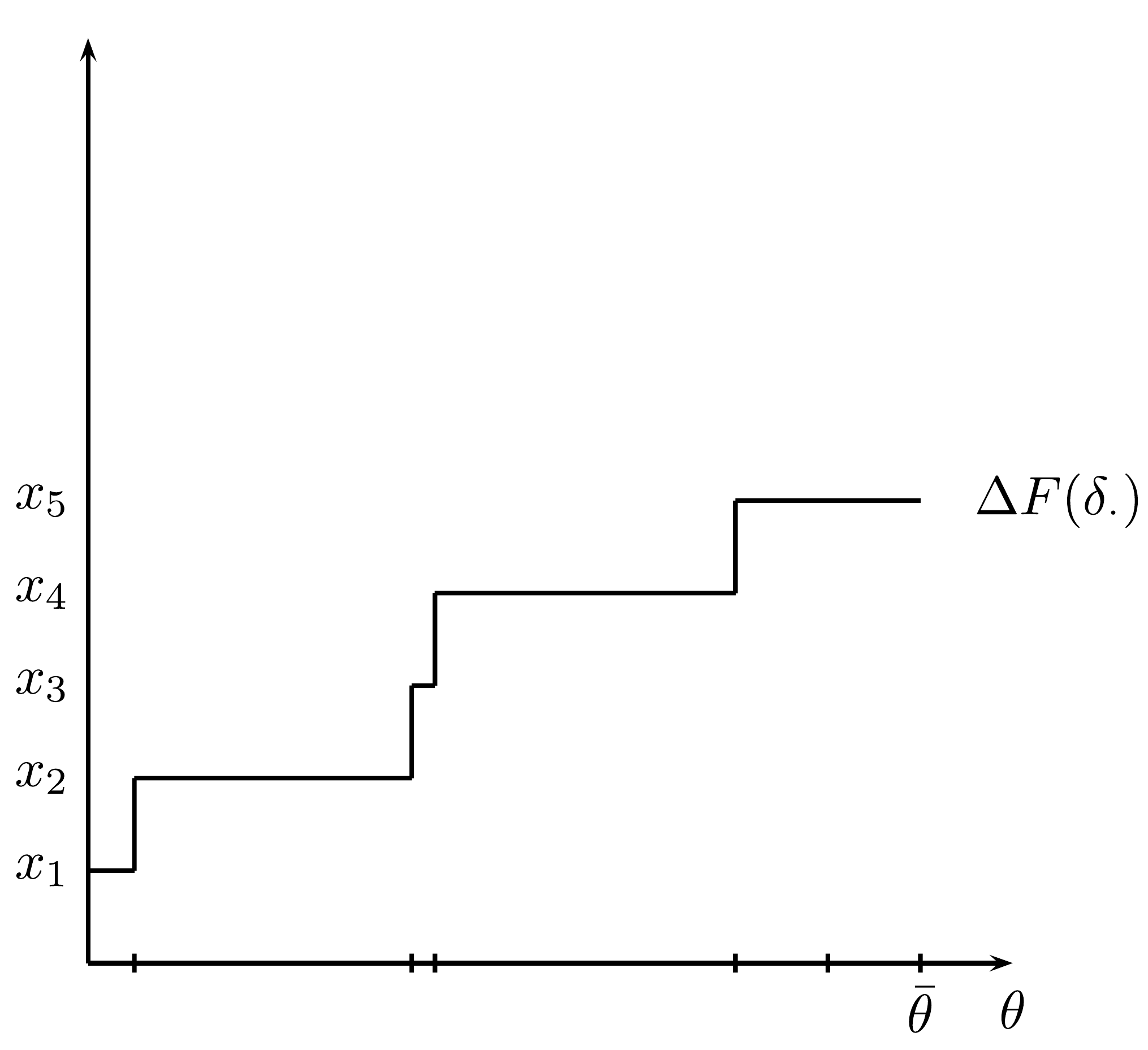}~~~~~\includegraphics[scale=0.4]{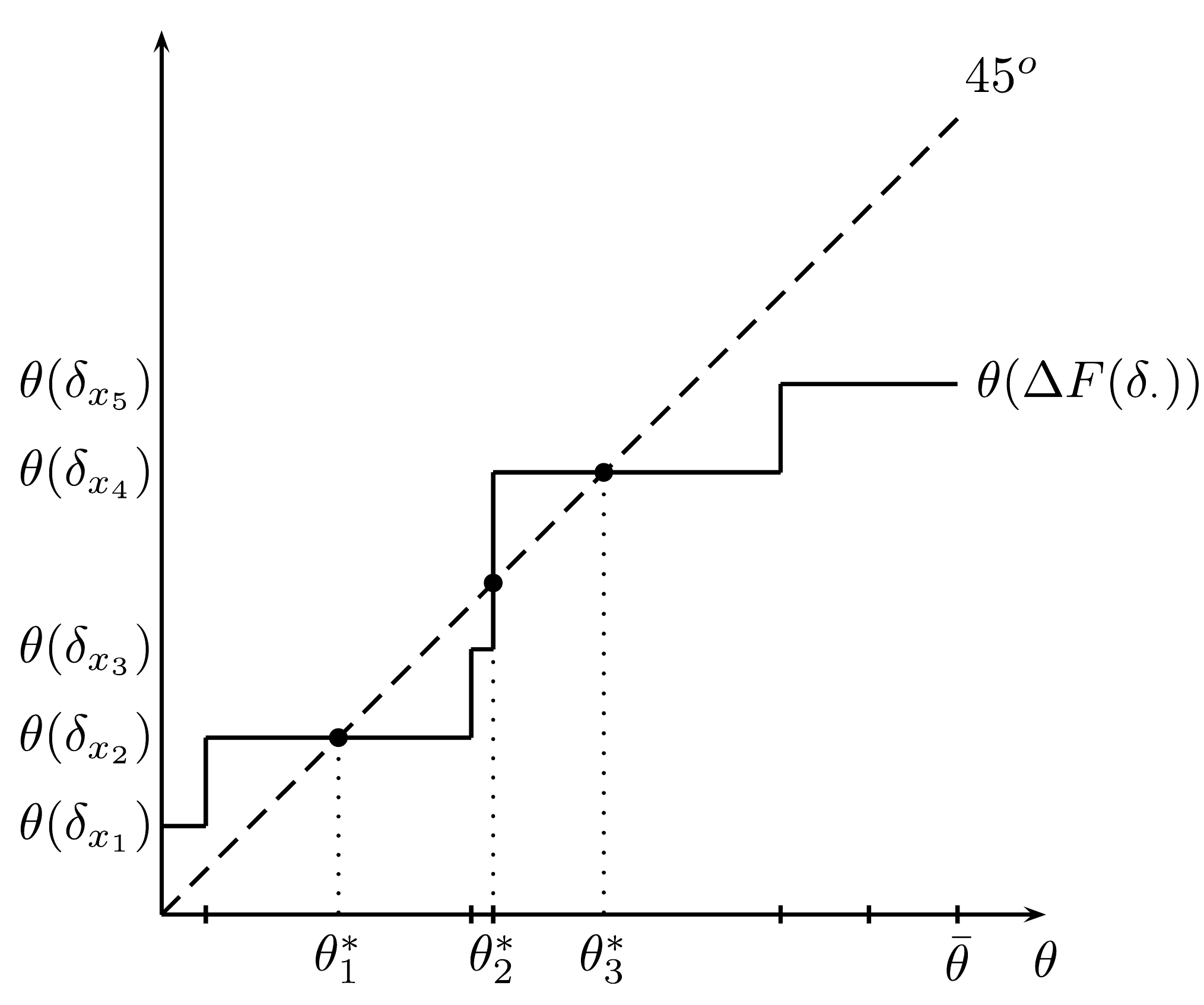}
\par\end{centering}
\caption{\label{fig:Example_positivelyreinf} Example of equilibrium with positively-reinforcing
actions: The left panel shows the correspondence of optimal actions
and the right panel illustrates equilibrium.}
\end{figure}
\medskip{}

\begin{prop}
\label{prop:app_monotone}If all pure-action equilibria are strict,
then the action sequence $(x_{t})_{t}$ almost surely converges to
a pure-action equilibrium.
\end{prop}
\begin{proof}
See Appendix \ref{pf:prop:app_monotone}.
\end{proof}

\section{\label{sec:berk-nash}Relationship to Berk-Nash equilibrium}

In this section, we relate the notion of equilibrium that arises naturally
from the differential inclusion approach to EP2016's definition of
Berk-Nash equilibrium. To facilitate the comparison, we assume that
the agent chooses optimal actions, and denote the correspondence of
optimal actions by $F_{\beta}$, where $\beta\in[0,1)$ denotes the
agent's discount factor ($\beta=0$ is the case of a myopic agent).

The definition of equilibrium that emerges from the differential inclusion
approach is that of a probability distribution over actions satisfying
$\sigma\in\Delta F_{\beta}(\Delta\Theta(\sigma))=\Delta\cup_{\mu\in\Delta\Theta(\sigma)}F_{\beta}(\mu)$
(see Definition \ref{def:equilibrium}). Equivalently, $\sigma$ is
an equilibrium if and only if for every action $x$ in the support
of $\sigma$ there exists a belief $\mu_{x}$ such that $x\in F_{\beta}(\mu_{x})$.
In contrast, EP2016 define a Berk-Nash equilibrium to be a probability
distribution over actions satisfying $\sigma\in\cup_{\mu\in\Delta\Theta(\sigma)}\Delta F_{0}(\mu)$.\footnote{This is the definition of Berk-Nash equilibrium for the single agent
case; EP2016 also consider the case of multiple agents.} Note that $\sigma$ is a Berk-Nash equilibrium if and only if there
exists a belief $\mu$ such that, for every $x$ in the support of
$\sigma$, $x\in F_{0}(\mu)$.

There are two differences between the definition of equilibrium in
this paper and a Berk-Nash equilibrium: A Berk-Nash equilibrium (1)
restricts actions to be supported by the same belief; and (2) requires
actions to be myopically optimal. These two properties are common
in most other standard equilibrium concepts, such as Nash equilibrium.
Following \citet{fudenberg1993self}, the first property is known
as the unitary-belief property, and puts restrictions on the set of
mixed actions that can constitute an equilibrium.\footnote{\citet{fudenberg1993self} showed that non-unitary equilibria make
sense in a game where there are multiple players and, for each player,
there is an underlying population of agents in the role of that player,
and different agents may have different experiences (hence, beliefs)
about other players.} The second property is convenient because myopic optimality is easier
to characterize than general optimality. We will show that both of
these properties are natural provided the following condition holds.

\medskip{}

\begin{defn}
The family of models is \textbf{weakly identified given} $\sigma\in\Delta X$
if $\theta,\theta'\in\Theta(\sigma)$ implies that $Q_{\theta}(\cdot\mid x)=Q_{\theta'}(\cdot\mid x)$
for all $x$ such that $\sigma(x)>0$.\medskip{}

The definition of weak identification was introduced by EP2016. It
says that the belief is uniquely determined along the equilibrium
path, but leaves open the possibility of multiple beliefs for actions
that are not in the support of $\sigma$. Weak identification is immediately
satisfied if the agent's family of models is correctly specified,
but it is also satisfied in many of the applications of misspecified
learning in the literature; see EP2016 for further discussion.\medskip{}
\end{defn}
\begin{prop}
\label{Prop:BerkNash}Suppose that the family of models is weakly
identified given $\sigma$. Then $\Delta\cup_{\mu\in\Delta\Theta(\sigma)}F_{\beta}(\mu)\subseteq\Delta\cup_{\mu\in\Delta\Theta(\sigma)}F_{0}(\mu)=\cup_{\mu\in\Delta\Theta(\sigma)}\Delta F_{0}(\mu)$.
\end{prop}
\begin{proof}
See the Appendix \ref{pf:Prop:BerkNash}.
\end{proof}
\medskip{}

Proposition \ref{Prop:BerkNash} implies that, if the agent is myopic,
then equilibrium and Berk-Nash equilibrium are equivalent concepts
under weak identification. Moreover, if the agent is not myopic, then
the set of equilibria are contained in the set of Berk-Nash equilibria.\footnote{Of course, if the agent is not myopic, then there might be Berk-Nash
equilibria that are not equilibria. This is similar to the idea in
the bandit literature that more patient agents might be willing to
experiment with actions that myopic agents would not.}

We conclude by relating Proposition \ref{prop:converge->eqm} in Section
\ref{sec:Convergence-to-equilibrium} to one of the main results in
EP2016: They show that, if the sequence of distributions over actions
converges, then it converges to a Berk-Nash equilibrium. In our environment
there is no motive for mixing, so convergence of the sequence of distributions
over actions implies that the actions converge. Propositions \ref{prop:converge->eqm}
and \ref{Prop:BerkNash} strengthen EP2016's conclusion by showing
that, under weak identification, even though actions may not converge,
if the action \emph{frequency} converges, then it converges to a Berk-Nash
equilibrium. Of course, the main contribution of this paper is
to go beyond the characterization of equilibrium and to provide tools
to tackle the question of convergence and stability.\newpage{}

\addcontentsline{toc}{section}{References}

\bibliographystyle{aer}
\bibliography{bibtex}

\renewcommand{\thesection}{\Alph{section}}
\setcounter{section}{0}

\section{Appendix}

In this appendix, we present the proofs omitted from the text. In
some places, we use the fact that $\theta\mapsto\log\frac{q(Y|x)}{q_{\theta}(Y|x)}$
is finite and continuous $Q(\cdot|x)-a.s.$ for all $x\in X$. This
fact follows from Assumptions \ref{ass:1}-\ref{ass:2}.

\subsection{\label{pf:Lemma:Theta(sigma)}Proof of Lemma \ref{Lemma:Theta(sigma)}}

Continuity of $K$: For any $(\theta,\sigma)\in\Theta\times\Delta X$
take a sequence $(\theta_{n},\sigma_{n})_{n}$ in $\Theta\times\Delta X$
that converges to this point. By the triangle inequality and the fact
that $K$ is finite under Assumption \ref{ass:2}(iii) it follows
that $\left|K(\theta_{n},\sigma_{n})-K(\theta,\sigma)\right|\leq\left|K(\theta_{n},\sigma)-K(\theta,\sigma)\right|+\left|K(\theta_{n},\sigma_{n})-K(\theta_{n},\sigma)\right|$.

It suffices to show that both terms on the RHS vanish as $n\rightarrow\infty$.
Regarding the first term in the RHS, observe that for any $\sigma\in\Delta X$,
$\theta\mapsto\log\frac{q(Y|X)}{q_{\theta}(Y|X)}$ is finite and continuous
$Q\cdot\sigma-a.s.$ Under Assumption \ref{ass:2}(iii), by the DCT
this implies that $\theta\mapsto K(\theta,\sigma)$ is continuous
for any $\sigma\in\Delta X$. Thus $\lim_{n\rightarrow\infty}\left|K(\theta_{n},\sigma)-K(\theta,\sigma)\right|=0$.
Regarding the other term in the RHS of the display, observe that under
Assumption \ref{ass:2}(iii)
\[
\left|K(\theta_{n},\sigma_{n})-K(\theta_{n},\sigma)\right|\leq\sum_{x\in X}\int g_{x}(y)Q(dy\mid x)|\sigma_{n}(x)-\sigma(x)|
\]
and the RHS vanishes as $\int g_{x}(y)Q(dy\mid x)<\infty$ for all
$x\in X$.

Finally, continuity of $K$, compactness of $\Theta$ (by Assumption
\ref{ass:2}(ii)) and the Theorem of the Maximum imply that $\sigma\mapsto\Theta(\sigma)$
is compact-valued, uhc, and that $\sigma\mapsto K^{\ast}(\sigma)$
is continuous.

\subsection{\label{pf:Lemma:uniformconvergence}Proof of Lemma \ref{Lemma:uniformconvergence}}

Let $(\theta,z)\mapsto g(\theta,z)\equiv\log\frac{q(y\mid x)}{q_{\theta}(y\mid x)}$,
where $z=(y,x)\in Y\times X$. For any $\theta\in\Theta$ and any
$\epsilon>0$, let $O(\theta,\epsilon)\equiv\left\{ \theta'\colon||\theta'-\theta||<\epsilon\right\} $.

$\textsc{Step 1.}$ Pointwise convergence. Fix any $\epsilon>0$ and
any $\theta\in\Theta$. For any $\tau\geq0$ and history $h$, let
\[
\zeta_{\tau}(h)\equiv\sup_{\theta'\in O(\theta,\epsilon)}g(\theta',z_{\tau}(h))-E_{Q(\cdot|x_{\tau}(h))}\left[\sup_{\theta'\in O(\theta,\epsilon)}g(\theta',Y,x_{\tau}(h))\right].
\]

The process $(\zeta_{t})_{t}$ is a Martingale difference under $P^{f}$
and the filtration generated by $\left\{ h^{t}\equiv(x_{0}(h),y_{0}(h),x_{1}(h),y_{1}(h),...,x_{t}(h))\colon t\geq0\right\} $,
because $E_{P^{f}(\cdot|h^{t})}\left[\zeta_{t}(h)\right]=0$ for all
$t$. Define $h\mapsto\zeta^{t}(h)\equiv\sum_{\tau=0}^{t}\left(1+\tau\right)^{-1}\zeta_{\tau}(h)$
for any $t\geq0$. Since $(\zeta_{t})_{t}$ is a Martingale difference
sequence, then $(\zeta^{t})_{t}$ is also a Martingale difference.

By the Martingale Convergence Theorem, there exist a $\mathcal{H}\subseteq\mathbb{H}$
(potentially depending on $\theta\in\Theta$) and $\zeta\in L^{2}(\mathbb{H},\mathbb{R},P^{f})$
such that $P^{f}(\mathcal{H})=1$ and, for any $h\in\mathcal{H}$,
$\zeta^{t}(h)\rightarrow\zeta(h)$, provided $\sup_{t}E_{P^{f}}\left[\left(\zeta^{t}\right)^{2}\right]<\infty$.
This condition is satisfied because
\begin{align*}
E_{P^{f}}\left[\left(\zeta^{t}\right)^{2}\right] & =E_{P^{f}}\left[\sum_{\tau=0}^{t}\left(1+\tau\right)^{-2}\left(\zeta_{\tau}\right)^{2}\right]+2E_{P^{f}}\left[\sum_{\tau>\tau'}\left(1+\tau\right)^{-1}\left(1+\tau'\right)^{-1}\zeta_{\tau}\zeta_{\tau'}\right]\\
 & =\sum_{\tau=0}^{t}\left(1+\tau\right)^{-2}E_{P^{f}}\left[\left(\zeta_{\tau}\right)^{2}\right]\\
 & \leq\sum_{\tau=0}^{t}\left(1+\tau\right)^{-2}E_{P^{f}}\left[\int\left(\sup_{\theta'\in O(\theta,\epsilon)}g(\theta',y,X_{\tau})\right)^{2}Q\left(dy\mid X_{\tau}\right)\right]\\
 & \le C\max_{x\in X}\int\sup_{\theta'\in O(\theta,\epsilon)}\left(g(\theta',y,x)\right)^{2}Q\left(dy\mid x\right),
\end{align*}
where the second line follows from the fact that, for any $\tau>\tau'$,
$E_{P^{f}}\left[\zeta_{\tau}\zeta_{\tau'}\right]=E_{P^{f}}\left[E_{P^{f}(\cdot|h^{\tau})}\left[\zeta_{\tau}\right]\zeta_{\tau'}\right]=0$,
and where the last line follows from the fact that $C\equiv\lim_{t\rightarrow\infty}\sum_{\tau=0}^{t}\left(1+\tau\right)^{-2}<\infty$.
By Assumption \ref{ass:2}(iii), for any $(x,y)\in X\times Y$, $\sup_{\theta'\in O(\theta,\epsilon)}\left(g(\theta',y,x)\right)^{2}\leq\left(g_{x}(y)\right)^{2}$
with $\int\left(g_{x}(y)\right)^{2}Q(dy\mid x)<\infty$. Thus, $\sup_{t}E_{P^{f}}\left[\left(\zeta^{t}\right)^{2}\right]<\infty$.
By invoking Kronecker Lemma it follows that $\lim_{t\rightarrow\infty}\left(1+t\right)^{-1}\sum_{\tau=0}^{t}\zeta^{t}=0$
$P^{f}$-a.s. Therefore, we have established that, for all $\theta\in\Theta$,
\[
\lim_{t\rightarrow\infty}\left(1+t\right)^{-1}\sum_{\tau=0}^{t}\left(\sup_{\theta'\in O(\theta,\epsilon)}g(\theta',z_{\tau})-E_{Q(\cdot|x_{\tau})}\left[\sup_{\theta'\in O(\theta,\epsilon)}g(\theta',Y,x_{\tau})\right]\right)=0\,\,\,\,\,\text{\ensuremath{P^{f}}-a.s.}
\]

$\textsc{Step 2}$. Uniform convergence. Observe that, for any $\epsilon>0$
and any $\theta\in\Theta$, there exists $\delta(\theta,\epsilon)$
such that
\begin{equation}
E_{Q(\cdot|x)}\left[\sup_{\theta'\in O(\theta,\delta(\theta,\epsilon))}g(\theta',Y,x)-g(\theta,Y,x)\right]<0.25\epsilon\label{eq:unifconv-1}
\end{equation}
for all $x\in X$. To see this claim, note that, since $\theta\mapsto g(\theta,Y,x)$
is continuous $Q(\cdot|x)-a.s.$ for all $x\in X$, $\lim_{\delta\rightarrow0}\sup_{\theta'\in O(\theta,\delta)}\left|g(\theta',Y,x)-g(\theta,Y,x)\right|=0$
a.s.$-Q(\cdot\mid x)$ for all $x\in X$. Also, by Assumption \ref{ass:2}(iii),
$\sup_{\theta'\in O(\theta,\delta)}\left|g(\theta',y,x)-g(\theta,y,x)\right|\leq2g_{x}(y)$
and $\int g_{x}(y)Q(dy|x)<\infty$, Thus, by the DCT, $\lim_{\delta\rightarrow0}E_{Q(\cdot|x)}\left[\sup_{\theta'\in O(\theta,\delta)}\left|g(\theta',Y,x)-g(\theta,Y,x)\right|\right]=0$
for all $x\in X$.

Observe that $(O(\theta,\delta(\theta,\epsilon)))_{\theta\in\Theta}$
is an open cover of $\Theta$. By compactness of $\Theta$, there
exists a finite sub-cover $(O(\theta_{j},\delta(\theta_{j},\epsilon)))_{j=1,...J(\epsilon)}$.
Thus, for all $\epsilon>0$, 

\begin{align*}
 & \sup_{\theta\in\Theta}\left|\left(1+t\right)^{-1}\sum_{\tau=0}^{t}\left(g(\theta,z_{\tau})-E_{Q(\cdot|x_{\tau})}\left[g(\theta,Y,x_{\tau})\right]\right)\right|\\
\leq & \max_{j}\sup_{\theta\in O(\theta_{j},\delta(\theta_{j},\epsilon))}\left|\left(1+t\right)^{-1}\sum_{\tau=0}^{t}\left(g(\theta,z_{\tau})-E_{Q(\cdot|x_{\tau})}\left[g(\theta,Y,x_{\tau})\right]\right)\right|\\
\leq & \max_{j}\left(1+t\right)^{-1}\sum_{\tau=0}^{t}\left(\sup_{\theta\in O(\theta_{j},\delta(\theta_{j},\epsilon))}\left|g(\theta,z_{\tau})-E_{Q(\cdot|x_{\tau})}\left[g(\theta,Y,x_{\tau})\right]\right|\right)\\
\leq & \max_{j}\left(1+t\right)^{-1}\sum_{\tau=0}^{t}\left(\left|\sup_{\theta\in O(\theta_{j},\delta(\theta_{j},\epsilon))}g(\theta,z_{\tau})-E_{Q(\cdot|x_{\tau})}\left[\inf_{\theta\in O(\theta_{j},\delta(\theta_{j},\epsilon))}g(\theta,Y,x_{\tau})\right]\right|\right)\\
\leq & \max_{j}\left(1+t\right)^{-1}\sum_{\tau=0}^{t}\left(\left|\sup_{\theta\in O(\theta_{j},\delta(\theta_{j},\epsilon))}g(\theta,z_{\tau})-E_{Q(\cdot|x_{\tau})}\left[\sup_{\theta\in O(\theta_{j},\delta(\theta_{j},\epsilon))}g(\theta,Y,x_{\tau})\right]\right|\right)\\
 & +\max_{j}\left(1+t\right)^{-1}\sum_{\tau=0}^{t}\left(E_{Q(\cdot|x_{\tau})}\left[\sup_{\theta\in O(\theta_{j},\delta(\theta_{j},\epsilon))}g(\theta,Y,x_{\tau})-\inf_{\theta\in O(\theta_{j},\delta(\theta_{j},\epsilon))}g(\theta,Y,x_{\tau})\right]\right)\\
= & I+II.
\end{align*}

By Step 1 and the fact that we are adding over a finite number of
$\theta_{j}$'s, the limit as $t\rightarrow\infty$ of the term $I$
is equal to zero $P^{f}$-a.s. For the second term, note that (\ref{eq:unifconv-1})
implies that 
\[
II\leq2\max_{x\in X}\int\sup_{\theta\in O(\theta_{j},\delta(\theta_{j},\epsilon))}\left|g(\theta,y,x)-g(\theta_{j},y,x)\right|Q(dy\mid x)\leq0.5\epsilon.
\]
Since $0\leq II\leq0.5\epsilon$ holds for all $\epsilon>0$, it follows
that $II=0$. Therefore, using the definition of $g$, we have established
that 
\[
\lim_{t\rightarrow\infty}\sup_{\theta\in\Theta}\left(1+t\right)^{-1}\sum_{\tau=0}^{t}\left(\log\frac{q(y_{\tau}\mid x_{\tau})}{q_{\theta}(y_{\tau}\mid x_{\tau})}-E_{Q(\cdot|x_{\tau})}\left[\log\frac{q(Y\mid x_{\tau})}{q_{\theta}(Y\mid x_{\tau})}\right]\right)=0
\]
$P^{f}$-a.s. The statement in the lemma then follows by noting that
\begin{align*}
K(\theta,\sigma_{t}) & =\sum_{x\in X}E_{Q(\cdot|x)}\left[\log\frac{q(Y\mid x)}{q_{\theta}(Y\mid x)}\right]\sigma_{t}(x)=\left(1+t\right)^{-1}\sum_{\tau=0}^{t}E_{Q(\cdot|x_{\tau})}\left[\log\frac{q(Y\mid x_{\tau})}{q_{\theta}(Y\mid x_{\tau})}\right].
\end{align*}

\subsection{\label{pf:eq:kappa_eps}Proof of equation (\ref{eq:kappa_eps}) in
Theorem \ref{Theo:Berk}}

For simplicity, set $k\equiv\varepsilon/2>0$. Continuity of $(\theta,\sigma)\mapsto\bar{K}(\theta,\sigma)\equiv K(\theta,\sigma)-K^{\ast}(\sigma)$
(see Lemma \ref{Lemma:Theta(sigma)}(i)) and compactness of $\Theta\times\Delta X$
imply that $\bar{K}$ is uniformly continuous.  For any $\sigma$,
take some $\theta_{\sigma}\in\Theta(\sigma)$ (this is possible because
$\Theta(\sigma)$ is nonempty; see Lemma \ref{Lemma:Theta(sigma)}(ii)).
By uniform continuity of $\bar{K}$, there exists $\delta_{k}>0$
such that $\left\Vert \theta_{\sigma}-\theta'\right\Vert <\delta_{k}$
and $\left\Vert \sigma-\sigma'\right\Vert <\delta_{k}$ imply $\bar{K}(\theta',\sigma')<\bar{K}(\theta_{\sigma},\sigma)+k=k$,
where the last equality follows because $\bar{K}(\theta_{\sigma},\sigma)=0$.
This implies that for all $\sigma$, $\{\theta':\left\Vert \theta_{\sigma}-\theta'\right\Vert <\delta_{k}\}\subseteq\{\theta:\bar{K}(\theta,\sigma')\leq k\}$
for all $\sigma'\in B(\sigma,\delta_{k})\equiv\{\sigma':\left\Vert \sigma-\sigma'\right\Vert <\delta_{k}\}$.
Thus, for all $\sigma$, $\mu_{0}(\{\theta:\bar{K}(\theta,\sigma')\leq k\})\geq\mu_{0}(\{\theta':\left\Vert \theta_{\sigma}-\theta'\right\Vert <\delta_{k}\})$
for all $\sigma'\in B(\sigma,\delta_{k})$. The balls $\{B(\sigma,\delta_{k})\}_{\sigma}$
form an open cover for $\Delta X$. Since $\Delta X$ is compact,
there exists a finite subcover $\{B(\sigma^{i},\delta_{k})\}_{i=1}^{n}$.
Let $r\equiv\min_{i\in\{1,...,n\}}\mu_{0}(\{\theta':\left\Vert \theta_{\sigma^{i}}-\theta'\right\Vert <\delta_{k}\})$
which is strictly positive by Assumption \ref{ass:3}. Take any $\sigma'$,
there exists $i$ such that $\sigma'\in B(\sigma^{i},\delta_{k})$;
by the previous argument $\mu_{0}(\{\theta:\bar{K}(\theta,\sigma')\leq k\})\geq\mu_{0}(\{\theta':\left\Vert \theta_{\sigma^{i}}-\theta'\right\Vert <\delta_{k}\})\geq r>0$.

\subsection{\label{pf:Theo:APT}Proof of Theorem \ref{Theo:APT}}

The proof of Theorem \ref{Theo:APT} consists of three parts. Part
1 defines an enlargement of the set of actions that allows us to adopt
the methods developed by BHS2005. Part 2 and 3 correspond to the arguments
in the proofs of Proposition 1.3 and Theorem 4.2 in BHS2005, respectively,
and we provide them here for completeness. Throughout the proof we
fix a history from the set of histories with probability 1 defined
by the statement of Theorem \ref{Theo:Berk}; we omit the history
from the notation. 

\emph{Part 1. Enlargement of the set} $\Delta F(\mu)$. Let $\mathbb{S}=\{a-b\mid a,b\in\Delta X\}$
and let $\Xi:\mathbb{R}_{+}\times\Delta X\rightrightarrows\mathbb{S}$
be defined such that, for all $(\delta,\sigma)\in\mathbb{R}_{+}\times\Delta X$,
\[
\Xi(\delta,\sigma)=\left\{ y\in\mathbb{S}:\begin{array}{c}
\exists\sigma'\in\Delta X,\mu'\in\Delta\Theta\,\,\,\,s.t.\,\,\,\,y\in\Delta F(\mu')-\sigma',\\
\mu'\in M(\delta,\sigma'),\left\Vert \sigma'-\sigma\right\Vert \leq\delta
\end{array}\right\} ,
\]
where $M:\mathbb{R}_{+}\times\Delta X\rightrightarrows\Delta\Theta$
is defined such that, for all $(\delta,\sigma')\in\mathbb{R}_{+}\times\Delta X$,
\[
M(\delta,\sigma')\equiv\{\mu'\in\Delta\Theta:\int_{\Theta}\bar{K}(\theta,\sigma')\mu'(d\theta)\leq\delta\},
\]
where $\bar{K}(\theta,\sigma')\equiv K(\theta,\sigma')-K^{*}(\sigma')$.
Note that $\Theta(0,\sigma)=\Theta(\sigma)$ and so $\Xi(0,\sigma)=\cup_{\mu\in\Delta\Theta(\sigma)}\Delta F(\mu)-\sigma$.
\medskip{}

Claim 1: $(\delta,\sigma)\mapsto\Xi(\delta,\sigma)$ is uhc.

Proof. Because $\mathbb{S}$ is compact, it suffices to show that
$\Xi$ has the closed graph property. For this purpose, we will first
show that $(\delta,\sigma')\mapsto M(\delta,\sigma')$ is uhc. To
establish this claim, note that $\Delta\Theta$ is compact because
of the assumption that $\Theta$ is compact. Hence, we will show that
$M$ has the closed graph property. Take $(\mu'_{n})_{n}$ converging
to $\mu'$ (in the weak topology), $(\delta_{n})_{n}$ converging
to $\delta$, and $(\sigma'_{n})_{n}$ converging to $\sigma'$. Suppose
that $\mu'_{n}\in M(\delta_{n},\sigma_{n}')$ for all $n$. We will
show that $\mu'\in M(\delta,\sigma')$. Since $(\mu'_{n})_{n}$ converges
(weakly) to $\mu'$ and $\bar{K}(\theta,\cdot)$ is continuous (see
Lemma \ref{Lemma:Theta(sigma)}), it follows that
\begin{align*}
\lim_{n}\left(\int_{\Theta}\bar{K}(\theta,\sigma_{n}')\mu_{n}'(d\theta)-\int_{\Theta}\bar{K}(\theta,\sigma')\mu'(d\theta)\right) & =\lim_{n}\left(\int_{\Theta}\bar{K}(\theta,\sigma_{n}')\mu_{n}'(d\theta)-\int_{\Theta}\bar{K}(\theta,\sigma')\mu_{n}'(d\theta)\right)\\
 & \,\,\,\,+\lim_{n}\left(\int_{\Theta}\bar{K}(\theta,\sigma')\mu_{n}'(d\theta)-\int_{\Theta}\bar{K}(\theta,\sigma')\mu'(d\theta)\right)\\
 & =0.
\end{align*}
Also, since $\mu'_{n}\in M(\delta_{n},\sigma_{n}')$, then $\int_{\Theta}\bar{K}(\theta,\sigma_{n}')\mu_{n}'(d\theta)\leq\delta_{n}$.
Taking limits of this last expression on both sides, we obtain $\int_{\Theta}\bar{K}(\theta,\sigma')\mu'(d\theta)\leq\delta$,
implying that $\mu'\in M(\delta,\sigma')$.

Next, to show that $\Xi$ has the closed graph property, take $(y_{n})_{n}$
converging to $y$, $(\delta_{n})_{n}$ converging to $\delta$, and
$(\sigma_{n})_{n}$ converging to $\sigma$. Suppose that $y_{n}\in\Xi(\delta_{n},\sigma_{n})$
for all $n$. We will show that $y\in\Xi(\delta,\sigma)$. Since $y_{n}\in\Xi(\delta_{n},\sigma_{n})$
for all $n$, there exists a sequence $(\mu'_{n},\sigma'_{n})_{n}$
such that $y_{n}\in\Delta F(\mu'_{n})-\sigma'_{n}$, $\left\Vert \sigma_{n}'-\sigma_{n}\right\Vert \leq\delta_{n}$,
and $\mu'_{n}\in M(\delta_{n},\sigma'_{n})$. Because the sequence
$(\mu_{n},\sigma'_{n})_{n}$ lives in a compact set, $\Delta\Theta\times\Delta X$,
there exists a subsequence, $(\mu'_{n(k)},\sigma'_{n(k)})_{k}$ that
converges to $(\mu',\sigma')$. By uhc of $M$ and of $\mu\mapsto\Delta F(\mu)$
(due to the assumption that $F$ is uhc), it follows that $y\in\Delta F(\mu')-\sigma',\left\Vert \sigma'-\sigma\right\Vert \leq\delta$,
and $\mu'\in M(\delta,\sigma')$. Thus, $y\in\Xi(\delta,\sigma)$.\medskip{}

Claim 2: There exists a sequence $(\delta_{t})_{t}$ with $\lim_{t\rightarrow\infty}\delta_{t}=0$
such that, for all $t$, $\sigma_{t+1}-\sigma_{t}\in\frac{1}{t+1}\Xi(\delta_{t},\sigma_{t})$.

Proof. By equation (\ref{eq:system0-1-1-1}) in the text, $\sigma_{t+1}-\sigma_{t}\in\frac{1}{t+1}(\Delta F(\mu_{t+1})-\sigma_{t})$
for all $t$. By Theorem \ref{Theo:Berk}, there exists a sequence
$(\delta_{t})_{t}$ with $\lim_{t\rightarrow\infty}\delta_{t}=0$
such that, for all $t$, $\int_{\Theta}\bar{K}(\theta,\sigma_{t})\mu_{t+1}(d\theta)\leq\delta_{t}$.
Thus, $\Delta F(\mu_{t+1})-\sigma_{t}\subseteq\Xi(\delta_{t},\sigma_{t})$
for all $t$, and the claim follows.\medskip{}

\emph{Part 2. The interpolation of $(\sigma_{t})_{t}$ is what BHS2005
call a perturbed solution of the differential inclusion. }Define $m(t)\equiv\sup\{k\geq0:t\geq\tau_{k}\}$,
where $\tau_{0}=0$ and $\tau_{k}=\sum_{i=1}^{k}1/i$. Let $\mathbf{w}$
be the continuous-time interpolation of \textbf{$(\sigma_{t})_{t}$},
as defined in equation (\ref{eq:interpolation}) in the text. By Claim
2, for any $t$, $\mathbf{w}(t)\in\sigma_{m(t)}+(t-\tau_{m(t)})\Xi(\delta_{m(t)},\sigma_{m(t)})$;
hence, $\dot{\mathbf{w}}(t)\in\Xi(\delta_{m(t)},\sigma_{m(t)})$ for
almost every $t$. Let $\boldsymbol{\gamma}(t)\equiv\delta_{m(t)}+\bigl\Vert\mathbf{w}(t)-\sigma_{m(t)}\bigr\Vert.$
Then $\dot{\mathbf{w}}(t)\in\Xi(\boldsymbol{\gamma}(t),\mathbf{w}(t))$
for almost every $t$. In addition, note that $\lim_{t\rightarrow\infty}\boldsymbol{\gamma}(t)=0$
because $(\delta_{t})_{t}$ goes to zero, $m(t)$ goes to infinity,
and $\boldsymbol{w}$ is the interpolation of \emph{$(\sigma_{t})_{t}$}.\medskip{}

\emph{Part 3. A perturbed solution is an asymptotic pseudotrajectory
(i.e., it satisfies equation (\ref{eq:APTeq}) in the text). }Let
$\boldsymbol{v}(t)\equiv\dot{\mathbf{w}}(t)\in\Xi(\boldsymbol{\gamma}(t),\mathbf{w}(t))$
for almost every $t$. Then
\begin{equation}
\mathbf{w}(t+s)-\mathbf{w}(t)=\int_{0}^{s}\mathbf{v}(t+\tau)d\tau.\label{eq:abscont}
\end{equation}
Since $\mathbb{S}$ is a bounded set, $\boldsymbol{v}$ is uniformly
bounded; therefore, $\mathbf{w}$ is uniformly continuous. Hence,
the family of functions $\left\{ s\mapsto\mathbf{S}^{t}(\mathbf{w})(s)\colon t\in\mathbb{R}\right\} $
\textemdash{} where for each $(t,s)$ $\mathbf{S}^{t}(\mathbf{w})(s)=\mathbf{w}(s+t)$
\textemdash{} is equicontinuous and, therefore, relatively compact
with respect to $L^{\infty}(\mathbb{R},\Delta X,Leb)$, where $Leb$
is the Lebesgue measure; all $L^{p}$ spaces in the proof are with
respect to Lebesgue, so we drop it from subsequent notation. Therefore,
there exists a subsequence $(t_{n})_{n}$ and a $\boldsymbol{w}^{*}\in L^{\infty}(\mathbb{R},\Delta X)$
such that $\mathbf{w}^{*}=\lim_{t_{n}\rightarrow\infty}\mathbf{S}^{t_{n}}(\mathbf{w})$.

Set $t=t_{n}$ in (\ref{eq:abscont}) and define $\mathbf{v}_{n}(s)=\mathbf{v}(t_{n}+s)$.
Then
\[
\mathbf{w^{*}}(s)-\mathbf{w^{*}}(0)=\lim_{n\rightarrow\infty}\int_{0}^{s}\mathbf{v}_{n}(\tau)d\tau.
\]
Since $\mathbf{v}_{n}\in L^{\infty}(\mathbb{R},\mathbb{S})$ for all
$n$, then $\mathbf{v}_{n}\in L^{2}([0,T],\mathbb{S})$. By the Banach-Alouglu
Theorem, there exists a subsequence, which we still denote as $(t_{n})_{n}$,
and a $\mathbf{v}^{*}\in L^{2}([0,T],\mathbb{S})$ such that $(\mathbf{v}_{n})_{n}$
converges in the weak topology to $\mathbf{v}^{*}$; therefore, 
\begin{equation}
\lim_{n\rightarrow\infty}\int_{0}^{s}\mathbf{v}_{n}(\tau)d\tau=\int_{0}^{s}\mathbf{v}^{*}(\tau)d\tau\label{eq:pointwise}
\end{equation}
pointwise in $s\in[0,T]$. Indeed, convergence is uniform because
the family $\left\{ s\mapsto\int_{0}^{s}\mathbf{v}_{n}(\tau)d\tau\colon n\in\mathbb{N}\right\} $
is equicontinuous and $[0,T]$ is compact. In addition, $\mathbf{v}^{*}\in L^{2}([0,T],\mathbb{S})$,
then $\mathbf{w}^{*}$ is absolutely continuous in $[0,T]$.

The proof concludes by showing the claim that $\mathbf{v}^{*}(\tau)\in\Delta F(\Delta\Theta(\mathbf{w}^{*}(\tau)))-\mathbf{w}^{*}(\tau)$
Lebesgue-a.s. in $\tau\in[0,T]$. We will prove it by showing that
$\mathbf{v}^{*}(\tau)\in Co(\Xi(0,\mathbf{w}^{*}(\tau)))$ Lebesgue-a.s.
in $\tau\in[0,T]$, where $Co$ denotes the convex hull; the desired
claim then follows because the facts that $\Delta F(\Delta\Theta(\sigma))-\sigma$
is a convex set and contains $\Xi(0,\sigma)$ and, by definition,
$Co(\Xi(0,\sigma))$ is the smallest convex set that contains $\Xi(0,\sigma)$,
imply that $Co(\Xi(0,\sigma))\subseteq\Delta F(\Delta\Theta(\sigma))-\sigma$.

We will prove $\mathbf{v}^{*}(\tau)\in Co(\Xi(0,\mathbf{w}^{*}(\tau)))$
Lebesgue-a.s., in several steps. First, we show that weak convergence
of $(\mathbf{v}_{n})_{n}$ to $\mathbf{v}^{*}$ implies almost sure
convergence of a weighted average of $(\mathbf{v}_{n})_{n}$ to $\mathbf{v}^{*}$.
Formally, by Mazur's Lemma, for each $n\in\mathbb{N}$, there exists
a $N(n)\in\mathbb{N}$ and a non-negative vector, $(\alpha_{n},...,\alpha_{N(n)})$,
such that $\sum_{i=n}^{N(n)}\alpha_{i}=1$, and $\lim_{n\rightarrow\infty}\left\Vert \bar{\boldsymbol{v}}_{n}-\boldsymbol{v}^{*}\right\Vert _{L^{2}([0,T],\mathbb{S})}=0$
where $\bar{\boldsymbol{v}}_{n}\equiv\sum_{k=n}^{N(n)}\alpha_{k}\boldsymbol{v}_{n}$.
Therefore, as $\lim_{n\rightarrow\infty}\left\Vert \bar{\boldsymbol{v}}_{n}-\boldsymbol{v}^{*}\right\Vert _{L^{2}([0,T],\mathbb{S})}=0$,
it follows that $\lim_{n\rightarrow\infty}\bar{\boldsymbol{v}}_{n}=\boldsymbol{v}^{*}$
a.s.-Lebesgue.

Fix $\tau\in[0,T]$ such that the previous claim holds. Define $\boldsymbol{\gamma}_{n}(\tau)\equiv\gamma(t_{n}+\tau)$
and $\mathbf{w}_{n}(\tau)\equiv\mathbf{w}(t_{n}+\tau)$. By uhc of
$\Xi$ at $(0,\sigma)$ for all $\sigma$ (see Claim 1 in Part 1)
and the facts that $\boldsymbol{\gamma}_{n}(\tau)\rightarrow0$ and
$\mathbf{w}_{n}(\tau)\rightarrow\boldsymbol{w}^{*}(\tau)$, it follows
that, for any $\varepsilon>0$, there exists $N_{\varepsilon}$ such
that, for all $n\geq N_{\varepsilon}$, $\Xi(\boldsymbol{\gamma}_{n}(\tau),\mathbf{w}_{n}(\tau))\subseteq\Xi^{\varepsilon}(0,\mathbf{w}^{*}(\tau))$,
where $\Xi^{\varepsilon}(0,\mathbf{w}^{*}(\tau))\equiv\{y'\in\mathbb{S}:\left\Vert y'-y\right\Vert \leq\varepsilon,y\in\Xi(0,\mathbf{w}^{*}(\tau))\}$.
Recall that $\mathbf{v}_{n}(\tau)\in\Xi(\gamma_{n}(\tau),\mathbf{w}_{n}(\tau))$
for all $n$; therefore, $\mathbf{\bar{v}}_{n}(\tau)\in Co(\Xi^{\varepsilon}(0,\mathbf{w}^{*}(\tau)))$
for all $n\geq N_{\varepsilon}$. Since $Co(\Xi^{\varepsilon}(0,\mathbf{w}^{*}(\tau)))$
is closed and $\lim_{j\rightarrow\infty}\bar{\boldsymbol{v}}_{n}(\tau)=\boldsymbol{v}^{*}(\tau)$,
it follows that $\boldsymbol{v}^{*}(\tau)\in Co(\Xi^{\varepsilon}(0,\mathbf{w}^{*}(\tau)))$.
Since this is true for all $\varepsilon>0$, it follows that $\boldsymbol{v}^{*}(\tau)\in Co(\Xi(0,\mathbf{w}^{*}(\tau)))$.

\subsection{\label{Pf:prop:converge->eqm}Proof of Proposition \ref{prop:converge->eqm}}

Let $\sigma^{\ast}$ be an arbitrary non-equilibrium point. Then there
is a pure action $x$ such that $\sigma^{\ast}(x)>0$ and $x\notin F(\triangle\Theta(\sigma^{\ast}))$.
Choose such $x$. By upper hemi-continuity of $F$ (Assumption \ref{ass:4})
and $\Theta(\cdot)$ (Lemma \ref{Lemma:Theta(sigma)}) it follows
that there exists a $\varepsilon>0$ such that $x\notin F(\triangle\Theta(\sigma))$
for all $\sigma\in B_{\varepsilon}(\sigma^{\ast})$ and such that
$\inf_{\sigma\in B_{\varepsilon}(\sigma^{\ast})}\sigma(x)>0$. Pick
such $\varepsilon>0$. Then there is some $T>0$ such that for any
initial value in this $\varepsilon$-neighborhood, $\bm{\sigma}(0)\in B_{\varepsilon}(\sigma^{\ast})$
and any solution $\bm{\sigma}\in S_{\bm{\sigma}(0)}^{\infty}$ to
the differential inclusion leaves this neighborhood within time $T$,
i.e., we have 
\begin{align}
\lVert\bm{\sigma}(\tau)-\sigma^{\ast}\rVert\geq\varepsilon\label{inequalityproof1}
\end{align}
for some $\tau<T$. Such $T$ exists, because the share of the action
$x$ decreases whenever $\bm{\sigma}(\tau)$ is in the set $B_{\varepsilon}(\sigma^{\ast})$.

Now, pick a sample path $h$ such that the property stated in Theorem
\ref{Theo:APT} holds. We will show that $\sigma_{t}$ cannot stay
in the $\frac{\varepsilon}{2}$-neighborhood of $\sigma^{\ast}$ forever.
This completes the proof, because it implies that almost surely, $\sigma_{t}$
cannot converge to any non-equilibrium point $\sigma^{\ast}$.

Pick $\tilde{T}$ such that for any time $t>\tilde{T}$, 
\begin{align}
\inf_{\bm{\sigma}\in S_{\bm{w}(t)}^{T}}\lVert\bm{w}(t+s)-\bm{\sigma}(s)\rVert<\frac{\varepsilon}{2}\quad\forall s\in[0,T]\label{inequalityproof2}
\end{align}
Suppose there exists a $t>\tilde{T}$ such that $\bm{w}(t)\in B_{\frac{\varepsilon}{2}}(\sigma^{\ast})$
(if no such $t$ exists, then the proof is finished because it follows
that $\sigma_{t}$ is outside a $\varepsilon/2$ neighborhood of $\sigma^{\ast}$
for all $t>\tilde{T}$). Then from (\ref{inequalityproof1}) and
(\ref{inequalityproof2}), there is $s\in[0,T]$ such that $\lVert\bm{w}(t+s)-\sigma^{\ast}\rVert\geq\varepsilon/2$.
So $\sigma_{t}$ cannot stay in the $(\varepsilon/2)$-neighborhood
forever.

\subsection{\label{Pf:Prop:Attracting}Proof of Proposition \ref{Prop:Attracting}}

Part (i) directly follows from part (ii). Proof of part (ii): Pick
a history from the set of histories with probability one defined by
the statement of Theorem \ref{Theo:APT}, and let $\bm{w}$ denote
the interpolation of the action frequency $\sigma_{t}$ given this
path. If there is $t^{\ast}$ such that $\bm{w}(t)\in E$ for all
$t>t^{\ast}$, the result follows. So we will focus on the case in
which for any $t^{\ast}$, there is $t>t^{\ast}$ such that $\bm{w}(t)\notin E$.

Pick attracting sets $(A_{1},\cdots,A_{N})$ as stated. Pick an arbitrarily
small $\varepsilon>0$. Without loss of generality, we assume that
for each attracting set $A_{n}$, the $\varepsilon$-neighborhood
of $A_{n}$ is in the basin of attraction $\mathcal{U}_{A_{n}}$.

Pick $T$ large enough that for any attracting set $A_{n}$, for any
initial value $\bm{\sigma}(0)\in\mathcal{U}_{A_{n}}$, for any $\bm{\sigma}_{\bm{\sigma}(0)}^{2T}$,
and for any $s\in[T,2T]$, 
\begin{align}
d(\bm{\sigma}(s),A_{n})<\frac{\varepsilon}{2}.\label{inequality1}
\end{align}
Also, pick $\tilde{T}$ large enough that for any $t>\tilde{T}$ and
for any $s\in[0,2T]$ 
\begin{align}
\inf_{\bm{\sigma}\in S_{\bm{w}(t)}^{2T}}\lVert\bm{w}(t+s)-\bm{\sigma}(s)\rVert<\frac{\varepsilon}{2}.\label{inequality2}
\end{align}

Recall that for any $t^{\ast}$, there is $t>t^{\ast}$ such that
$\bm{w}(t)\notin E$. This implies that there is $t>\tilde{T}$ and
an attracting set $A_{n}$ such that $\bm{w}(t)\in\mathcal{U}_{A_{n}}$.
Pick such $t$ and $A_{n}$. From (\ref{inequality1}) and (\ref{inequality2}),
we have $d(\bm{w}(t+s),A_{n})<\varepsilon$ for all $s\in[T,2T]$.
This implies that $\bm{w}(t+s)\in\mathcal{U}_{A_{n}}$ for all $s\in[T,2T]$,
so applying the same argument iteratively, we have $d(\bm{w}(t+s),A_{n})<\varepsilon$
for all $s\geq T$, which means that $\bm{w}$ will stay in the $\varepsilon$-neighborhood
of the attracting set $A_{n}$ forever. Since $\varepsilon$ can be
arbitrarily small, $d(\bm{w}(t)-A_{n})$ converges to zero as $t\to\infty$.
(Note that choosing smaller $\varepsilon$ does not influence $A_{n}$.)

\subsection{\label{Pf:Prop:Repelling}Proof of Proposition \ref{Prop:Repelling}}

Let $\sigma^{\ast}$ be a repelling equilibrium, and pick $\mathcal{U}$
and $T$ as in the definition of repelling equilibrium. Pick a history
from the set of histories with probability one defined by the statements
of Theorems \ref{Theo:Berk} and \ref{Theo:APT}. Let $\bm{w}$ denote
the interpolation of the action frequency $\sigma_{t}$ given this
path. It suffices to show that $\bm{w}(t)$ does not converge to $\sigma^{\ast}$
given this history.

Pick a sufficiently small $\varepsilon>0$, so that $2\varepsilon$-neighborhood
of $\sigma^{\ast}$ is a subset of $\mathcal{U}$. Without loss of
generality, we can assume that there is $\eta>0$ such that $F(\mu)\subseteq F(\triangle\Theta(\sigma^{\ast}))$
for any $\mu$ such that $\int(K(\theta,\sigma)-K^{\ast}(\sigma))\mu(d\theta)<\eta$
for some $\sigma\in B_{\varepsilon}(\sigma^{\ast})$. (If necessary,
take $\varepsilon$ small.) Pick such $\eta>0$.

From Theorems \ref{Theo:Berk} and \ref{Theo:APT}, there is $T^{\ast}$
and $\tau^{\ast}$ such that $T^{\ast}=\sum_{i=1}^{\tau^{\ast}}\frac{1}{i}$,
\begin{align}
\int(K(\theta,\sigma_{\tau})-K^{\ast}(\sigma_{\tau}))\mu_{\tau+1}(d\theta)<\eta\label{inequalitytheorem3}
\end{align}
for all $\tau\geq\tau^{\ast}$, and 
\begin{align}
\inf_{\bm{\sigma}\in S_{\bm{w}(t)}^{T}}\sup_{0\leq s\leq T}\lVert\bm{w}(t+s)-\bm{\sigma}(s)\rVert<\varepsilon\label{inequalitytheorem4}
\end{align}
for all $t\geq T^{\ast}$.

Suppose that $\bm{w}(t)$ is in the $\varepsilon$-neighborhood of
$\sigma^{\ast}$ for some $t>T^{\ast}$ such that $t=\sum_{i=1}^{\tau}\frac{1}{i}$
for some $\tau$. We will show that there is $t^{\prime}>0$ such
that $\bm{w}(t+t^{\prime})$ is not in the $\varepsilon$-neighborhood
of $\sigma^{\ast}$. This completes the proof, because it implies
that $\bm{w}$ cannot stay around $\sigma^{*}$ forever. Let $\sigma=\bm{w}(t)$
\textit{\emph{satisfy the condition in the definition of repelling
equilibrium.}} From (\ref{inequalitytheorem3}) and the definition
of $\eta$, the agent chooses some action $x\in F(\triangle\Theta(\sigma^{\ast}))$
in the current period. This means that $\bm{w}(\tilde{t})$ moves
toward $\delta_{x}$ during the time $\tilde{t}\in[t,t+\frac{1}{\tau+1}]$.
Then from the condition in the definition of repelling equilibrium,
there is $\tilde{t}\in[t,t+\frac{1}{\tau+1}]$ such that for any $\bm{\sigma}\in S_{\bm{w}(\tilde{t})}^{\infty}$,
we have $\bm{\sigma}(t)\notin\mathcal{U}$ for some $t\in[0,T]$.
Then as in the previous case, we can show that there is $t^{\prime}\leq T$
such that $\lVert\bm{w}(\tilde{t}+t^{\prime})-\sigma\rVert>\varepsilon$,
as desired.

\subsection{\label{subsec:ex_robust_attract} Attracting Sets Need Not Be Robustly
Attracting}

The agent has three actions, $x_{1}$, $x_{2}$, and $x_{3}$. Given
an action $x_{k}$, a consequence $y$ is randomly drawn from $Y=\mathbb{R}^{3}$
according to the normal distribution $N(e_{k},I)$, so the action
influences the mean of the consequence $y$. However, the agent does
not recognize that the action influences the consequence. Her model
space is the probability simplex $\Theta=\triangle X$, and for each
model $\theta$, she believes that the consequence follows the normal
distribution $N(\theta,I)$. So given a mixture $\sigma\in\triangle X$,
the closest model is $\theta=\sigma$, i.e., $\Theta(\sigma)=\{\sigma\}$
for each $\sigma$.

For each degenerate belief $\delta_{\theta}$, the optimal policy
is given as follows. Consider the model space $\Theta$, and choose
the points $A=(\frac{2}{3},0,\frac{1}{3})$, $B=(\frac{1}{3},\frac{2}{3},0)$,
$C=(0,\frac{1}{3},\frac{2}{3})$, and $\sigma^{\ast}=(\frac{1}{3},\frac{1}{3},\frac{1}{3})$
as in Figure \ref{figurerobustbr}. For each model $\theta$ in the
interior of the triangle $AB\sigma^{\ast}$, $F(\delta_{\theta})=\{x_{2}\}$,
i.e., the optimal policy is $x_{2}$ if the belief puts probability
one on some model $\theta$ in this triangle. Similarly, the optimal
action is $x_{3}$ for the triangle $BC\sigma^{\ast}$, and $x_{1}$
for the triangle $CA\sigma^{\ast}$. For the point $\sigma^{\ast}$
and the models outside the triangle $ABC$, all actions are optimal,
that is, $F(\delta_{\theta})=X$ for these models $\theta$. For all
models on the boundary of the triangles, the optimal policy is chosen
in such a way that $F(\delta_{\theta})$ is upper hemi-continuous
with respect to $\theta$. For example, on the line $A\sigma^{\ast}$,
$F(\delta_{\theta})=\{x_{1},x_{2}\}$.

 \begin{figure}[htbp] 	
\begin{minipage}{0.5\hsize} 		
\begin{center} 			
\includegraphics{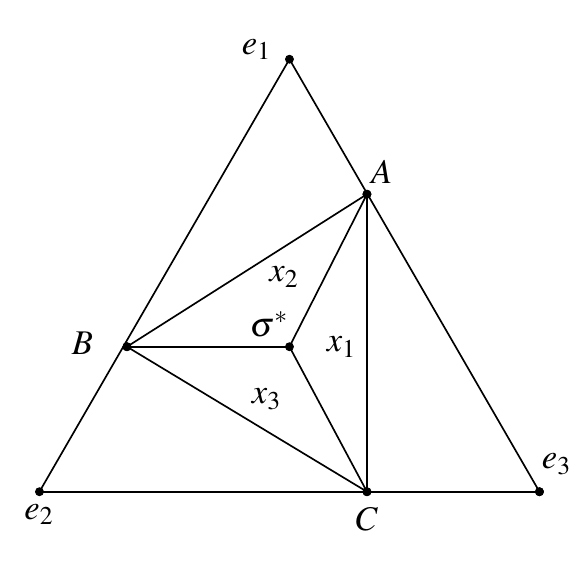} 			
\caption{Policy $F(\delta_{\theta})$ for each model $\theta$}  			
\label{figurerobustbr} 		
\end{center}       	
\end{minipage} 	
\begin{minipage}{0.5\hsize} 		
\begin{center} 			
\includegraphics{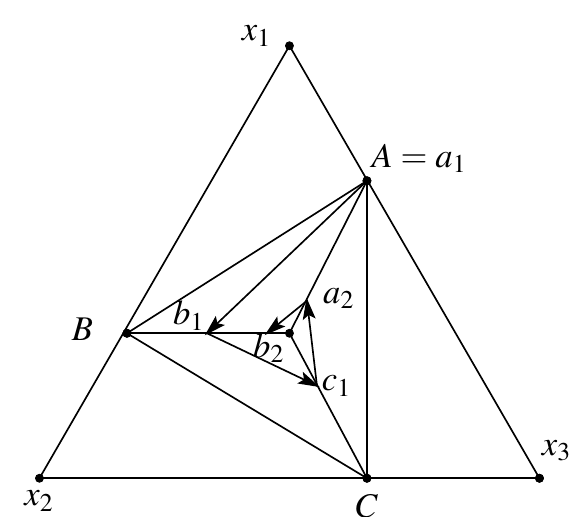} 	
\caption{Path starting from $a_{1}$.}   
\label{figurerobustpath} 		
\end{center} 	
\end{minipage} 
\end{figure}

In this example, the model $\theta=\sigma^{\ast}$ is an attracting
equilibrium, and its basin of attraction is the interior of the triangle
$ABC$. For example, suppose that the action frequency so far is the
point $a_{1}=A$, and the action $x_{2}$ is chosen today. Then the
new action frequency is an interior point of the triangle $AB\sigma^{\ast}$,
and the agent chooses $x_{2}$ until the action frequency hits the
point $b_{1}=(\frac{1}{3},\frac{1}{2},\frac{1}{6})$ on the line $B\sigma^{\ast}$.
After that, the agent chooses the action $x_{3}$ until the action
frequency hits the point $c_{1}=(\frac{2}{9},\frac{1}{3},\frac{4}{9})$
on the line $C\sigma^{\ast}$; then the agent chooses the action $x_{1}$
until the action frequency hits the point $a_{2}=(\frac{5}{12},\frac{1}{4},\frac{1}{3})$.
From there on, the solution to the differential inclusion takes the
path $a_{2}b_{2}c_{2}a_{3}b_{3}c_{3}\cdots$ and converges to $\sigma^{\ast}$,
where 
\begin{align*}
 & a_{n}=(a_{n}^{1},a_{n}^{2},a_{n}^{3})=\left(1-\frac{1}{3}-\frac{1}{9c_{n-1}^{3}},\frac{1}{9c_{n-1}^{3}},\frac{1}{3}\right)\\
 & b_{n}=(b_{n}^{1},b_{n}^{2},b_{n}^{3})=\left(\frac{1}{3},1-\frac{1}{3}-\frac{1}{9a_{n}^{1}},\frac{1}{9a_{n}^{1}}\right)\\
 & c_{n}=(c_{n}^{1},c_{n}^{2},c_{n}^{3})=\left(\frac{1}{9b_{n}^{2}},\frac{1}{3},1-\frac{1}{3}-\frac{1}{9b_{n}^{2}}\right).
\end{align*}
See Figure \ref{figurerobustpath}. Similarly, starting from any interior
points of the triangle $ABC$, any solution $\bm{\sigma}$ to the
differential inclusion will eventually converge to $\sigma^{\ast}$.

Now we will modify this example in such a way that the equilibrium
$\sigma^{\ast}$ is still attracting but not robustly attracting.
Take the points $d_{0}$, $d_{1}$, $\cdots$ as in Figure \ref{figurerobustd},
that is, $d_{0}$ is the intersection point of the line $AB$ and
the line passing through $\sigma^{\ast}$ and $C$, and for each $n\geq1$,
$d_{n}$ is the intersection point of the line $a_{n}b_{n}$ and the
line passing through $\sigma^{\ast}$ and $C$. Then take the sequence
$(z_{0},z_{1},\cdots)$ such that $z_{0}=d_{1}$, $z_{1}=(d_{1}^{1},\frac{d_{1}^{2}+d_{2}^{2}}{2},1-d_{1}^{1}-\frac{d_{1}^{2}+d_{2}^{2}}{2})$,
and $z_{k}=\frac{z_{k-2}+d_{2}}{2}$ for each $k\geq2$. Intuitively,
$z_{0}z_{1}\cdots$ is a ``jagged bridge'' which connects $d_{1}$
and $d_{2}$, whose step size shrinks as it goes. See Figure \ref{figurerobustbridge}.

 \begin{figure}[htbp] 	
\begin{minipage}{0.5\hsize} 		
\begin{center} 			
\includegraphics{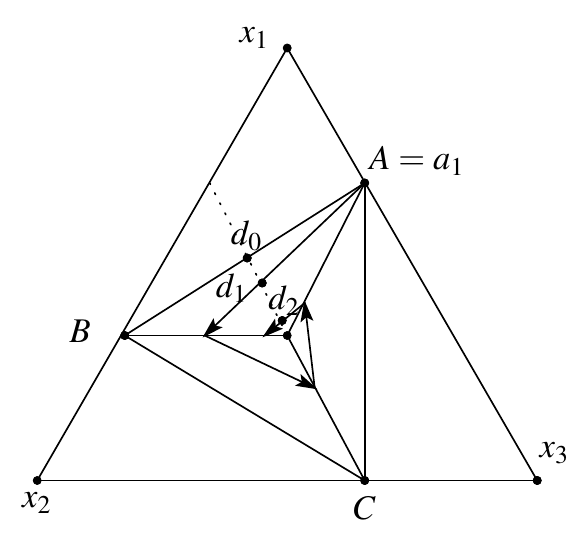}			
\caption{Policy $F(\delta_{\theta})$ for each model $\theta$}  			
\label{figurerobustd} 		
\end{center}       	
\end{minipage} 	
\begin{minipage}{0.5\hsize} 		
\begin{center} 			
\includegraphics{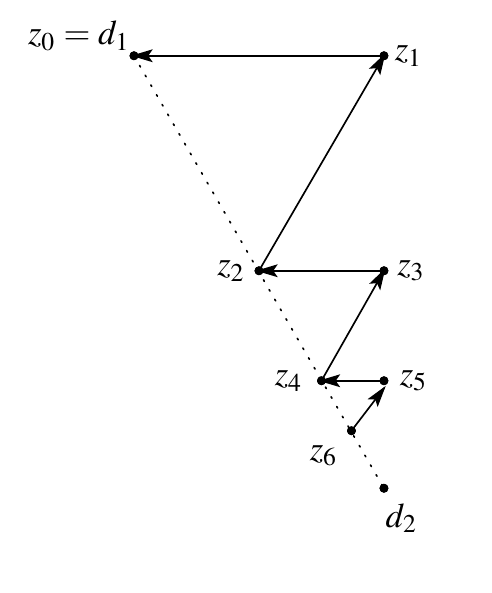}
\caption{Jagged path. It does not reach $d_{2}$.}   
\label{figurerobustbridge} 		
\end{center} 	
\end{minipage} 
\end{figure}

Assume that for each model $\theta$ on this bridge $z_{0}z_{1}z_{2}\cdots$,
the optimal policy is $F(\delta_{\theta})=\{x_{1},x_{2}\}$. Then
starting from any point on this bridge $z_{0}z_{1}\cdots$, a solution
$\bm{\sigma}$ to the differential inclusion can move along this bridge
and reach the point $d_{1}$. However, starting from the point $d_{2}$,
$\bm{\sigma}$ cannot move to $d_{1}$; this is so because for every
large $n$, $z_{n}$ is slightly different from $d_{2}$, which means
that the bridge $z_{0}z_{1}\cdots$ do not reach the point $d_{2}$
exactly. Accordingly, the asymptotic motion of $\bm{\sigma}$ is the
same as before, i.e., as long as the starting point is in the interior
of the triangle $ABC$, $\bm{\sigma}$ converges to $\sigma^{\ast}$.

The same is true even if we add more bridges. Suppose that for each
$n$, there is a jagged path from $d_{n}$ toward $d_{n+1}$. Even
with this change, $\sigma^{\ast}$ is still attracting, for example,
starting from the point $b_{1}$, $\bm{\sigma}$ must follow the path
$b_{1}c_{1}a_{2}b_{2}c_{2}\cdots$ and eventually converge to $\sigma^{\ast}$.

However, adding these bridges significantly changes the solution $\tilde{\bm{\sigma}}$
to the perturbed differential inclusion. Indeed, starting from the
point $d_{n}$, $\tilde{\bm{\sigma}}$ can move to $d_{n-1}$ through
the jagged path, because this path is $\varepsilon$-close to the
point $d_{n}$ for any small $\varepsilon$. For the same reason,
$\tilde{\bm{\sigma}}$ can move to $d_{n-2}$, $d_{n-3}$, $\cdots$,
and can eventually reach the point $d_{0}$, which is outside of the
basin of $\sigma^{\ast}$. This implies that $\sigma^{\ast}$ is not
robustly attracting with these bridges.

\subsection{\label{pf:prop:converge_robust_attracting}Proof of Proposition \ref{prop:converge_robust_attracting}}

We will first present a few preliminary results. We have seen in Lemma
\ref{Lemma:uniformconvergence} that given any initial prior $\mu_{0}$
and given any policy $f$, there is $T$ such that with positive probability,
the consequence frequency is close to the mean (more formally, the
sample average of the likelihood $L_{t}$ is close to the mean) for
\textit{all} periods after $T$. The following claim shows that this
$T$ can be chosen independently of $\mu_{0}$ and $f$. The proof
can be found at the end.

\begin{claim} \label{claim1} For any $\eta>0$, there is $T$ and
$\underline{q}>0$ such that for any initial prior $\mu_{0}$ with
full support and for any $f$, 
\begin{align}
P^{f}(\forall t\geq T\forall\theta\;\;|L_{t}(\theta)-K(\theta,\sigma_{t})|<\eta)>\underline{q}.\label{claiminequality1}
\end{align}
\end{claim}

The next claim just summarizes what we have seen in the proof of Theorem
\ref{Theo:Berk}: It shows that if the past consequence frequency
is close to the mean as stated in the above claim, and if the initial
prior $\mu_{0}$ satisfies some technical condition, then the posterior
belief $\mu_{t+1}$ concentrates on the states which approximately
minimize $K(\theta,\sigma_{t})$ for large $t$.

\begin{claim} \label{claim2} For any $\eta>0$ and for any $\kappa>0$,
there is $T$ such that for any initial prior $\mu_{0}$ and for any
$t>T$ such that $|L_{t}(\theta)-K(\theta,\sigma_{t})|<\frac{\eta}{16}$
and $\mu_{0}(\{\theta:K(\theta,\sigma_{t})-K^{\ast}(\sigma_{t})\leq\frac{\eta}{4}\})\geq\kappa$,
\[
\int(K(\theta,\sigma_{t})-K^{\ast}(\sigma_{t}))\mu_{t+1}(d\theta)<\eta.
\]
\end{claim}

\begin{proof} Directly follows from the proof of Theorem \ref{Theo:Berk}.
\end{proof}

The next claim shows that if the posterior belief $\mu_{t+1}$ is
concentrated as stated in the above claim then the motion of the action
frequency $\sigma_{t}$ is described by the perturbed differential
inclusion. A difference from Theorem \ref{Theo:APT} is that here
the motion of the action frequency $\bm{w}$ exactly matches a solution
to the perturbed differential inclusion. In contrast, in Theorem \ref{Theo:APT},
we take the limit as $t\to\infty$, so a solution to the differential
inclusion is an approximation of the action frequency $\bm{w}$.

\begin{claim} \label{claim3} Let $F$ be an uhc policy correspondence.
Then for any $\varepsilon>0$, there is $\eta>0$ such that given
a sample path $h$, for any $t>\frac{1}{\varepsilon}$ such that $\int(K(\theta,\sigma_{t})-K^{\ast}(\sigma_{t}))\mu_{t}(d\theta)<\eta$,
there is $\bm{\sigma}\in\bm{S}_{\bm{w}(T)}^{\infty,\varepsilon}$
such that $\bm{w}(T+s)=\bm{\sigma}(s)$ for all $s\in[0,\frac{1}{t+1}]$,
where $T=\sum_{\tau=1}^{t}\frac{1}{\tau}$. \end{claim}

\begin{proof} Pick $\varepsilon>0$ arbitrarily. It is sufficient
to show that there is $\eta>0$ such that for any $\sigma$ and for
any $\mu$ such that $\int(K(\theta,\sigma)-K^{\ast}(\sigma))\mu(d\theta)<\eta$,
there is $\tilde{\sigma}\in B_{\frac{\varepsilon}{2}}(\sigma)$ such
that $F(\mu)\subseteq F(\triangle\Theta(\tilde{\sigma}))$.

Note that for each $\sigma$, there is $\varepsilon_{\sigma}<\varepsilon$
and $\eta_{\sigma}$ such that $F(\mu)\subseteq F(\triangle\Theta(\sigma))$
for all $\mu$ such that $\int(K(\theta,\tilde{\sigma})-K^{\ast}(\tilde{\sigma}))\mu(d\theta)<\eta_{\sigma}$
for some $\tilde{\sigma}\in B_{\varepsilon_{\sigma}}(\sigma)$. Since
$\triangle X$ is compact, there is a finite subcover $\{B_{\varepsilon_{\sigma_{1}}}(\sigma_{1}),\cdots,B_{\varepsilon_{\sigma_{M}}}(\sigma_{M})\}$.
Let $\eta=\min_{m}\eta_{\sigma_{m}}>0$. This $\eta$ satisfies the
property we want. Indeed, for any $\sigma$ and $\mu$ such that $\int(K(\theta,\sigma)-K^{\ast}(\sigma))\mu(d\theta)<\eta$,
if we set $\tilde{\sigma}=\sigma_{m}$ such that $\sigma\in B_{\varepsilon_{\sigma_{m}}}(\sigma_{m})$,
we have $F(\mu)\subseteq F(\triangle\Theta(\tilde{\sigma}))$. \end{proof}

The next claim shows that to prove convergence to an attracting set
$A$, it suffices to show that $\sigma_{t}$ visits the basin of $A$
infinitely often with positive probability.

\begin{claim} \label{claim4} Suppose that given an initial prior
$\mu_{0}$ and a policy $f$, $\sigma_{t}$ visits the basin of $A$
infinitely often with positive probability, i.e., $P^{f}(\forall T\exists t>T\;\;\;\;\sigma_{t}\in\mathcal{U}_{A})>0$.
Then $P^{f}(\lim_{t\to\infty}d(\sigma_{t},A)=0)>0$. \end{claim}

\begin{proof} Let $\mathcal{H}$ be the set of all $h$ which satisfies
the property stated in Theorem \ref{Theo:Berk}. Note that $P^{f}(\mathcal{H})=1$.
Also, let $\tilde{\mathcal{H}}$ be the set of all $h$ such that
$\sigma_{t}$ visits the basin of $A$ infinitely often, i.e., it
is the set of all $h$ such that for any $T$, there is $t>T$ such
that $\sigma_{t}\in\mathcal{U}_{A}$. Let $\mathcal{H}^{\ast}=\mathcal{H}\cap\tilde{\mathcal{H}}$.
By the assumption, we have $P^{f}(\mathcal{H}^{\ast})=P^{f}(\tilde{\mathcal{H}})>0$.

Pick a sample path $h\in\mathcal{H}^{\ast}$. To prove the claim,
it suffices to show that $\lim_{t\to\infty}d(\sigma_{t},A)=0$ given
this path. Pick an arbitrary $\varepsilon>0$. Without loss of generality,
we assume that $B_{\varepsilon}(A)$ is in the basin of attraction
$\mathcal{U}_{A}$.

Pick $T$ large enough that (\ref{inequality1}) holds for any initial
value $\bm{\sigma}(0)\in\mathcal{U}_{A}$, for any $\bm{\sigma}\in S_{\bm{\sigma}(0)}^{2T}$,
and for any $s\in[T,2T]$. Also, pick $\tilde{T}$ large enough that
(\ref{inequality2}) holds for any $t>\tilde{T}$ and for any $s\in[0,2T]$.

Since $\sigma_{t}$ visits $\mathcal{U}_{A}$ infinitely often, there
is $t>\tilde{T}$ such that $\bm{w}(t)\in\mathcal{U}_{A}$. Pick such
$t$. Then as in the proof of Proposition \ref{Prop:Attracting}(ii),
we can show that $\bm{w}$ will stay in the $\varepsilon$-neighborhood
of the set $A$ forever. Since $\varepsilon$ can be arbitrarily small,
$\lim_{t\to\infty}d(\bm{w}(t),A)=0$. \end{proof}

Now we will prove the proposition. Let $A$ be a robustly attracting
set, and let $\overline{\zeta}>0$ be such that $B_{\overline{\zeta}}(A)\subset\mathcal{U}_{A}$.
Let $\zeta$ and $\varepsilon$ be as in the definition of robustly
attracting set. Then pick $\eta$ as in Claim \ref{claim3}, pick
an arbitrary $\kappa>0$, and pick $T$ as stated in Claim \ref{claim2}.

Pick $t^{\ast}$ large enough that $\frac{1}{t^{\ast}}<\varepsilon$
and 
\begin{align}
\frac{t^{\ast}}{t^{\ast}+T}\sigma+\frac{T}{t^{\ast}+T}\tilde{\sigma}\in B_{\zeta}(A)\label{claiminequality3}
\end{align}
for all $\sigma\in B_{\frac{\zeta}{2}}(A)$ and $\tilde{\sigma}\in\triangle X$.
Now, consider the following hypothetical situation: 
\begin{itemize}
\item[(a)] The initial prior is $\mu_{0}$ such that $\mu_{0}(\{\theta:K(\theta,\sigma)-K^{\ast}(\sigma)\leq\frac{\eta}{4}\})>\kappa$
for all $\sigma$. The current period is $t^{\ast}+1$. 
\item[(b)] The action frequency in the past is close to $A$, in that $\sigma_{t^{\ast}}\in B_{\frac{\zeta}{2}}(A)$. 
\item[(c)] The past observation is close to the mean, in that $|L_{t^{\ast}}(\theta)-K(\theta,\sigma_{t^{\ast}})|<\frac{\eta}{16}$
for all $\theta$.\footnote{Claim \ref{claim1} ensures that this condition can be satisfied by
some consequence sequence.} 
\end{itemize}
Let $h^{t^{\ast}}$ be a history which satisfies all the properties
above. (Given a policy $f$, the probability of such a history $h^{t^{\ast}}$
may be zero, but this does not affect the following argument.) Let
$\mathcal{H}$ be the set of histories such that the history during
the first $t^{\ast}$ periods is exactly $h^{t^{\ast}}$ and 
\begin{align}
|L_{t^{\ast}+1,t}(\theta)-K(\theta,\sigma_{t^{\ast}+1,t})|<\frac{\eta}{16}\label{claiminequality2}
\end{align}
for all $t\geq t^{\ast}+T$, where $L_{t^{\ast}+1,t}$ is the sample
average of the likelihood from period $t^{\ast}+1$ to period $t$,
and $\sigma_{t^{\ast}+1,t}$ is the action frequency from period $t^{\ast}+1$
to period $t$. From Claim \ref{claim1}, we know $P^{f}(\mathcal{H}|h^{t^{\ast}})>\underline{q}$.

Pick a path $h\in\mathcal{H}$. We claim that given this path, $\sigma_{t}$
never leaves the basin of $A$ after period $t^{\ast}$. The proof
can be found at the end.

\begin{claim} \label{claim5} For each path $h\in\mathcal{H}$, $\sigma_{t}\in\mathcal{U}_{A}$
for all $t>t^{\ast}$. \end{claim}

Let $\mu^{\ast}$ be the posterior belief induced by the initial prior
$\mu_{0}$ and the history $h^{t^{\ast}}$ above.

Now, consider a new game in which the agent's initial prior is $\mu^{\ast}$.
Since the agent's action is determined by the belief, her play in
this new game is exactly the same as her play in the continuation
game induced by the initial prior $\mu_{0}$ and the history $h^{t^{\ast}}$.
So Claim \ref{claim5} implies that in this new game, with positive
probability, the action frequency $\sigma_{t}$ will stay in the basin
$\mathcal{U}_{A}$ in all periods $t>\tilde{T}$, where $\tilde{T}$
is a sufficiently large number. (This is so because the action frequency
$\sigma_{t^{\ast}}$ during the first $t^{\ast}$ periods has almost
no impact on the action frequency $\sigma_{t}$ for large $t$.) Then
Claim \ref{claim4} implies that in this new game, the action frequency
$\sigma_{t}$ converges to $A$ with positive probability.

\subsubsection{Proof of Claim \ref{claim1}}

Let $P^{x}$ denote the probability distribution of the histories
$h=(x_{t},y_{t})_{t=1}^{\infty}$ when the agent chooses $x$ every
period,

\begin{claim} For any $\eta>0$, there is $T$ such that for any
action $x$, 
\[
P^{x}(\forall t\geq T\forall\theta\quad|L_{t}(\theta)-K(\theta,x)|<\eta)>0
\]
\end{claim}

\begin{proof} Pick any $\eta>0$. From Lemma \ref{Lemma:uniformconvergence},
$\lim_{T\to\infty}P^{x}(\forall t\geq T\forall\theta\quad|L_{t}(\theta)-K(\theta,x)|<\eta)=1$.
This implies the result we want. \end{proof}

Now we will prove Claim \ref{claim1}. Let $L_{t}(\theta,x)=\frac{1}{t\sigma_{t}(x)}\sum_{\tau=1}^{t}1_{\{x_{\tau}=x\}}\log\frac{q(y_{t}|x_{t})}{q_{\theta}(y_{t}|x_{t})}$
be the sample average of the likelihood ratio, where the sample is
taken from the periods in which the agent chooses $x$. Note that
we have $L_{t}(\theta)=\sum_{x\in X}\sigma_{t}(x)L_{t}(\theta,x)$.

Pick $\eta>0$ arbitrarily, and pick $T$ as in the above claim. Let
$\mathcal{H}$ be the set of histories $h$ such that $|L_{t}(\theta,x)-K(\theta,x)|<\eta$
for all $x$ and $t$ such that $t\sigma_{t}(x)>T$. Then there is
$\underline{q}>0$ such that $P^{f}(\mathcal{H})>\underline{q}$ for
any initial prior $\mu_{0}$ and any policy $f$.

Pick an arbitrary $h\in\mathcal{H}$, and let $\xi>0$ be such that
$\left|\log\frac{q(y|x)}{q_{\theta}(y|x)}-\log\frac{q(\tilde{y}|x)}{q_{\theta}(\tilde{y}|x)}\right|<\xi$
for all $x$, $\theta$, $y$, and $\tilde{y}$. Then we have 
\[
|L_{t}(\theta,x)-K(\theta,x)|<\left\{ \begin{array}{ll}
\eta & \textrm{if }t\sigma_{t}(x)>T\\
\xi & \textrm{otherwise}
\end{array}\right.
\]
for all $x$, $\theta$, and $t$. This implies that $\sigma_{t}(x)|L_{t}(\theta,x)-K(\theta,x)|<\max\{\eta,\frac{T\xi}{t}\}$.
So for any $t>T^{\ast}\equiv\frac{T\xi}{\varepsilon}$, we have $\sigma_{t}(x)|L_{t}(\theta,x)-K(\theta,x)|<\eta$.
Hence for any $t>T^{\ast}$, 
\begin{align*}
|L_{t}(\theta)-K(\theta,\sigma_{t})|\leq\sum_{x\in X}\sigma_{t}(x)|L_{t}(\theta,x)-K(\theta,x)|<\eta.
\end{align*}
Since $K$ and $T^{\ast}$ are chosen independently of $h\in\mathcal{H}$,
this implies the result we want.

\subsubsection{Proof of Claim \ref{claim5}}

Pick $h$ as stated. From (b) and (\ref{claiminequality3}), we have
$\sigma_{t}\in B_{\zeta}(A)\subseteq\mathcal{U}_{A}$ for all $t\in\{t^{\ast}+1,\cdots,t^{\ast}+T\}$,
regardless of the agent's play during these periods.

So what remains is to show that $\sigma_{t}\in\mathcal{U}_{A}$ for
all $t>t^{\ast}+T$. From (\ref{claiminequality2}), (a), (c), and
Claim \ref{claim2}, we have $\int(K(\theta,\sigma_{t})-K^{\ast}(\sigma_{t}))\mu_{t+1}(d\theta)<\eta$
for all $t\geq t^{\ast}+T$. Then Claim \ref{claim3} implies that
the motion of the action frequency after period $t^{\ast}+T$ is described
by some solution to the $\varepsilon$-perturbed differential inclusion.
Since $\sigma_{t+T^{\ast}}\in B_{\zeta}(A)$ and $\sigma^{\ast}$
is robustly attracting, we have $\sigma_{t}\in B_{\overline{\zeta}}(A)\subseteq\mathcal{U}_{A}$
for all $t\geq t^{\ast}+T$.

\subsection{\label{pf:propmonotone}Proof of Proposition \ref{propmonotone}}

We will start with a useful lemma, which shows that Assumption \ref{assumptionidentifiability}
essentially requires single-peakedness of the Kullback-Leibler divergence
$K(\theta,\delta_{x})$. Let $\underline{\theta}=\min_{\sigma\in\triangle X}\theta(\sigma)$,
and let $\overline{\theta}=\max_{\sigma\in\triangle X}\theta(\sigma)$.
Then we have the following lemma. The proof can be found at the end.

\begin{lemma} \label{lemmasinglepeaked} If Assumption \ref{assumptionidentifiability}
holds, then for each action frequency $\sigma$, the Kullback-Leibler
divergence $K(\theta,\sigma)$ is \textit{single-peaked} with respect
to $\theta$ in $[\underline{\theta},\overline{\theta}]$, that is,
we have $K(\theta,\sigma)>K(\tilde{\theta},\sigma)$ for each $\theta\in[\underline{\theta},\theta(\sigma))$
and $\tilde{\theta}\in(\theta,\theta(\sigma)]$, and $K(\theta,\sigma)<K(\tilde{\theta},\sigma)$
for each $\theta\in[\theta(\sigma),\overline{\theta})$ and $\tilde{\theta}\in(\theta,\overline{\theta}]$.
\end{lemma}

\bigskip{}

\textbf{Proof of Lemma \ref{lemmasinglepeaked}:} We first prove the
following claim:

\begin{claim} \label{claimsinglepeaked} Under Assumption \ref{assumptionidentifiability},
for each $\sigma$ and $\tilde{\sigma}$ such that $\theta(\sigma)>\theta(\tilde{\sigma})$,
$K(\theta,\sigma)$ is strictly decreasing with respect to $\theta$
in $[\theta(\tilde{\sigma}),\theta(\sigma)]$, and $K(\theta,\tilde{\sigma})$
is strictly increasing with respect to $\theta$ in $[\theta(\tilde{\sigma}),\theta(\sigma)]$.
\end{claim}

\begin{proof} Pick $\sigma$ and $\tilde{\sigma}$ as stated. For
each $\beta\in[0,1]$, let $\sigma_{\beta}=\beta\sigma+(1-\beta)\tilde{\sigma}$.

We will prove only that $K(\theta,\sigma)$ is strictly decreasing
with respect to $\theta$ on $[\theta(\tilde{\sigma}),\theta(\sigma)]$.
Suppose not, so that there is $\theta^{\prime}$, $\theta^{\prime\prime}\in[\theta(\tilde{\sigma}),\theta(\sigma)]$
such that $\theta^{\prime}<\theta^{\prime\prime}$ and $K(\theta^{\prime},\sigma)\leq K(\theta^{\prime\prime},\sigma)$.
We consider the following two cases.

\textit{Case 1:} $K(\theta^{\prime},\tilde{\sigma})\leq K(\theta^{\prime\prime},\tilde{\sigma})$.
In this case, $K(\theta^{\prime},\sigma_{\beta})\leq K(\theta^{\prime\prime},\sigma_{\beta})$
for all $\beta$, so $\theta^{\prime\prime}$ cannot be the unique
minimizer of $K(\theta,\sigma_{\beta})$, i.e., $\theta(\sigma_{\beta})\neq\theta^{\prime\prime}$
for all $\beta$. But this is a contradiction, because $\theta(\sigma_{\beta})$
is continuous in $\beta$ and $\theta(\sigma_{0})\leq\theta^{\prime\prime}\leq\theta(\sigma_{1})$.

\textit{Case 2:} $K(\theta^{\prime},\tilde{\sigma})>K(\theta^{\prime\prime},\tilde{\sigma})$.
Let $\beta^{\prime}$ be such that $\theta(\sigma_{\beta^{\prime}})=\theta^{\prime}$.
Then we have $K(\theta^{\prime},\sigma_{\beta^{\prime}})<K(\theta^{\prime\prime},\sigma_{\beta^{\prime}})$,
which is equivalent to 
\[
\beta^{\prime}(K(\theta^{\prime},\sigma)-K(\theta^{\prime\prime},\sigma))<(1-\beta^{\prime})(K(\theta^{\prime\prime},\tilde{\sigma})-K(\theta^{\prime},\tilde{\sigma})).
\]
Then for all $\beta\geq\beta^{\prime}$, 
\[
\beta(K(\theta^{\prime},\sigma)-K(\theta^{\prime\prime},\sigma))<(1-\beta)(K(\theta^{\prime\prime},\tilde{\sigma})-K(\theta^{\prime},\tilde{\sigma})),
\]
which implies $K(\theta^{\prime},\sigma_{\beta})<K(\theta^{\prime\prime},\sigma_{\beta})$.
So $\theta(\sigma_{\beta})\neq\theta^{\prime\prime}$ for all $\beta\geq\beta^{\prime}$.
But this is a contradiction, because $\theta(\sigma_{\beta})$ is
continuous in $\beta$ and $\theta(\sigma_{\beta^{\prime}})<\theta^{\prime\prime}<\theta(\sigma_{1})$.
\end{proof}

Pick an arbitrary $\sigma^{\ast}$. We will show that the Kullback-Leibler
divergence $K(\theta,\sigma^{\ast})$ is single-peaked in $[\underline{\theta},\overline{\theta}]$.
First, consider the case in which $\theta(\sigma^{\ast})=\underline{\theta}$.
Let $\tilde{\sigma}=\sigma^{\ast}$, and let $\sigma$ be such that
$\theta(\sigma)=\overline{\theta}$. Then from the claim above, $K(\theta,\sigma^{\ast})$
is strictly increasing with respect to $\theta$ in $[\underline{\theta},\overline{\theta}]$,
which implies single-peakedness.

Next, consider the case in which $\theta(\sigma^{\ast})<\underline{\theta}$.
Let $\tilde{\sigma}=\sigma^{\ast}$, and let $\sigma$ be such that
$\theta(\sigma)=\overline{\theta}$. Then from the claim above, $K(\theta,\sigma^{\ast})$
is strictly increasing with respect to $\theta$ in $[\theta(\sigma^{\ast}),\overline{\theta}]$.
Similarly, letting $\sigma=\sigma^{\ast}$ and $\tilde{\sigma}$ be
such that $\theta(\tilde{\sigma})=\underline{\theta}$, the claim
above implies that $K(\theta,\sigma^{\ast})$ is strictly decreasing
with respect to $\theta$ in $[\underline{\theta},\theta(\sigma^{\ast})]$.
Hence $K(\theta,\sigma^{\ast})$ is single-peaked.\bigskip{}

\textbf{Proof of Proposition \ref{propmonotone}:}

Part (i): A standard algebra shows that 
\[
K(\theta,\sigma_{\beta})=\beta K(\theta,\sigma)+(1-\beta)K(\theta,\tilde{\sigma})
\]
for each $\theta$. Then the result follows immediately.

Part (ii): We first show that $\theta(\sigma_{\beta})\geq\theta(\tilde{\sigma})$
for all $\beta$. Suppose not so that there is $\beta_{1}\in(0,1)$
such that $\theta(\sigma_{\beta_{1}})<\theta(\tilde{\sigma})$. Then
since $\theta(\sigma_{\beta})$ is continuous in $\beta$ and $\theta(\sigma_{\beta_{1}})<\theta(\tilde{\sigma})<\theta(\sigma_{1})$,
there must be some $\beta_{2}$ such that $\beta_{1}<\beta_{2}<1$
and $\theta(\sigma_{\beta_{2}})=\theta(\tilde{\sigma})$. But then
from part (i), we have $\theta(\sigma_{\beta})=\theta(\tilde{\sigma})$
for all $\beta\in[0,\beta_{2}]$, and in particular $\theta(\sigma_{\beta_{1}})=\theta(\tilde{\sigma})$.
This is a contradiction.

Similarly, we can show that $\theta(\sigma_{\beta})\leq\theta(\sigma)$
for all $\beta$. Taken together, we have $\theta(\sigma_{\beta})\in[\theta(\tilde{\sigma}),\theta(\sigma)]$
for all $\beta$. Now, from Claim \ref{claimsinglepeaked} in the
proof of Lemma \ref{lemmasinglepeaked}, $K(\theta,\sigma)$ has increasing
differences, in that 
\[
\frac{\partial^{2}K(\theta,\sigma)}{\partial\theta\partial\beta}=\frac{\partial K(\theta,\sigma)}{\partial\theta}-\frac{\partial K(\theta,\tilde{\sigma})}{\partial\theta}\geq0.
\]
for all $\beta$ and $\theta\in[\theta(\tilde{\sigma}),\theta(\sigma)]$.
So the monotone selection theorem of Topkis implies the result we
want.

Part (iii): Pick $\beta_{1}$ and $\beta_{2}$ as stated. Let $\theta^{\ast}=\theta(\sigma_{\beta_{1}})$.
This $\theta^{\ast}$ is an interior solution, so it must solve the
first-order condition $\frac{\partial K(\theta^{\ast},\sigma_{\beta_{1}})}{\partial\theta}=0$,
which is equivalent to 
\begin{equation}
\beta_{1}\frac{\partial K(\theta^{\ast},\sigma)}{\partial\theta}+(1-\beta_{1})\frac{\partial K(\theta^{\ast},\tilde{\sigma})}{\partial\theta}=0.\label{foc}
\end{equation}

We claim that each term in the left-hand side is non-zero:

\begin{claim} $\frac{\partial K(\theta^{\ast},\sigma)}{\partial\theta}\neq0$.
\end{claim}

\begin{proof} Suppose not so that $\frac{\partial K(\theta^{\ast},\sigma)}{\partial\theta}=0$.
Then from (\ref{foc}), we have $\frac{\partial K(\theta^{\ast},\tilde{\sigma})}{\partial\theta}=0$,
that is, $\theta^{\ast}$ satisfies the first-order condition for
$\sigma$ and $\tilde{\sigma}$. Then we must have $\frac{\partial^{2}K(\theta^{\ast},\sigma)}{\partial\theta^{2}}\geq0$.
Indeed, if not and $\frac{\partial^{2}K(\theta^{\ast},\sigma)}{\partial\theta^{2}}<0$,
$\theta^{\ast}$ becomes the local maxima for $K(\theta,\sigma)$,
which contradicts with the single-peakedness of $K(\theta,\sigma)$.
Similarly we have $\frac{\partial^{2}K(\theta^{\ast},\tilde{\sigma})}{\partial\theta^{2}}\geq0$

Also, from Assumption \ref{assumptionidentifiability}, we know that
the second-order condition, $\frac{\partial^{2}K(\theta^{\ast},\sigma_{\beta_{1}})}{\partial\theta^{2}}>0$,
is satisfied for $\sigma_{\beta_{1}}$, which is equivalent to 
\[
\beta_{1}\frac{\partial^{2}K(\theta^{\ast},\sigma)}{\partial\theta^{2}}+(1-\beta_{1})\frac{\partial^{2}K(\theta^{\ast},\tilde{\sigma})}{\partial\theta^{2}}>0.
\]
This inequality implies $\frac{\partial^{2}K(\theta^{\ast},\sigma)}{\partial\theta^{2}}>0$
or $\frac{\partial^{2}K(\theta^{\ast},\tilde{\sigma})}{\partial\theta^{2}}>0$.
Suppose for now that $\frac{\partial^{2}K(\theta^{\ast},\sigma)}{\partial\theta^{2}}>0$.
(The argument for the case with $\frac{\partial^{2}K(\theta^{\ast},\tilde{\sigma})}{\partial\theta^{2}}>0$
is symmetric, so we will omit it.) Then since $\frac{\partial^{2}K(\theta^{\ast},\tilde{\sigma})}{\partial\theta^{2}}\geq0$,
we have $\frac{\partial^{2}K(\theta^{\ast},\sigma_{\beta})}{\partial\theta^{2}}>0$
for all $\beta\neq0$. Also, since $\frac{\partial K(\theta^{\ast},\sigma)}{\partial\theta}=\frac{\partial K(\theta^{\ast},\tilde{\sigma})}{\partial\theta}=0$,
we have $\frac{\partial K(\theta^{\ast},\sigma_{\beta})}{\partial\theta}=0$
for all $\beta$. So $\theta^{\ast}$ satisfies both the first-order
and the second-order conditions, which implies that $\theta(\sigma_{\beta})=\theta^{\ast}$
for all $\beta\neq0$. Then since $\theta(\sigma_{\beta})$ is continuous
in $\beta$, we have $\theta(\sigma_{\beta})=\theta^{\ast}$ for all
$\beta\in[0,1]$. But this is a contradiction, because we have $\theta(\tilde{\sigma})<\theta(\sigma)$.
\end{proof}

The above claim and (\ref{foc}) imply that 
\[
\frac{\partial K(\theta^{\ast},\sigma_{\beta_{2}})}{\partial\theta}=\beta_{2}\frac{\partial K(\theta^{\ast},\sigma)}{\partial\theta}+(1-\beta_{2})\frac{\partial K(\theta^{\ast},\tilde{\sigma})}{\partial\theta}\neq0,
\]
which means that $\theta^{\ast}$ cannot be the optimal solution for
$\beta_{2}$. (Note that $\theta^{\ast}$ is an interior value, so
the first-order condition is necessary for it to be optimal.) Then
from part (ii), the result follows.

\subsection{\label{pf:proponedimensional}Proof of Proposition \ref{proponedimensional}}

Let $\Theta^{\ast\ast}$ be the union of the equilibrium models and
the boundary points, that is, $\Theta^{\ast\ast}=\Theta^{\ast}\cup\{0,1\}$.
Since $\Theta^{\ast}$ is finite, it can be written as $\Theta^{\ast*}=\{\theta_{0},\theta_{1},\cdots,\theta_{N}\}$
where $0=\theta_{0}<\cdots<\theta_{N}=1$.

We first show that each interval $(\theta_{n},\theta_{n+1})$ has
a useful property.
\begin{lem}
\label{lemmaconvergence0} Each interval $(\theta_{n},\theta_{n+1})$
must satisfy one of the following properties:
\end{lem}
\begin{itemize}
\item[(i)] For each $\theta\in(\theta_{n},\theta_{n+1})$ and for each $x\in F(\delta_{\theta})$,
we have $\theta(\delta_{x})>\theta$.
\item[(ii)] For each $\theta\in(\theta_{n},\theta_{n+1})$ and for each $x\in F(\delta_{\theta})$,
we have $\theta(\delta_{x})<\theta$. 
\end{itemize}
\begin{proof}
If there is $\theta\in(\theta_{n},\theta_{n+1})$ such that $\theta(\delta_{x})=\theta$
for some $x\in F(\delta_{\theta})$, then this $\theta$ is an equilibrium
model, which is a contradiction. So such $\theta$ does not exist.

Similarly, if there is $\theta\in(\theta_{n},\theta_{n+1})$ such
that $\theta(\delta_{x})<\theta<\theta(\delta_{\tilde{x}})$ for some
$x,\tilde{x}\in F(\delta_{\theta})$, then there is a mixture $\sigma$
of $x$ and $\tilde{x}$ such that $\theta(\sigma)=\theta$, which
implies that $\theta$ is a mixed-strategy equilibrium model. So again
such $\theta$ does not exist.

Accordingly, $(\theta_{n},\theta_{n+1})$ must be the union of the
two sets, $\Theta_{1}$ and $\Theta_{2}$: $\Theta_{1}$ is the set
of all $\theta\in(\theta_{n},\theta_{n+1})$ such that $\theta(\delta_{x})>\theta$
for all $x\in F(\delta_{\theta})$. $\Theta_{2}$ is the set of all
$\theta\in(\theta_{n},\theta_{n+1})$ such that $\theta(\delta_{x})<\theta$
for all $x\in F(\delta_{\theta})$. However, since $F(\delta_{\theta})$
is upper hemi-continuous in $\theta$, one of these sets must be empty.
This implies the result.
\end{proof}
Next, we characterize how the KL minimizer $\theta(\sigma_{t})$ changes
over time, when the motion of $\sigma_{t}$ is determined by the differential
inclusion. Consider an interval $(\theta_{n},\theta_{n+1})$ which
satisfies property (i) in the lemma above. Pick a solution $\bm{\sigma}$
to the differential inclusion, and suppose that $\theta(\bm{\sigma}(t))\in(\theta_{n},\theta_{n+1})$
in the current period $t$. Then from property (i), the agent will
choose an action $x$ such that $\theta(\delta_{x})>\theta(\bm{\sigma}(t))$,
which means that $\theta(\bm{\sigma}(t))$ should move up and eventually
reaches (a neighborhood of) $\theta_{n+1}$. Also, once $\theta(\bm{\sigma}(t))$
goes above $\theta_{n+1}$, it cannot be lower than $\theta_{n+1}$
in any later period. Formally, we have the following result:
\begin{lem}
\label{lemmaconvergence1} Suppose that the interval $(\theta_{n},\theta_{n+1})$
satisfies property (i) stated in Lemma \ref{lemmaconvergence0}. Then
for any $\varepsilon>0$, there is $T>0$ such that given any initial
value $\bm{\sigma}(0)$ with $\theta(\bm{\sigma}(0))>\theta_{n}$
and given any solution $\bm{\sigma}\in S_{\bm{\sigma}(0)}^{\infty}$
to the differential inclusion, we have $\theta(\bm{\sigma}(t))>\theta_{n+1}-\varepsilon$
for all $t\geq T$.
\end{lem}
\begin{proof}
Let $X^{\ast}=\cup_{\theta\in(\theta_{n},\theta_{n+1})}F(\delta_{\theta})$.
We first consider the special case in which $\theta(\delta_{x})\geq\theta_{n+1}$
for all $x\in X^{\ast}$. Then we will explain how to extend the proof
for a general case.

\textit{Case 1:} $\theta(\delta_{x})\geq\theta_{n+1}$ \textit{for
all} $x\in X^{\ast}$. Let $\mathcal{X}$ be the set of all mixed
strategies $\sigma$ such that $\theta(\sigma)\geq\theta_{n+1}$.
From Proposition \ref{propmonotone}(ii), this set is convex. Similarly,
the set $\triangle X\setminus\mathcal{X}$ is convex. So there is
a hyperplane $H$ which separates these two sets; i.e., there is a
vector $\lambda\in\mathbb{R}^{|X|}$ and $k\in\mathbb{R}$ such that
$\lambda\cdot\sigma\geq k$ for all $\sigma$ such that $\theta(\sigma)\geq\theta_{n+1}$,
and $\lambda\cdot\sigma<k$ for all $\sigma$ such that $\theta(\sigma)<\theta_{n+1}$.
From Proposition \ref{propmonotone}(ii), for any $\sigma\in\triangle X^{\ast}$,
we have $\theta(\sigma)\geq\theta_{n+1}$ and hence $\lambda\cdot\sigma\geq k$.

Pick an arbitrary solution $\bm{\sigma}$ to the differential inclusion.
Pick any time $t$ such that $\theta(\bm{\sigma}(t))\in(\theta_{n},\theta_{n+1})$.
Then we have 
\begin{align}
\dot{\bm{\sigma}}(t)=\sigma-\bm{\sigma}(t)\label{equation1}
\end{align}
for some $\sigma\in\triangle X^{\ast}$, and also we have 
\begin{align}
\lambda\cdot\dot{\bm{\sigma}}(t)=\lambda\cdot(\sigma-\bm{\sigma}(t))\geq k-\lambda\cdot\bm{\sigma}(t)>0.\label{equation2}
\end{align}
The first equation (\ref{equation1}), together with Proposition \ref{propmonotone}(ii),
implies that $\theta(\bm{\sigma}(t))$ weakly increases as time goes
for all these $t$. That is, if $\theta(\bm{\sigma}(t))\in(\theta_{n},\theta_{n+1})$
in the current time $t$, then we have $\theta(\bm{\sigma}(t+\eta))\geq\theta(\bm{\sigma}(t))$
at the next instant $t+\eta$. The second equation (\ref{equation2})
implies that $\lambda\cdot\dot{\bm{\sigma}}(t)$ strictly increases
as time goes. So $\bm{\sigma}(t)$ moves toward the hyperplane $H$
if $\theta(\bm{\sigma}(t))\in(\theta_{n},\theta_{n+1})$ in the current
time $t$.

These observations immediately imply the result we want. Pick an arbitrary
$\varepsilon>0$, and let $\tilde{\varepsilon}>0$ be such that $\theta(\sigma)>\theta_{n+1}-\varepsilon$
for all $\sigma$ such that $\lambda\cdot\sigma>k-\tilde{\varepsilon}$.
Pick $T$ large enough that 
\begin{align}
\tilde{\varepsilon}T>k-\lambda\cdot\sigma\label{equation3}
\end{align}
for all $\sigma$. From (\ref{equation2}), if $\lambda\cdot\bm{\sigma}(t)>k-\tilde{\varepsilon}$
in the current period $t$, we have $\lambda\cdot\dot{\bm{\sigma}}(t)\geq\tilde{\varepsilon}$
that is, $\lambda\cdot\bm{\sigma}(t)$ increases at a rate at least
$\tilde{\varepsilon}$. Then from (\ref{equation3}), given any initial
value $\bm{\sigma}(0)$ with $\theta(\bm{\sigma}(0))>\theta_{n}$,
there is $t<T$ such that $\lambda\cdot\bm{\sigma}(t)>k-\tilde{\varepsilon}$,
which implies $\theta(\bm{\sigma}(t))>\theta_{n+1}-\varepsilon$.
Also (\ref{equation1}) implies that after this time $t$, $\theta(\bm{\sigma}(\tilde{t}))$
cannot fall below $\theta_{n+1}-\varepsilon$, that is, $\theta(\bm{\sigma}(\tilde{t}))>\theta_{n+1}-\varepsilon$
for all $\tilde{t}>t$. This implies the result, because $t<T$.

\textit{Case 2:} $\theta(\delta_{x})<\theta_{n+1}$ \textit{for some}
$x\in X^{\ast}$. Let $X^{\ast\ast}=\{x_{1},x_{2},\cdots,x_{M}\}$
denote the set of all $x\in X^{\ast}$ such that $\theta(\delta_{x})<\theta_{n+1}$.
For each action $x_{m}$, let $\xi_{m}$ denote the maximal value
of $\theta\in(\theta_{n},\theta_{n+1})$ such that $x_{m}\in F(\delta_{\theta})$.
Note that the maximum exists, because $F$ is upper hemi-continuous.
Also, by the assumption, we have $\xi_{m}<\theta(\delta_{x_{m}})$.
Without loss of generality, assume that $\theta_{n}<\xi_{1}\leq\xi_{2}\cdots\leq\xi_{M}<\theta_{n+1}$.

Then we can show that there is $T_{1}$ such that given any initial
value $\bm{\sigma}(0)$ with $\theta(\bm{\sigma}(0))\in(\theta_{1},\xi_{1}]$
and given any solution $\bm{\sigma}$ to the differential inclusion,
we have $\theta(\bm{\sigma}(t))>\xi_{1}$ for some time $t<T_{1}$.
The proof is very similar to the argument in the previous case: Let
$\lambda_{1}$ and $k_{1}$ be such that $\lambda_{1}\cdot\sigma\geq k_{1}$
for all $\theta(\sigma)\geq\theta(\delta_{x_{1}})$ and $\lambda_{1}\cdot\sigma<k_{1}$
for all $\theta(\sigma)<\theta(\delta_{x_{1}})$. Then for any $t$
such that $\theta(\bm{\sigma}(t))\in(\theta_{n},\xi_{1}]$, we have
$\dot{\bm{\sigma}}(t)=\sigma-\bm{\sigma}(t)$ for some $\sigma\in\triangle X^{\ast}$,
and also 
\[
\lambda_{1}\cdot\dot{\bm{\sigma}}(t)=\lambda_{1}\cdot(\sigma-\bm{\sigma}(t))>k_{1}-\lambda_{1}\cdot\bm{\sigma}(t)>0.
\]
Note that $k_{1}-\lambda_{1}\cdot\bm{\sigma}(t)$ is bounded away
from zero uniformly in $\bm{\sigma}(t)$ with $\theta(\bm{\sigma}(t))\in(\theta_{n},\xi_{1}]$,
because property (i) in Lemma \ref{lemmaconvergence0} ensures $\theta(\delta_{x_{m}})>\xi_{m}$
for each $m$. This immediately implies the existence of $T_{1}$.

Similarly, there is $T_{2}$ such that given any initial value $\bm{\sigma}(0)$
with $\theta(\bm{\sigma}(0))\in(\xi_{1},\xi_{2}]$ and given any solution
$\bm{\sigma}$ to the differential inclusion, we have $\theta(\bm{\sigma}(t))>\xi_{2}$
for some time $t<T_{2}$. Again the proof is very similar to the argument
in Case 1; the only difference is that here we use the fact that the
action $x_{1}$ is never chosen when $\theta(\bm{\sigma}(t))\in(\xi_{1},\xi_{2}]$.

We iterate this process and define $T_{1}$, $T_{2}$, $\cdots$,
$T_{M}$. Also, pick an arbitrarily small $\varepsilon>0$, and let
$T_{M+1}$ be such that given any initial value $\bm{\sigma}(0)$
with $\theta(\bm{\sigma}(0))\in(\xi_{M},\theta_{n+1})$ and given
any solution $\bm{\sigma}$ to the differential inclusion, we have
$\theta(\bm{\sigma}(t))>\theta_{n+1}-\varepsilon$ for some time $t<T_{M+1}$.
Then let $T=T_{1}+\cdots+T_{M+1}$. This $(\varepsilon,T)$ obviously
satisfies the property stated in the lemma.
\end{proof}
The next lemma relates the result in the previous lemma to the motion
of $\theta(\bm{w}(t))$, where $\bm{w}(t)$ is the actual frequency.
It shows that if $\theta(\bm{w}(t))$ visits the interval $(\theta_{n},\theta_{n+1})$
infinitely often, then after a long time, $\theta(\bm{w}(t))$ cannot
be less than $\theta_{n+1}$. That is, $\theta(\bm{w}(t))$ cannot
move against the solution to the differential inclusion in the long
run.
\begin{lem}
\label{lemmaconvergence2} Consider an interval $(\theta_{n},\theta_{n+1})$
which satisfies property (i) in Lemma \ref{lemmaconvergence0}. Pick
a sample path $h$ such that the property stated in Theorem 3 is satisfied
and such that $\theta(\bm{w}(t))$ exceeds $\theta_{n}$ infinitely
often, i.e., for any $T>0$, there is $t>T$ such that $\theta(\bm{w}(t))>\theta_{n}$.
Then $\liminf_{t\to\infty}\theta(\bm{w}(t))\geq\theta_{n+1}$.
\end{lem}
\begin{proof}
The proof is very similar to that of Proposition \ref{Prop:Attracting}(ii).
Pick $(\theta_{n},\theta_{n+1})$ and $h$ as stated. Pick an arbitrarily
small $\eta>0$. Then pick $\varepsilon>0$ such that $\theta(\sigma)>\theta_{n+1}-\eta$
for all $\sigma$ such that $\lVert\sigma-\tilde{\sigma}\rVert<\varepsilon$
for some $\tilde{\sigma}$ with $\theta(\tilde{\sigma})>\theta_{n+1}-\frac{\eta}{2}$.

From Lemma \ref{lemmaconvergence1}, there is $T>0$ such that given
any initial value $\bm{\sigma}(0)$ with $\theta(\bm{\sigma}(0))>\theta_{n}$
and given any solution $\bm{\sigma}\in S_{\bm{\sigma}(0)}^{\infty}$
to the differential inclusion, 
\begin{align}
\theta(\bm{\sigma}(t))>\theta_{n+1}-\frac{\eta}{2}\label{inequalitylemmaconvergence2}
\end{align}
for all $t\geq T$. Pick such $T$. Also, pick $\tilde{T}$ large
enough that (\ref{inequality2}) holds for any $t>\tilde{T}$ and
for any $s\in[0,2T]$.

By the assumption, there is $t>\tilde{T}$ such that $\theta(\bm{w}(t))>\theta_{n}$.
Pick such $t$. Then from (\ref{inequality2}), (\ref{inequalitylemmaconvergence2}),
and the definition of $\varepsilon$, we have $\theta(\bm{w}(t+s))>\theta_{n+1}-\eta$
for all $s\in[T,2T]$. Applying the same argument again, we obtain
$\theta(\bm{w}(t+s))>\theta_{n+1}-\eta$ for all $s\geq T$, which
implies that $\liminf_{t\to\infty}\theta(\bm{w}(t))\geq\theta_{n+1}-\eta$.
Since $\eta$ can be arbitrarily small, we obtain the result.
\end{proof}
Now we will show that $\theta(\sigma_{t})$ converges almost surely.
Suppose not, so that we have $\liminf_{t\to\infty}\theta(\sigma_{t})<\limsup_{t\to\infty}\theta(\sigma_{t})$
with positive probability. Then there is a path $h$ such that the
property stated in Theorem 3 is satisfied and such that $\liminf_{t\to\infty}\theta(\sigma_{t})<\limsup_{t\to\infty}\theta(\sigma_{t})$.
Pick such $h$.

Let $(\theta_{n},\theta_{n+1})$ be an interval such that the intersection
of the interval and $[\liminf_{t\to\infty}\theta(\sigma_{t}),$ $\limsup_{t\to\infty}\theta(\sigma_{t})]$
is non-empty. Assume for now that this interval satisfies property
(i) stated in lemma \ref{lemmaconvergence0}. By the definition of
$h$, $\theta(\bm{w}(t))$ must exceed $\theta_{n}$ infinitely often,
so Lemma \ref{lemmaconvergence2} implies $\liminf_{t\to\infty}\theta(\bm{w}(t))\geq\theta_{n+1}$.
But this is a contradiction, because it implies that the intersection
of $(\theta_{n},\theta_{n+1})$ and $[\liminf_{t\to\infty}\theta(\sigma_{t}),\limsup_{t\to\infty}\theta(\sigma_{t})]$
is empty.

Likewise, if the interval $(\theta_{n},\theta_{n+1})$ satisfies property
(ii) in Lemma \ref{lemmaconvergence0}, there is a contradiction.
Hence we must have $\liminf_{t\to\infty}\theta(\sigma_{t})=\limsup_{t\to\infty}\theta(\sigma_{t})$
almost surely.

Also, Lemma \ref{lemmaconvergence2} implies that for each interval
$(\theta_{n},\theta_{n+1})$ which satisfies property (i) in Lemma
\ref{lemmaconvergence0}, we have $\lim_{t\to\infty}\theta(\sigma_{t})\in(\theta_{n},\theta_{n+1})$
with zero probability. Obviously the same is true for each interval
$(\theta_{n},\theta_{n+1})$ which satisfies property (ii). Hence
$\lim_{t\to\infty}\theta(\sigma_{t})\in\Theta^{\ast\ast}$ almost
surely.

So for the case in which the boundary points $\{0,1\}$ are equilibrium
models, we have $\lim_{t\to\infty}\theta(\sigma_{t})\in\Theta^{\ast}$.
If $\theta=0$ is not an equilibrium model, then as in Lemma \ref{lemmaconvergence0},
we can show that $\theta(\delta_{x})>\theta$ for any model $\theta\in[\theta_{0},\theta_{1})$
and for any $x\in F(\delta_{\theta})$. Then as in Lemma \ref{lemmaconvergence2},
we can show that if a sample path $h$ satisfies the property stated
in Theorem 3 and $\theta(\bm{w}(t))\in[\theta_{0},\theta_{1})$ infinitely
often, then $\liminf_{t\to\infty}\theta(\bm{w}(t))\geq\theta_{1}$.
This immediately implies that $\sigma_{t}$ converges to $\theta_{0}=0$
with zero probability. Similarly, if $\theta=1$ is not an equilibrium
model, then $\sigma_{t}$ converges to this model with zero probability.
Hence the result follows.

\subsection{\label{pf:propequivalence}Proof of Proposition \ref{propequivalence}}

It is obvious that (c) implies (b). So in this proof, we will show
that (a) implies (c), and (b) implies (a).

\subsubsection{Step 1: (a) implies (c)}

Pick an attracting model $\theta^{\ast}$, and let $A=\{\sigma\in\triangle F(\delta_{\theta^{\ast}})|\theta(\sigma)=\theta^{\ast}\}$.
This set is non-empty, because $F$ is upper hemi-continuous in $\sigma$
and $\theta(\sigma)$ is continuous. We will show that this set $A$
is robustly attracting.

The following notation is useful. Let $\mathcal{X}_{1}$ be the set
of all mixed strategies $\sigma$ such that $\theta(\sigma)<\theta^{\ast}$.
From Proposition \ref{propmonotone}, this set is convex. Similarly,
the set $\triangle X\setminus\mathcal{X}_{1}$ is convex. So there
is a hyperplane $H_{1}$ which separates these two sets; i.e., there
is a vector $\lambda_{1}\in\mathbb{R}^{|X|}$ and $k_{1}$ such that
$\lambda_{1}\cdot\sigma<k_{1}$ for all $\sigma$ such that $\theta(\sigma)<\theta^{\ast}$,
and $\lambda_{1}\cdot\sigma\geq k_{1}$ for all $\sigma$ such that
$\theta(\sigma)\geq\theta^{\ast}$. Similarly, letting $\mathcal{X}_{2}$
be the set of all $\sigma$ such that $\theta(\sigma)>\theta^{\ast}$,
there is a hyperplane $H_{2}$ which separates $\mathcal{X}_{2}$
and $\triangle X\setminus\mathcal{X}_{2}$, i.e., there is a vector
$\lambda_{2}\in\mathbb{R}^{|X|}$ and $k_{2}$ such that $\lambda_{2}\cdot\sigma<k_{2}$
for all $\sigma$ such that $\theta(\sigma)>\theta^{\ast}$, and $\lambda_{2}\cdot\sigma\geq k_{2}$
for all $\sigma$ such that $\theta(\sigma)\leq\theta^{\ast}$. (These
hyperplanes $H_{1}$ and $H_{2}$ may or may not coincide.) Let $\mathcal{X}^{\ast}$
be the set of all $\sigma$ such that $\theta(\sigma)=\theta^{\ast}$.

We first consider the special case in which $F(\delta_{\theta^{\ast}})=X$,
i.e., the agent is indifferent over all actions in the model $\theta^{\ast}$.
In this case, $A=\mathcal{X}^{\ast}$, i.e., the set $A$ is the set
of all mixed actions $\sigma$ with $\theta(\sigma)=\theta^{\ast}$.
Later on, we will explain how to extend the proof technique to the
case with $F(\delta_{\theta^{\ast}})\subset X$.
\begin{flushleft}
\textit{Case 1: $F(\delta_{\theta^{\ast}})=X$}. 
\par\end{flushleft}

\begin{flushleft}
Pick $\varepsilon>0$ as in the definition of attracting models. Then
let $\mathcal{X}_{\varepsilon}$ be the set of all $\sigma$ such
that $|\theta(\sigma)-\theta^{\ast}|<\varepsilon$. We show that this
set $\mathcal{X}_{\varepsilon}$ is (a subset of) the basin of attraction.
That is, given any initial value $\bm{\sigma}(0)\in\mathcal{X}_{\varepsilon}$,
any solution $\bm{\sigma}\in S_{\bm{\sigma}(0)}^{\infty}$ to the
differential inclusion will enter a neighborhood of the set $A=\mathcal{X}^{\ast}$
in finite time and stay there forever.
\par\end{flushleft}

So pick any initial value $\bm{\sigma}(0)\in\mathcal{X}_{\varepsilon}$
and any solution $\bm{\sigma}\in S_{\bm{\sigma}(0)}^{\infty}$. We
first show that this solution $\bm{\sigma}$ never leaves the set
$\mathcal{X}_{\varepsilon}$.
\begin{lem}
\label{lemmaequivalence1} $\bm{\sigma}(t)\in\mathcal{X}_{\varepsilon}$
for all $t$, that is, $|\bm{\sigma}(t)-\theta^{\ast}|<\varepsilon$
for all $t$.
\end{lem}
\begin{proof}
Suppose that $\theta(\bm{\sigma}(t))\in(\theta^{\ast}-\varepsilon,\theta^{\ast})$
for some $t$. Then we have $\dot{\bm{\sigma}}(t)=\sigma-\bm{\sigma}(t)$
for some $\sigma\in\triangle F(\delta_{\theta(\bm{\sigma}(t))})$.
By the definition of $\varepsilon$, we must have $\theta(\sigma)\geq\theta^{\ast}$;
then from Proposition \ref{propmonotone}, at the next instant $t+\eta$,
we have $\theta(\bm{\sigma}(t+\eta))\geq\theta(\bm{\sigma}(t))$,
i.e., $\theta(\bm{\sigma}(t))$ is weakly increasing in $t$ if $\bm{\sigma}(t)\in[\theta^{\ast}-\varepsilon,\theta^{\ast})$.
Similarly, if $\theta(\bm{\sigma}(t))\in(\theta^{\ast},\theta^{\ast}+\varepsilon)$,
then $\theta(\bm{\sigma}(t))$ is weakly decreasing in $t$. This
implies the result we want.
\end{proof}
Next, we show that $A$ is attracting. It suffices to prove the following
lemma:
\begin{lem}
\label{lemmaequivalence2} For any $\varepsilon>0$, there is $T>0$
such that for any initial value $\bm{\sigma}(0)\in\mathcal{X}_{\varepsilon}$
and any solution $\bm{\sigma}\in S_{\bm{\sigma}(0)}^{\infty}$, we
have $d(\bm{\sigma}(t),\mathcal{X}^{\ast})<\varepsilon$ for all $t>T$.
\end{lem}
\begin{proof}
Suppose that $\theta(\bm{\sigma}(t))\in(\theta^{\ast}-\varepsilon,\theta^{\ast})$
for some $t$. Then as shown in the proof of the previous lemma, we
have $\dot{\bm{\sigma}}(t)=\sigma-\bm{\sigma}(t)$ for some $\sigma$
such that $\theta(\sigma)\geq\theta^{\ast}$. This in turn implies
that 
\[
\lambda_{1}\cdot\dot{\bm{\sigma}}(t)=\lambda_{1}\cdot(\sigma-\bm{\sigma}(t))\geq k_{1}-\lambda_{1}\cdot\bm{\sigma}(t)>0.
\]
Here the weak inequality follows from $\theta(\sigma)\geq\theta^{\ast}$,
and the strict inequality follows from $\theta(\sigma)<\theta^{\ast}$.
Note that $k_{1}-\lambda_{1}\cdot\bm{\sigma}(t)$ measures the current
distance from $\bm{\sigma}(t)$ to the hyperplane $H_{1}$, and $\lambda_{1}\cdot\dot{\bm{\sigma}}(t)$
measures how much $\bm{\sigma}(t)$ gets closer to the hyperplane
$H_{1}$ at the next instant, So the equation above implies that $\bm{\sigma}(t)$
gets closer to $H_{1}$ as time goes, and the speed of convergence
is bounded away from zero until $\bm{\sigma}(t)$ enters a neighborhood
of $H_{1}$.

Similarly, if $\theta(\bm{\sigma}(t))\in(\theta^{\ast},\theta^{\ast}+\varepsilon)$
for some $t$, then $\bm{\sigma}(t)$ gets closer to the hyperplane
$H_{2}$ as time goes, and the speed of convergence is bounded away
from zero until $\bm{\sigma}(t)$ enters a neighborhood of $H_{2}$.
This implies the result we want, because the set $A=\mathcal{X}^{\ast}$
is the space sandwiched by $H_{1}$ and $H_{2}$ (formally, $A=\triangle X\setminus(\mathcal{X}_{1})\cup\mathcal{X}_{2})$).
\end{proof}
As a last step, we show that the set $A$ is robustly attracting:
\begin{lem}
\label{lemmaequivalence3} The set $A$ is robustly attracting.
\end{lem}
\begin{proof}
Let $H_{1}^{\prime}$ be the set of all $\sigma$ with $\theta(\sigma)=\theta^{\ast}-\frac{\varepsilon}{2}$,
and $H_{2}^{\prime}$ be the set of all $\sigma$ with $\theta(\sigma)=\theta^{\ast}+\frac{\varepsilon}{2}$.
Take a small $\varepsilon^{\ast}>0$ such that $\theta(\sigma)\in(\theta^{\ast}-\varepsilon,\theta^{\ast})$for
all $\sigma$ with $d(\sigma,H_{1}^{\prime})<\varepsilon^{\ast}$,
and such that $\theta(\sigma)\in(\theta^{\ast},\theta^{\ast}+\varepsilon)$for
all $\sigma$ with $d(\sigma,H_{2}^{\prime})<\varepsilon^{\ast}$.
Note that such $\varepsilon^{\ast}$ exists because $H_{1}^{\prime}$,
$H_{2}^{\prime}$, $\mathcal{X}^{\ast}$, $\{\sigma|\theta(\sigma)=\theta^{\ast}-\varepsilon\}$,
and $\{\sigma|\theta(\sigma)=\theta^{\ast}+\varepsilon\}$ are all
compact and disjoint.

Consider any solution to the $\varepsilon^{\ast}$-perturbed differential
inclusion, and suppose that $\bm{\sigma}(t)\in H_{1}^{\prime}$ for
some $t$, i.e., suppose that $\theta(\bm{\sigma}(t))=\theta^{\ast}-\frac{\varepsilon}{2}$.
Then by the definition of $\varepsilon^{\ast}$, $\dot{\bm{\sigma}}(t)=\sigma-\bm{\sigma}(t)$
for some $\sigma$ such that $\theta(\sigma)\geq\theta^{\ast}$, which
implies that $\theta(\bm{\sigma}(t))$ moves up at the next instant.
Likewise, if $\theta(\bm{\sigma}(t))=\theta^{\ast}+\frac{\varepsilon}{2}$
for some $t$, then $\theta(\bm{\sigma}(t))$ moves down at the next
instant. Accordingly, if the initial value is in the set $\{\sigma|\theta^{\ast}-\frac{\varepsilon}{2}\leq\theta(\sigma)\leq\frac{\varepsilon}{2}\}$,
any solution to the $\varepsilon^{\ast}$-perturbed differential inclusion
cannot leave this set. This implies the result.
\end{proof}
\begin{flushleft}
\textit{ Case 2: $F(\delta_{\theta^{\ast}})\subset X$}.
\par\end{flushleft}

Pick small $\varepsilon>0$ as stated in the the definition of attracting
models. Without loss of generality, we assume that $F(\delta_{\tilde{\theta}})\subseteq F(\delta_{\theta^{\ast}})$
for all $\theta$ such that $|\theta-\theta^{\ast}|<\varepsilon$.
(Take $\varepsilon$ small, if necessary.)

Let $\mathcal{X}_{\varepsilon}$ be as in the previous case. We show
that this set $\mathcal{X}_{\varepsilon}$ is (a subset of) the basin
of attraction. That is, given any initial value $\bm{\sigma}(0)\in\mathcal{X}_{\varepsilon}$,
any solution $\bm{\sigma}\in S_{\bm{\sigma}(0)}^{\infty}$ to the
differential inclusion will enter a neighborhood of the set $A$ in
finite time and stay there forever. Note that now the set $A$ is
a strict subset of $\mathcal{X}^{\ast}$.

Pick any initial value $\bm{\sigma}(0)\in\mathcal{X}_{\varepsilon}$
and any solution $\bm{\sigma}\in S_{\bm{\sigma}(0)}^{\infty}$. Then
Lemma \ref{lemmaequivalence1} still holds, that is, $\bm{\sigma}(t)$
never leaves the set $\mathcal{X}_{\varepsilon}$. Also Lemma \ref{lemmaequivalence2}
still holds, that is, $\bm{\sigma}(t)$ moves toward to the set $\mathcal{X}^{\ast}$
as time goes.

Also, by the definition of $\varepsilon$, we have $F(\triangle\Theta(\sigma))\subseteq F(\delta_{\theta^{\ast}})$
for any $\sigma\in\mathcal{X}_{\varepsilon}$. This implies that at
every time $t$, we have $\dot{\bm{\sigma}}(t)[x]=-\bm{\sigma}(t)[x]$
for each $x\notin F(\delta_{\theta^{\ast}})$. This implies that $\bm{\sigma}(t)$
assigns probability zero on any action $x\notin F(\delta_{\theta^{\ast}})$
in the limit as $t\to\infty$, and in particular, for any $\varepsilon>0$,
there is $T$ such that $\dot{\bm{\sigma}}(t)[x]=-\bm{\sigma}(t)[x]$
for all $x\notin F(\delta_{\theta^{\ast}})$ and $t>T$. This and
Lemma \ref{lemmaequivalence2} imply that the set $A$ is an attractor.

Also, we can show that the set $A$ is robustly attracting; the proof
is very similar to that of Lemma \ref{lemmaequivalence3}, and hence
omitted.

\subsubsection{Step 2: (b) implies (a)}

Pick an arbitrary $\theta^{\ast}$, and let $A=\{\sigma\in\triangle F(\delta_{\theta^{\ast}})|\theta(\sigma)=\theta^{\ast}\}$.
Suppose that $A$ is an attractor. We will show that the model $\theta^{\ast}$
is attracting.

Let $\mathcal{U}_{A}$ be the basin of the set $A$, and let $\bm{\sigma}(0)$
be such that $\theta(\bm{\sigma}(0))<\theta^{\ast}$. Then we have
the following lemma:
\begin{lem}
For any $\theta\in(\theta(\bm{\sigma}(0)),\theta^{\ast})$ and for
any $\sigma\in F(\delta_{\theta})$, we have $\theta(\sigma)>\theta$.
\end{lem}
\begin{proof}
Suppose not, so that there is $\overline{\theta}\in(\theta(\bm{\sigma}(0),\theta^{\ast})$
and $\overline{\sigma}\in F(\delta_{\theta})$ such that $\theta(\overline{\sigma})\leq\overline{\theta}$.
Consider a solution to the differential inclusion $\bm{\sigma}\in S_{\bm{\sigma}(0)}^{\infty}$
such that for any time $t$ such that $\theta(\bm{\sigma}(t))=\overline{\theta}$,
we have $\dot{\bm{\sigma}}(t)=\overline{\sigma}-\bm{\sigma}(t)$.
Then $\bm{\sigma}(t)\leq\overline{\theta}$ for all $t$; by the definition
of $\overline{\sigma}$, $\theta(\bm{\sigma}(t))$ must go down whenever
it hits $\theta(\bm{\sigma}(t))=\overline{\theta}$. This contradicts
with the fact that $\bm{\sigma}(0)$ is in the basin of attraction.
\end{proof}
Since $F(\delta_{\theta})$ is upper hemi-continuous in $\theta$,
there is $\varepsilon>0$ such that $F(\delta_{\theta})=F(\delta_{\tilde{\theta}})$
for all $\theta,\tilde{\theta}\in(\theta^{\ast}-\varepsilon,\theta^{\ast})$.
Then the lemma above implies that for any $\theta\in(\theta^{\ast}-\varepsilon,\theta^{\ast})$
and for any $\sigma\in F(\delta_{\theta})$, we have $\theta(\sigma)\geq\theta^{\ast}$.

Similarly, we can show that there is $\tilde{\varepsilon}>0$ such
that for any $\theta\in(\theta^{\ast},\theta^{\ast}+\tilde{\varepsilon})$
and for any $\sigma\in F(\delta_{\theta})$, we have $\theta(\sigma)\leq\theta^{\ast}$.
Hence $\theta^{\ast}$ is attracting.

\subsection{\label{pf:propequivalence2}Proof of Proposition \ref{propequivalence2}}

\textbf{Only if:} Suppose that a model $\theta^{\ast}\in(0,1)$ is
repelling. Then upper hemi-continuity of $F$ implies that there are
pure actions $x$ and $\tilde{x}$ such that $\theta(\delta_{x})<\theta^{\ast}<\theta(\delta_{\tilde{x}})$
and and $x,\tilde{x}\in F(\delta_{\theta^{\ast}})$. Then from Proposition
\ref{propmonotone}, there is a mixture $\sigma^{\ast}$ of these
actions $x$ and $\tilde{x}$ such that $\theta(\sigma^{\ast})=\theta^{\ast}$.
Obviously this $\sigma^{\ast}$ is a mixed equilibrium with $\theta(\sigma^{\ast})=\theta^{\ast}$.
So it suffices to show that all mixed equilibria with $\theta(\sigma^{\ast})=\theta^{\ast}$
are repelling.

Choose $\varepsilon>0$ as stated in the definition of repelling models.
Then as in the proof of Proposition \ref{propequivalence} there is
a hyperplane which separates mixed actions $\sigma$ with $\theta(\sigma)\geq\theta^{\ast}+\varepsilon$
from others. That is, there is $\lambda_{1}\in\mathbb{R}^{|X|}$ and
$k_{1}\in\mathbb{R}$ such that $\lambda_{1}\cdot\sigma\geq k_{1}$
if and only if $\theta(\sigma)\geq\theta^{\ast}+\varepsilon$. Likewise,
there is $\lambda_{2}$ and $k_{2}$ such that $\lambda_{2}\cdot\sigma\geq k_{2}$
if and only if $\theta(\sigma)\leq\theta^{\ast}-\varepsilon$.

Let $\mathcal{U}$ be the set of all $\sigma$ such that $\theta(\sigma)\in(\theta^{\ast}-\frac{\varepsilon}{2},\theta^{\ast}+\frac{\varepsilon}{2})$.
Also, choose $T$ sufficiently large so that 
\begin{align}
(k_{1}-\lambda_{1}\cdot\sigma)T>k_{1}-\lambda_{1}\cdot\tilde{\sigma}\label{t1}
\end{align}
for all $\sigma$ with $\theta(\sigma)\in(\theta^{\ast},\theta^{\ast}+\frac{\varepsilon}{2})$
and for all $\tilde{\sigma}$ with $\theta(\tilde{\sigma})\in(\theta^{\ast},\theta^{\ast}+\frac{\varepsilon}{2})$,
and that 
\begin{align}
(k_{2}-\lambda_{2}\cdot\sigma)T>k_{2}-\lambda_{2}\cdot\tilde{\sigma}\label{t2}
\end{align}
for all $\sigma$ with $\theta(\sigma)\in(\theta^{\ast}-\frac{\varepsilon}{2},\theta^{\ast})$
and for all $\tilde{\sigma}$ with $\theta(\tilde{\sigma})\in(\theta^{\ast}-\frac{\varepsilon}{2},\theta^{\ast})$.

We will show that these $\mathcal{U}$ and $T$ satisfy the property
stated in the definition of repelling equilibria. This completes the
proof, because any mixed equilibrium $\sigma^{*}$ with $\theta(\sigma^{\ast})=\theta^{\ast}$
is in the interior of $\mathcal{U}$. 

The following result is useful: 

\begin{claim} For any initial point $\sigma\in\mathcal{U}$ with
$\theta(\sigma)\neq\theta^{\ast}$ and for any solution $\bm{\sigma}\in S_{\sigma}^{\infty}$
to the differential inclusion, there is $t<T$ such that $\bm{\sigma}(t)\notin\mathcal{U}$.
\end{claim}

\begin{proof} First, consider the case in which $\theta(\sigma)\in(\theta^{\ast},\theta^{\ast}+\frac{\varepsilon}{2})$.
Pick an arbitrary path $\bm{\sigma}\in S_{\sigma}^{\infty}$. Suppose
that $\theta(\bm{\sigma}(t))\in(\theta^{\ast},\theta^{\ast}+\frac{\varepsilon}{2})$
in some period $t$. Then since $\theta^{\ast}$ is repelling, $\dot{\bm{\sigma}}(t)=\sigma-\bm{\sigma}(t)$
for some $\sigma$ such that $\theta(\sigma)\geq\theta^{\ast}+\varepsilon$.
Hence 
\[
\lambda_{1}\cdot\dot{\bm{\sigma}}(t)=\lambda_{1}\cdot(\sigma-\bm{\sigma}(t))\geq k_{1}-\lambda_{1}\cdot\bm{\sigma}(t)>0
\]
where the weak inequality follows from $\theta(\sigma)\geq\theta^{\ast}+\varepsilon$,
and the strict inequality follows from $\theta(\bm{\sigma}(t))<\theta^{\ast}+\frac{\varepsilon}{2}$.
This implies that $\lambda_{1}\cdot\dot{\bm{\sigma}}(t)$ increases
whenever $\theta(\bm{\sigma}(t))\in(\theta^{\ast},\theta^{\ast}+\frac{\varepsilon}{2})$
in the current period $t$. Hence there is $t<T$ such that $\lambda_{1}\cdot\theta(\bm{\sigma}(t))\geq\theta^{\ast}+\frac{\varepsilon}{2}$,
implying $\bm{\sigma}(t)\notin\mathcal{U}$. A similar argument applies
to the case in which $\theta(\sigma)\in(\theta^{\ast}-\frac{\varepsilon}{2},\theta^{\ast})$.
\end{proof}

Now we will show that $\mathcal{U}$ and $T$ satisfy the property
stated in the definition of repelling equilibria. Pick an arbitrary
$\sigma\in\mathcal{U}$. There are two cases to be considered.

$\textit{Case 1}$: $\theta(\sigma)\neq\theta^{\ast}$. Pick any action
$x$. For $\beta$ close to one, a perturbed mixture $\beta\sigma+(1-\beta)\delta_{x}$
is still in the set $\mathcal{U}$, and $\theta(\beta\sigma+(1-\beta)\delta_{x})\neq\theta^{\ast}$.
Hence from the claim above, starting from this perturbed mixture $\beta\sigma+(1-\beta)\delta_{x}$,
any solution to the differential inclusion must leave the set $\mathcal{U}$
within time $T$. So this $\sigma$ satisfies the property stated
in the definition of repelling equilibria.

$\textit{Case 2}$: $\theta(\sigma)=\theta^{\ast}$. Pick an arbitrary
pure action $x\in F(\delta_{\theta^{\ast}})$. Since $\theta^{\ast}$
is repelling, $\theta(\delta_{x})\neq\theta^{\ast}$. So from Proposition
\ref{propmonotone}(iii), for any $\beta$ sufficiently close to one,
$\theta(\beta\sigma+(1-\beta)\delta_{x})\in(\theta^{\ast}-\frac{\varepsilon}{2},\theta^{\ast})\cup(\theta^{\ast},\theta^{\ast}+\frac{\varepsilon}{2})$.
Hence, from the above claim, starting from this perturbed mixture
$\beta\sigma+(1-\beta)\delta_{x}$, any solution to the differential
inclusion must leave the set $\mathcal{U}$ within time $T$. So this
$\sigma$ satisfies the property stated in the definition of repelling
equilibria.

\textbf{If:} Let $\theta^{\ast}\in(0,1)$ be such that $\theta(\delta_{x})\neq\theta^{\ast}$
for each pure action $x\in F(\delta_{\theta^{\ast}})$, there is at
least one mixed equilibrium $\sigma^{\ast}$ with $\theta(\sigma^{\ast})=\theta^{\ast}$,
and all such mixed equilibria are repelling. We will show that the
model $\theta^{\ast}$ is repelling.

Pick an arbitrary repelling equilibrium $\sigma^{\ast}$ with $\theta(\sigma^{\ast})=\theta^{\ast}$,
and let $T$ and $\mathcal{U}$ be as in the definition of repelling
equilibria. Then each point $\sigma\in\mathcal{U}$ satisfies property
(i) or (ii) in the definition of repelling equilibria. In particular,
$\sigma=\sigma^{\ast}$ satisfies property (ii), i.e., starting from
a perturbed action frequency $\beta\sigma^{\ast}+(1-\beta)\delta_{x}$,
any solution to the differential inclusion must leave the neighborhood
$\mathcal{U}$ of $\sigma^{\ast}$ within time $T$. This is so because
$\sigma^{\ast}$ is an equilibrium and never satisfies property (i).

As the following lemma shows, this property implies that $\theta^{\ast}$
is indeed repelling.
\begin{lem}
$\theta^{\ast}$ is repelling.
\end{lem}
\begin{proof}
We prove by contradiction, so suppose that $\theta^{\ast}$ is not
repelling. Then from the upper hemi-continuity of $F$, there is $\varepsilon>0$
and a pure action $x\in F(\delta_{\theta^{\ast}})$ such that

(a)~~~~$\theta(\delta_{x})>\theta^{\ast}$ and $x\in F(\delta_{\theta})$
for all $\theta\in(\theta^{\ast}-\varepsilon,\theta^{\ast})$, or 

(b)~~~$\theta(\delta_{x})<\theta^{\ast}$ and $x\in F(\delta_{\theta})$
for all $\theta\in(\theta^{\ast},\theta^{\ast}+\varepsilon)$. 

Pick such $\varepsilon$ and $x$. In what follows, we focus on the
case in which this action $x$ satisfies property (a). (The proof
for the other case is symmetric, and hence omitted.)

Since $\sigma^{\ast}$ is a mixed equilibrium with $\theta(\sigma^{\ast})=\theta^{\ast}$
and since there is no pure action $x$ with $\theta(\delta_{x})=\theta^{\ast}$,
there must be two actions $x$ and $\tilde{x}$ such that $x,\tilde{x}\in F(\delta_{\theta^{\ast}})$
and $\theta(\delta_{x})<\theta^{\ast}<\theta(\delta_{\tilde{x}})$.
Pick such $x$ and $\tilde{x}$. Pick $\beta\in(0,1)$ close to one
so that $\theta(\beta\sigma^{\ast}+(1-\beta)\delta_{x})\in(\theta^{\ast}-\varepsilon,\theta^{\ast})$.
Then consider the following path $\bm{\sigma}$ which starts from
$\beta\sigma^{\ast}+(1-\beta)\delta_{x}$: 
\begin{align*}
\dot{\bm{\sigma}}(t)=\left\{ \begin{array}{ll}
\delta_{x^{\ast}}-\bm{\sigma}(t) & \textrm{if }\theta(\bm{\sigma}(t))<\theta^{\ast}\\
\sigma^{\ast}-\bm{\sigma}(t) & \textrm{if }\theta(\bm{\sigma}(t))=\theta^{\ast}
\end{array}\right..
\end{align*}
In words, on this path, the share of $x^{\ast}$ increases until $\theta(\bm{\sigma}(t))$
hits $\theta^{\ast}$, and after that $\bm{\sigma}(t)$ moves toward
the equilibrium $\sigma^{\ast}$. Clearly this path solves the differential
inclusion with the initial value $\beta\sigma^{\ast}+(1-\beta)\delta_{x}$,
and in particular, if $\beta$ is sufficiently close to one, this
path never leaves the neighborhood $\mathcal{U}$ of $\sigma^{\ast}$.
This implies that $\sigma=\sigma^{\ast}$ does not satisfy property
(ii) in the definition of repelling equilibria, which is a contradiction.
\end{proof}

\subsection{\label{pf:prop:app_monotone}Proof of Proposition \ref{prop:app_monotone}}

We say a correspondence $B:[0,1]\rightarrow\mathbb{R}$ has the \textbf{staircase
property} if there exists a $K<\infty$, $0=a_{0}<a_{1}<...<a_{K}=1$
in $[0,1]$ and $(A_{i})_{i=0}^{K}$ such that $A_{i}\subseteq\mathbb{R}$
closed, bounded and convex such that for each $i\in\{0,...,K\}$,
(i) $B(a_{i})=A_{i}$ and for each $\theta\in(a_{i},a_{i+1})$, $B(\theta)=\bar{x}_{i}\equiv\max\{x\colon x\in A_{i}\}$,
(ii) $\bar{x}_{i}=\underline{x}_{i+1}\equiv\min\{x\colon x\in A_{i+1}\}$.

\begin{claim} The correspondence $\theta\mapsto B(\theta)\equiv\theta(\Delta F(\delta_{\theta}))$
has the staircase property.\end{claim}

\begin{proof}We first show that if the mapping $\theta\mapsto G(\theta)\equiv F(\delta_{\theta})$
is uhc and satisfies $\max F(\delta_{\theta})\leq\min F(\delta_{\theta'})$
for all $\theta'>\theta$, then it follows that there exists a $K<\infty$,
$0=a_{0}<a_{1}<...<a_{K}=1$ in $[0,1]$ and $(A_{i})_{i=0}^{K}$
such that $A_{i}\subseteq X$ finite such that for each $i\in\{0,...,K\}$,
(i) $G(a_{i})=A_{i}$ and for each $\theta\in(a_{i},a_{i+1})$, $G(\theta)=\bar{x}_{i}\equiv\max\{x\colon x\in A_{i}\}$,
(ii) $\bar{x}_{i}=\underline{x}_{i+1}\equiv\min\{x\colon x\in A_{i+1}\}$.
Second, we will show that $B$ has the staircase property with $K$,
$0=a_{0}<a_{1}<...<a_{K}=1$ in $[0,1]$ and $(B_{i})_{i=0}^{K}$
where $B_{i}\equiv\theta(\Delta A_{i})=[\theta(\delta_{\underline{x}_{i}}),\theta(\delta_{\text{\ensuremath{\underline{x}}}_{i+1}})]$.

For the first part note that if $G$ is constant over a non-trivial
interval, it must be single valued. If not, there exists a $\theta'>\theta$
such $\max F(\delta_{\theta})>\min F(\delta_{\theta'})$. Moreover
by the order condition and the fact that $X$ is finite, it follows
that the set of points in $[0,1]$ wherein $G$ is multi-valued is
finite, we denote it as $0=a_{0}<a_{1}<...<a_{K}=1$. For each $i\in\{1,..,K\}$,
let $A_{i}\equiv G(a_{i})$. As $\max G(\theta)\leq\min G(\theta)$
for all $\theta'>\theta$, it follows that for each $\theta\in(a_{i},a_{i+1})$,
$G(\theta)\geq\bar{x}_{i}\equiv\max\{x\colon x\in A_{i}\}$ and $\bar{x}_{i}\leq\underline{x}_{i+1}\equiv\min\{x\colon x\in A_{i+1}\}$.
By the fact that $G$ is uhc and $X$ is finite, it must hold that
$G(\theta)=\bar{x}_{i}$, otherwise there exists a sequence a sequence
$(\theta_{n},x)_{n}$ converging to $(a_{i},x)$ such that $x=G(\theta_{n})$
but $x>\max\{x\colon x\in G(a_{i})\}$. A similar argument shows that
$\bar{x}_{i}=\underline{x}_{i+1}$.

We now show the correspondence $\theta\mapsto B(\theta)\equiv\theta(\Delta F(\delta_{\theta}))$
has the staircase property. Take any $i\in\{0,...,K\}$, for any $b\in(a_{i},a_{i+1})$,
$B(b)=\theta_{i}\equiv\theta(\delta_{\bar{x}_{i}})$ and $B(a_{i})=\theta(\Delta F(\delta_{a_{i}}))=\theta(\Delta A_{i})$;
this establishes condition (i) of the definition. We will show now
that $\theta(\Delta A_{i})=[\theta(\delta_{\underline{x}_{i}}),\theta(\delta_{\bar{x}_{i+1}})]$
which will finish the proof. By finiteness of $X$, $A_{i}=\{x(1),...,x(L)\}$
such that $x(1)<x(2)<...<x(L)$ (where the indexes can depend on $i$),
and as $x\mapsto\theta(\delta_{x})$ is increasing, $\theta(\delta_{x(1)})<...<\theta(\delta_{x(L)})$. 

We first show that $\theta(\Delta A_{i})\subseteq[\theta(\delta_{\underline{x}_{i}}),\theta(\delta_{\bar{x}_{i+1}})]$.
For this, we first show that
\[
\theta(\delta_{\underline{x}_{i}})=\theta(\delta_{x(1)})=\min\left\{ \theta\colon\theta\in\theta(\Delta A_{i})\right\} .
\]
We do this iteratively. If $L=2$, then $\Delta A_{i}=\{\sigma_{\lambda}(x(1),x(2))\equiv\lambda_{1}\delta_{x(1)}+\lambda_{2}\delta_{x(2)}\colon\lambda\in\Delta^{1}\}$
(throughout, $\sigma_{\lambda}(x(i),..,x(m))$ denotes a mixed action
over actions $x(i),...,x(m)$ with weights $\lambda$). By Proposition
\ref{propmonotone}, $\lambda\mapsto\theta(\sigma_{\lambda}(x(1),x(2)))$
is non-decreasing, there 
\[
\theta(\delta_{x(1)})=\min\{\theta\colon\theta\in\theta(\Delta A_{i})\}\leq\theta(\delta_{x(2)})=\max\{\theta\colon\theta\in\theta(\Delta A_{i})\}.
\]
If $L=3$, we want to show that $\theta(\delta_{x(1)})\leq\theta(\sigma_{\lambda}(x(1),x(2),x(3)))$
for any $\lambda\in\Delta^{2}$. If $\lambda_{1}=0$, then it suffices
to show that $\theta(\delta_{x(1)})\leq\theta(\sigma_{\lambda}(x(2),x(3)))$
for any $\lambda\in\Delta^{1}$. By our calculations for the case
$L=2$, it follows that $\theta(\delta_{x(2)})\leq\theta(\sigma_{\lambda}(x(2),x(3)))$
and since $\theta(\delta_{x(1)})\leq\theta(\delta_{x(2)})$ the desired
result follows. If $\lambda_{1}>0$, then
\[
\sigma_{\lambda}(x(1),x(2),x(3))=\lambda_{1}\delta_{x(1)}+(1-\lambda_{1})\left(\frac{\lambda_{2}}{1-\lambda_{1}}\delta_{x(2)}+\frac{\lambda_{3}}{1-\lambda_{1}}\delta_{x(3)}\right)=\lambda_{1}\delta_{x(1)}+(1-\lambda_{1})\sigma_{\lambda'}(x(2),x(3))
\]
 for a $\lambda'$ that is a function of $\lambda$. So, by Proposition
\ref{propmonotone}, it suffices to show that $\theta(\delta_{x(1)})\leq\theta(\sigma_{\lambda}(x(2),x(3)))$
for any $\lambda\in\Delta^{1}$, which we already showed. Iterating
in this fashion the result can be showed for $L\geq1$. 

The proof that $\theta(\delta_{\bar{x}_{i}})=\theta(\delta_{x(L)})=\max\left\{ \theta\colon\theta\in\theta(\Delta A_{i})\right\} $
is analogous and thus omitted. 

Finally, since $\bar{x}_{i}=\underline{x}_{i+1}$, it follows that
$\theta(\delta_{\bar{x}_{i}})=\theta(\delta_{\underline{x}_{i+1}})$.
Thus, we have shown that $\theta(\Delta A_{i})\subseteq[\theta(\delta_{\underline{x}_{i}}),\theta(\delta_{\underline{x}_{i+1}})]$
as desired.

We now show that $\theta(\Delta A_{i})\supseteq[\theta(\delta_{\underline{x}_{i}}),\theta(\delta_{\underline{x}_{i+1}})]$.
Suppose not, that is, there exists a $\theta\in[\theta(\delta_{\underline{x}_{i}}),\theta(\delta_{\underline{x}_{i+1}})]\setminus\theta(\Delta A_{i})$.
There exists a $l\in\{1,...,L\}$ such that $\theta(\delta_{x(l)})\leq\theta\leq\theta(\delta_{x(l+1)})$.
By Lemma \ref{Lemma:Theta(sigma)}, $\lambda\mapsto\theta(\sigma_{\lambda}(x(l),x(l+1))$
is continuous. By Bolzano, there exists a $\lambda^{\ast}$ such that
$\theta=\theta(\sigma_{\lambda^{\ast}}(x(l),x(l+1))$. But $\sigma_{\lambda^{\ast}}(x(l),x(l+1))\in\Delta A_{i}$
and we thus arrived to a contradiction.\end{proof}

Let $\Theta^{\ast}$ be the set of equilibrium models. It is easy
to check that $\Theta^{\ast}$ coincides with the set of fixed points
of $B$. Also, this set is nonempty: we can take a nondecreasing selection
from the correspondence $B$ and, by Tarski's fixed point theorem,
it must have at least one fixed point; since it is a selection of
$B$, then $B$ must have at least one fixed point. Moreover, as $B$
has the staircase property, the set of fixed points $\Theta^{*}$
is finite and its elements can only be in the following form:

1. An end-point of a vertical segment of $B$.

2. An interior point of a vertical segment of $B$.

3. A point in the interior of a horizontal segment of $B$.

4. $\theta=0$ when there exists a $a>0$ such that $B(\theta)=0$
for all $\theta\in[0,a]$, or $\theta=1$ when there exists a $a>0$
such that $B(\theta)=1$ for all $\theta\in[1-a,1]$.

By Proposition \ref{proponedimensional}, $(\theta(\sigma_{t}))_{t}$
converges almost surely to one element in $\Theta^{\ast}$. We now
show that it has to converge to an element characterized by cases
3 or 4.

Case 1 is ruled out because at an end-point of a vertical segment
$\theta'$, $F(\delta_{\theta'})$ contains more than one action,
and so the corresponding pure action equilibrium is not strict, a
situation that was ruled out by assumption.

An equilibrium model characterized by case 2 is repelling by Definition
\ref{def:repellingmodel}. By Proposition \ref{propequivalence2},
the mixed action associated to this equilibrium model is repelling.
Thus, by Proposition \ref{Prop:Repelling}, the action frequency converges
to such equilibrium model with probability zero. Therefore, convergence
occurs to an equilibrium model characterized in case 3 or 4.

Equilibrium models characterized in cases 3 and 4 are attracting by
Definition \ref{Prop:Attracting} and are supported by a pure action.
To conclude, let $\theta^{*}$ be the point to which the sequence
$(\theta(\sigma_{t}))_{t}$ converges to, and denote the associated
pure strategy by $x^{*}$. The fact that $\theta(.)$ is increasing
implies that $(\sigma_{t})_{t}$ converges to $x^{\ast}$. Finally,
we observe that using Theorem \ref{Theo:Berk} it can be shown that
$(\sigma_{t})_{t}$ converging to $x^{\ast}$ implies that the beliefs
$(\mu_{t})_{t}$ converges to $\delta_{\theta^{\ast}}$ where $\theta^{\ast}=\theta(\delta_{x^{\ast}})$.
As $F(\delta_{\theta^{\ast}})$ is a pure action, upper hemi-continuity
of $F$ implies that $(F(\mu_{t}))_{t}$ also converges to $x^{\ast}$.

\subsection{\label{pf:Prop:BerkNash}Proof of Proposition \ref{Prop:BerkNash}}

$\Delta\cup_{\mu\in\Delta\Theta(\sigma)}F_{\beta}(\mu)\subseteq\Delta\cup_{\mu\in\Delta\Theta(\sigma)}F_{0}(\mu)$:
Let $\sigma\in\Delta\cup_{\mu\in\Delta\Theta(\sigma)}F_{\beta}(\mu)$.
Fix any $x$ such that $\sigma(x)>0$. Since $\sigma\in\Delta\cup_{\mu\in\Delta\Theta(\sigma)}F_{\beta}(\mu)$,
there exists $\mu_{x}\in\Delta\Theta(\sigma)$ such that $x\in F_{\beta}(\mu_{x})$.
It suffices to show that $x\in F_{0}(\mu_{x})$. Since $x\in F_{\beta}(\mu_{x})$,
for any $x'\in X$,
\begin{align*}
\int(\pi(x,y)+\beta V(B(x,y,\mu_{x}))\bar{Q}_{\mu_{x}}(dy\mid x) & =\int(\pi(x,y)\bar{Q}_{\mu_{x}}(dy\mid x)+\beta V(\mu_{x})\\
 & \geq\int\left(\pi(x',y)+\beta V(B(x',y,\mu_{x}))\right)\bar{Q}_{\mu_{x}}(dy\mid x')\\
 & \geq\int(\pi(x',y)\bar{Q}_{\mu_{x}}(dy\mid x')+\beta V(\mu_{x}),
\end{align*}
where the first line follows from weak identification (which implies
$B(x,y,\mu_{x})=\mu_{x}$ for all $y$ in the support of $\bar{Q}_{\mu_{x}}(\cdot\mid x)$),
the second line follows from $x\in F_{\beta}(\mu_{x})$, and the third
line follows from the convexity of the value function and the martingale
property of Bayesian updating (which imply, using Jensen's inequality,
$\int V(B(x',y,\mu_{x}))\bar{Q}_{\mu_{x}}(dy\mid x')\geq V(\int B(x',y,\mu_{x})\bar{Q}_{\mu_{x}}(dy\mid x'))=V(\mu_{x})$).
Therefore, $x$ is myopically the best action, i.e., $x\in F_{0}(\mu_{x})$.

$\Delta\cup_{\mu\in\Delta\Theta(\sigma)}F_{0}(\mu)=\cup_{\mu\in\Delta\Theta(\sigma)}\Delta F_{0}(\mu)$:
The direction $\supseteq$ holds trivially, so we only establish $\subseteq$.
Let $\sigma\in\Delta\cup_{\mu\in\Delta\Theta(\sigma)}F_{0}(\mu)$.
Fix any $x,x'$ such that $\sigma(x)>0$. Since $\sigma\in\Delta\cup_{\mu\in\Delta\Theta(\sigma)}F_{\beta}(\mu)$,
there exist $\mu_{x},\mu_{x'}\in\Delta\Theta(\sigma)$ such that $x\in F_{0}(\mu_{x})$
and $x'\in F_{0}(\mu_{x'})$. By weak identification and the fact
that $\mu_{x}$ and $\mu_{x'}$ both belong to $\Delta\Theta(\sigma)$,
$\bar{Q}_{\mu_{x}}(\cdot\mid\tilde{x})=\bar{Q}_{\mu_{x'}}(\cdot\mid\tilde{x})$
for all $\tilde{x}$ in the support of $\sigma$. Therefore, for any
$x''\in X$,
\[
\int\pi(x',y)\bar{Q}_{\mu_{x}}(dy|x')=\int\pi(x,y')\bar{Q}_{\mu_{x}}(dy|x)\geq\int\pi(x'',y)\bar{Q}_{\mu_{x}}(dy|x''),
\]
and so $x'\in F_{0}(\mu_{x})$. Since $x'$ is an arbitrary element
in the support of $\sigma$, we have shown that there is a common
belief $\mu_{x}$ under which any action in the support of $\sigma$
is optimal.
\end{document}